\documentclass[a4paper,12pt]{article}

\usepackage[utf8]{inputenc}
\usepackage[dvipsnames]{xcolor}
\usepackage{booktabs}
\usepackage{graphicx}
\usepackage{verbatim}
\usepackage{longtable}
\usepackage{tabularx}
\usepackage{multicol}
\usepackage[english]{babel}
\usepackage{array}
\usepackage{graphicx}
\usepackage{caption}
\usepackage{subcaption}
\usepackage{float}
\usepackage{tikz}
\usepackage[font=itshape]{quoting}
\usepackage{a4}
\usepackage{fancyhdr}
\usepackage{fontenc}
\usepackage{hyperref}
\usepackage{amssymb,amsfonts}
\usepackage{lscape}
\usepackage{amsmath}
\usepackage{mathtools}
\usepackage{amsthm}
\usepackage{threeparttable}
\usepackage{rotating}
\usepackage{longtable}
\usepackage{epstopdf}
\usepackage{enumerate}
\usepackage[page,title]{appendix}
\usepackage{fullpage}
\usepackage{algorithm}
\usepackage[noend]{algpseudocode}

\newtheorem{theorem}{Theorem}
\newtheorem{lemma}{Lemma}[section]

\newtheorem{proposition}{Proposition}[section]

\theoremstyle{definition}
\newtheorem{definition}{Definition}

\def\X{\mathcal{X}}

\def\N{\mathbb{N}}
\def\R{\mathbb{R}}

\def\NN{\mathcal{N}}

\def\F{\mathcal{F}}

\renewcommand{\epsilon}{\varepsilon}

\DeclareMathOperator*{\argmax}{arg\,max}
\newcommand{\ubar}[1]{\text{\b{$#1$}}}

\title{Market Areas in General Equilibrium}
\author{Gianandrea Lanzara\thanks{IGIER, Bocconi University, \texttt{gianandrea.lanzara@unibocconi.it}}~ and Matteo Santacesaria\thanks{MaLGa Center, Department of Mathematics, University of Genova, \texttt{matteo.santacesaria@unige.it}}}
\date{}

\begin{document}
\maketitle

\begin{abstract}

This paper proposes a spatial model with a realistic geography where a continuous distribution of agents (e.g., farmers) engages in economic interactions with one location from a finite set (e.g., cities). The spatial structure of the equilibrium consists of a tessellation, i.e., a partition of space into a collection of mutually exclusive market areas. After proving the existence of a unique equilibrium, we characterize how the location of borders and, in the case with mobile labor, the set of inhabited cities change in response to economic shocks. To deal with a two-dimensional space, we draw on tools from computational geometry and from the theory of shape optimization. Finally, we provide an empirical application to illustrate the usefulness of the framework for applied work.

\end{abstract}

\vspace{1.5cm}

\vspace{.5cm}

\newpage
\section{Introduction}

This paper proposes a spatial model with a realistic geography where a continuous distribution of agents engages in economic interactions with at most one location from a finite set. The spatial structure of the equilibrium consists of a tessellation, i.e. a partition of space into a collection of mutually exclusive market areas. This equilibrium structure exhibits some novel properties that depart from the extant literature. First, the model comprises a notion of borders that may be compared with borders observed in the real world. Second, when labor is mobile, some of the locations in the finite set may fail to attract workers, and therefore remain vacant. Thus the model also provides a framework for thinking about the emergence and the location of discrete economic entities, such as firms within a neighborhood, business districts in a metropolitan area, or cities within a larger economy. Both the geography of borders and the set of inhabited locations are equilibrium outcomes that depend on the parameters of the model. 

These research topics have been the focus of a large number of theoretical studies in the fields of new economic geography \cite{krugman1993, fujita1995, fujita2002}, urban economics  \cite{fujita1982, henderson1974, henderson2005}, and political economy \cite{alesina1997, alesina2000}. These studies posit a continuous space interacting with a discrete set of locations, but with the spatial dimension confined to stylized geographies. While simple geographies are often appropriate to generate valuable theoretical insights, they also offer a weaker connection with the data and do not lend themselves easily to quantitative analysis.
 
On the other hand, a more recent literature has developed economic models with realistic geographies to assess the importance of spatial frictions for market outcomes and welfare \cite{allen2014, redding2016, redding2017}. Location and trading choices in these models are typically based on a random utility approach\footnote{Another approach in trade models is based on the Armington set up; here, the consumption good is a CES composite of location-specific varieties; as is well-known, the gravity equation for bilateral trade flows is isomorphic to the one obtained with extreme value shocks.}, whose equilibrium outcome is that, in general, each location interacts with all other locations, unless exogenous factors prevent it. While this feature renders the framework flexible and tractable enough to be confronted with the data in a wide variety of empirical contexts, it also implies that endogenous borders and vacant locations struggle to materialize in this class of models. 

In this paper, we combine endogenous market areas and realistic geographies within a tractable theoretical framework. To do so, we introduce a set of tools from the mathematical literature on Voronoi diagrams. A standard Voronoi diagram is a simple assignment rule such that each point in a set $X$ is assigned to the nearest point in a finite set of locations $S$. An additively weighted Voronoi diagram is a generalization such that each location $S$ is associated with a weight $\lambda_i \in \R$, $i=1\dots n$ that determines its relative attractiveness over and above geographic distance. Hence a point $x \in X$ is assigned to a point $s_i \in S$ if and only if
\begin{equation}\notag
d(x, s_i) - \lambda_i \leq \min_{j \neq i}\left\{ d(x, s_j) - \lambda_j\right\}.
\end{equation}
This construction has surfaced in previous work in economics and geography to describe the size and shape of market areas on the Euclidean plane.\footnote{An early example is Frank Fetter's 1924 article \cite{fetter1924} in the \textit{Quarterly Journal of Economics.} In that note, Fetter studied the properties of market areas around two locations on the Euclidean plane. A number of subsequent papers extended these results in various directions: \cite{hyson1950} introduced city-specific freight rates (see also \cite{parr1995}), \cite{boots1980} considered an arbitrary number of cities, and \cite{hanjoul1989} extended the transport cost function to depend nonlinearly on Euclidean distance. For a discussion of the history of this idea before Fetter's article, see \cite{hebert1972} and \cite{shieh1985}.} In these studies, the Voronoi weights are usually set equal to the market prices $p_i$, $i=1\dots n$ in a partial equilibrium setting. If, for instance, $X$ represents a set of sellers and $S$ represents a set of markets, then this simple model works as illustrated in Figure \ref{fig:fig1} for $n = 3$. In Figure \ref{fig:fig1a}, all markets offer the same price. In this case, sellers care only about Euclidean distance and the tessellation reduces to a standard Voronoi diagram. Figure \ref{fig:fig1b} depicts an additively weighted Voronoi tessellation with $\lambda = (1.6, 1.1, 1.2)$, going anticlockwise. The boundaries shift in an intuitive manner. Market $s_1$ offers the highest price and therefore attracts sellers from longer distances, expanding its market area at the expense of the other cities. Market $s_3$ loses territory to market $s_1$ but expands its market area in the direction of market $s_2$, which offers the lowest price. Finally, in Figure \ref{fig:fig1c}, the vector of prices (or weights) is $\lambda = (1.6, 0.1, 1.2)$ and market $s_2$ fails to attract any sellers. 

\begin{figure}
\caption{Three Examples of Voronoi Tessellations With Euclidean Distances and Three Sites}
\begin{subfigure}{0.5\textwidth}
\centering
\includegraphics[width=\textwidth]{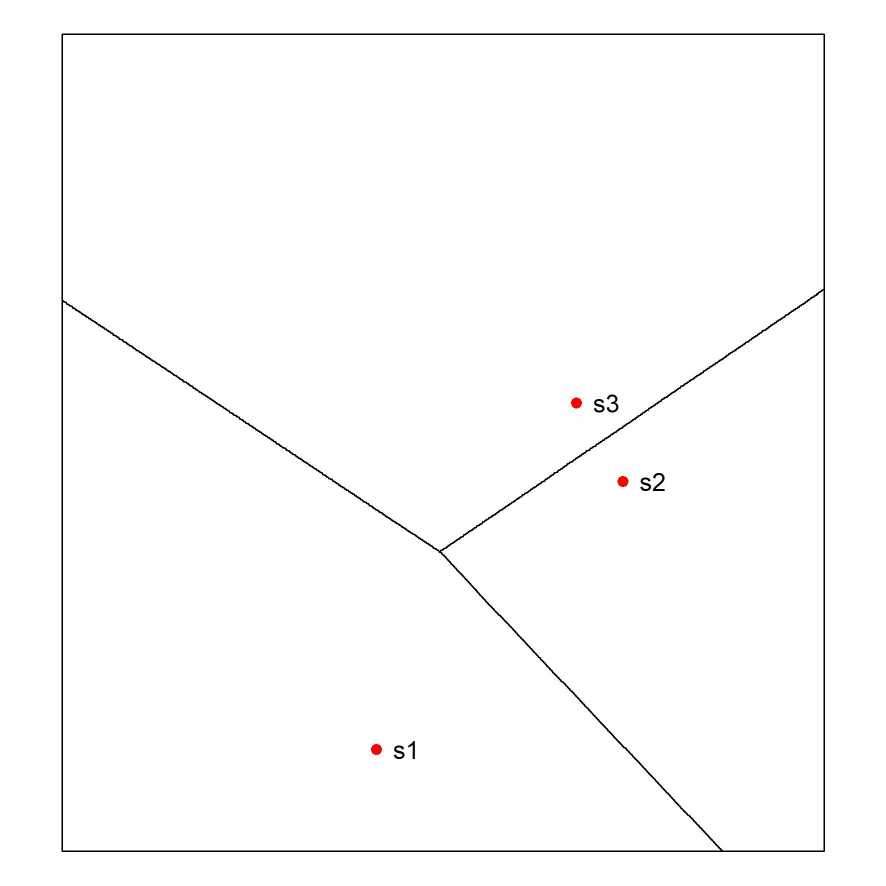}
\caption{Standard}
\label{fig:fig1a}
\end{subfigure}
\begin{subfigure}{0.5\textwidth}
\centering
\includegraphics[width= \textwidth]{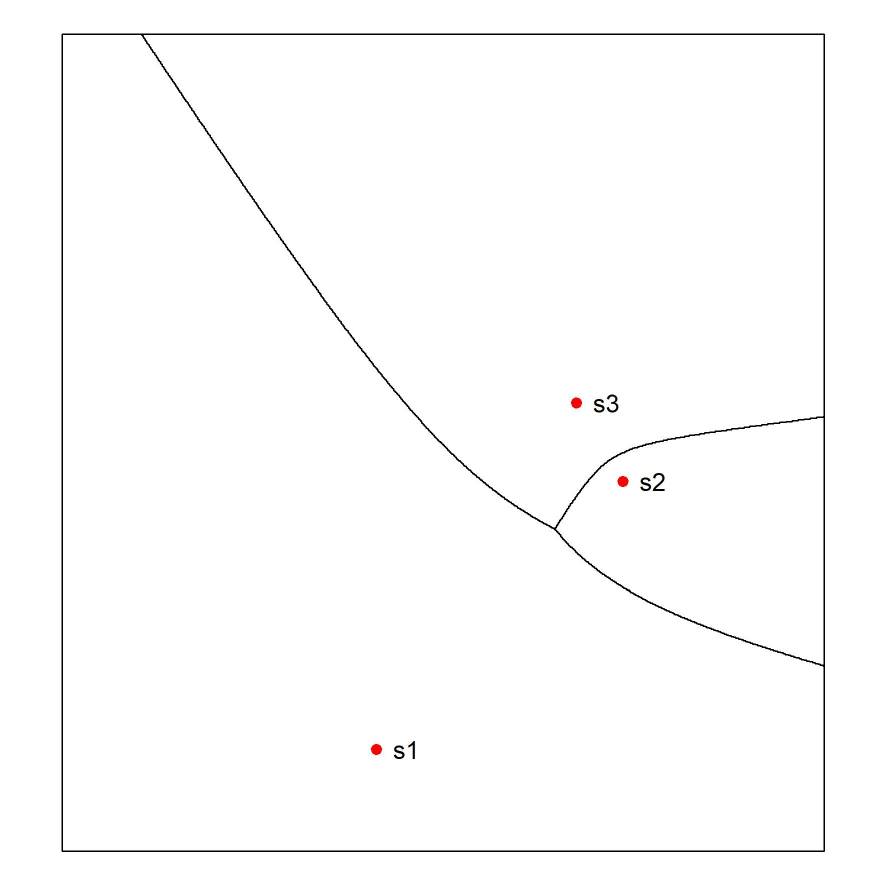}
\caption{Additively Weighted}
\label{fig:fig1b} 
\end{subfigure} 
\begin{center}
\begin{subfigure}{0.5\textwidth}
\centering
\includegraphics[width= \textwidth]{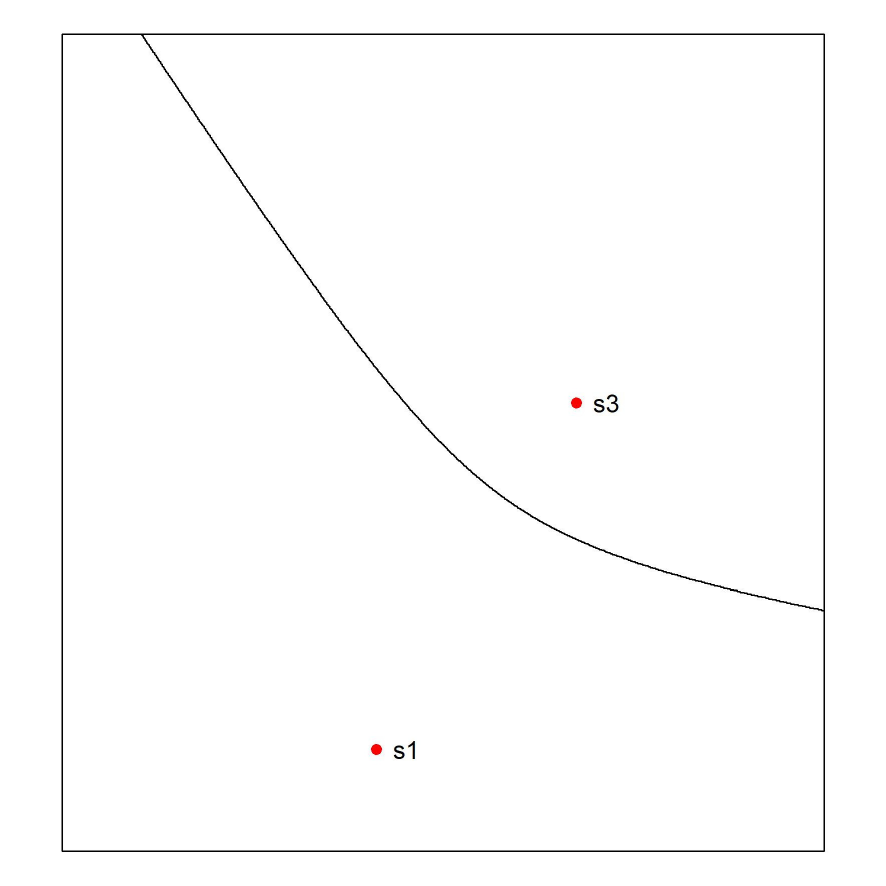}
\caption{Additively Weighted With Vacant Locations}
\label{fig:fig1c} 
\end{subfigure}
\end{center}
\parbox{15cm}{\footnotesize \emph{Notes:} In the top-left panel, all cities have the same weight and the boundaries are straight line segments. In the top-right panel, the vector of weights is not constant, giving rise to curved boundaries (hyperbolic segments), but all sites attain a market area with positive measure. In the bottom panel, the vector of weights is not constant, but the market area of site $2$ is of measure zero.}
\label{fig:fig1}
\end{figure}

The main idea of this paper is to treat the Voronoi weights as endogenous objects that are determined in general equilibrium. As a result, the problem of characterizing the equilibrium borders and the set of inhabited locations reduces to studying the properties of the equilibrium weights. Furthermore, the weights will not necessarily coincide with the market prices; they will be more complex functions derived from the economic primitives and will depend on the parameters of the model. 

We frame our discussion in terms of an urban model where the elements of $S$ are cities and the Voronoi regions in $X$ are rural areas or hinterlands. These sets are endowed with different technologies for the production of either an urban or a rural good, respectively, which enter with a constant elasticity of substitution (CES) into the consumer's utility function. Space matters, in that carrying goods from the countryside to an urban market incurs the payment of a shipping cost that increases with distance. 

We consider both a scenario with a fixed population distribution and a scenario with a mobile population across locations and sectors. In both cases, we find that a unique equilibrium exists independently of the underlying geography and of the size of shipping costs, and we characterize its comparative statics with respect to some of the model parameters. The comparative statics of the Voronoi weights provide insights into which market areas will expand, shrink, or vanish following a parameter change.  

In the setting without labor mobility, we first prove that the equilibrium conditions can be expressed as first-order conditions of a given cost function $\F$. With this result in hand, we then show that $\F$ indeed attains a unique critical point. We also show that the model delivers intuitive comparative statics as long as shipping costs increase sufficiently fast with distance. Specifically, if urban population increases in, say, city $s_i \in S$, then its market area expands at the expense of the surrounding market areas. This result relies on a graph-theoretic interpretation of Voronoi diagrams, which allows us to circumvent the fact that the gross substitution property \cite{mascolell1995} does not hold in our setting. As a matter of fact, the \textit{direct} effect of city $s_i$'s population change is nil in cities that do not share a border with $s_i$. However, because the Voronoi diagram defines a connected graph, the impact of a shock is eventually transmitted to all cities, and this weaker condition turns out to suffice to sign the comparative statics.

In a setting with labor mobility, the welfare equalization condition for urban workers yields a closed-form expression for the Voronoi weights. The main challenge here is to prove that there is a one-to-one relationship between the welfare scalar and total population, so as to ensure the uniqueness of the equilibrium in a ``closed economy'' scenario where total population is fixed. The reason is that the shape of the tessellation varies with the level of welfare. We thus have to account for these endogenous border changes and their impact on the equilibrium total population, in a two-dimensional setting with a heterogeneous geography. To overcome this challenge, we make use of a classic result in the theory of shape optimization which provides a general formula for the derivative of a function over a variable domain (a generalization of the Leibniz rule) \cite{henrot2006}. To the best of our knowledge, this is the first application of the shape derivative to an economic problem. The key observation, then, is that a parameter change affecting all Voronoi weights at the same time will exert on each border segment between neighboring market areas two opposing forces, one from each side of the border: decomposed segment-wise in this fashion, the overall effect can be handled and signed. 

Once these technical hurdles are surpassed, the model becomes highly tractable and allows for a rich set of comparative statics. Thanks to the closed-form solution for the Voronoi weights, we are able to characterize the effect of changes in shipping costs and total population on the size of market areas. The sign is unambiguous for all cities that are either more productive or less productive than all their neighbors. In particular, reductions in shipping costs increase the market area of more-productive urban centers. In contrast, the effect of changes in the size of the total population depends on the elasticity of substitution between urban and rural goods: with elastic demand, an increase in total  population will favor urban centers that are more productive than their neighbors, similar to a reduction in shipping costs.

As mentioned above, the spatial equilibrium of our framework features two novel aspects: the emergence of a well-defined notion of borders, and a set of vacant locations. Regarding this latter aspect, we use the model to derive sharp conditions for an urban site to be inhabited or vacant in equilibrium. For instance, when shipping costs are prohibitively high, all urban sites will be inhabited, and as shipping costs decrease, urban sites are ``sequentially'' abandoned until only one of them (the most productive one) remains inhabited. The conditions depend explicitly on the underlying geography: urban sites that happen to be located near more-productive ones will be abandoned earlier. While the the set of \textit{potential} urban sites is a primitive of the model, we remark that no other restriction is placed on it except that it is finite. 

We conclude the paper by applying the model to the case of Switzerland. This case study illustrates the workings of the version of the model without labor mobility, as well as the usefulness of a notion of hard borders for applied work. The empirical counterpart of the set of cities $S$ is the set of cantonal capitals: the capital cities of the federated state (cantons) forming the Swiss confederation. We compute alternative theoretical tessellations, which we evaluate against the tessellation of cantonal borders. We find that the equilibrium tessellation computed via our model is a better approximation of Swiss internal borders than alternative tessellations that neglect the roles of geography and market forces. 

\paragraph{Related literature.}

Endogenous market areas have appeared in \cite{nagy2020} and \cite{nagy2020a} to investigate, respectively, the link between trade and urbanization in Hungary after the First World War and the link between city formation and growth during the U.S. westward expansion in the 19th century. In \cite{nagy2020}, location-specific varieties of a tradable good can be exchanged at a finite set of trading locations. In the model presented below, we also take the set of urban locations as given, whereas we follow \cite{nagy2020a} (where, instead, the location technology is not predetermined) in assuming a distinction between a rural good, produced in hinterlands, and an urban good, produced in cities. 

Our contribution to this line of research is twofold: first, though we forego some important features (such as, for instance, trade costs for urban goods), we obtain a more-complete analytic characterization of the equilibrium properties of the model; second, we leverage a set of technical tools, such as Voronoi diagrams and the shape derivative, that make the problem more tractable and will be useful, we hope, in spurring further work in this area.

Our characterization of the set of inhabited cities, in the version of the model with mobile labor, contributes to the economic theories of city formation, such as the system-of-cities theories (see \cite{henderson1974} and subsequent work) and the new economic geography approach (see, for instance, \cite{fujita2002}). However, our analysis differs in two main respects. First, we work with a realistic (rather than a stylized) geographic setting. Compared to \cite{henderson1974}, for instance, urban sites are arranged in space, which allows us to derive implications not only on the \textit{number} of inhabited sites but  also on their \textit{location}. Second, the economic mechanism is different. In those traditions, economies of scale at the urban level are a key ingredient for generating the agglomeration of economic activity in a subset of the available locations, even on a featureless line or plane. 

In contrast, in our model, a necessary condition is that urban sites have heterogenous characteristics, even in the absence of urban spillovers. However, it would be incorrect to conclude that agglomeration forces are entirely missing from the model: as we explain below, if on one hand rising agricultural prices encourage urban workers to disperse, they may, on the other hand, lead farmers to serve a limited subset of urban markets. In sum, heterogeneity must be combined with endogenous trading choices among spatially ordered locations for some urban sites to remain vacant.

Our work is also related to a small number of papers that borrow tools from computational geometry to solve economic problems. For instance, \cite{rossi-hansberg2020} focus on the optimal plant-location decision of a firm given a continuous distribution of consumers in space. Because consumers patronize one plant only, the model delivers a spatial tessellation that corresponds to a weighted Voronoi diagram. In a recent working paper, \cite{allen2022} leverages centroidal Voronoi tessellation to study the evolution of national borders in a quantitative framework. Finally, Voronoi diagrams appear in \cite{merlo2016} within the context of a spatial theory of voting and in \cite{lambert2019} in the context of mechanism design.

\paragraph{Structure of the paper.} The rest of the paper proceeds as follows. In Section \ref{sub:voronoi}, we introduce some basic concepts and definitions of additively weighted Voronoi diagrams. In Section ~\ref{sec: model}, we lay out the economic framework. Sections \ref{sec: results} and \ref{sec:factor}, respectively, characterize the properties of the equilibrim with immobile labor and mobile labor. Section \ref{sec:factor} also presents our results on the endogeneous set of inhabited cities and discusses them in light of extant theories of city formation. In Section~\ref{sec:exp}, we apply the model empirically to Switzerland. Section~\ref{sec:conclusions} concludes.

\section{Additively Weighted Voronoi Diagrams}\label{sub:voronoi}

In this section, we introduce some basic concepts on additively weighted Voronoi diagrams. This geometric construct will describe the spatial structure of the economic model presented in the next section. We also present some properties that we will use repeatedly in the following analysis. For a comprehensive treatment of Voronoi diagrams and their applications, see \cite{okabe2000} and \cite{deberg2008}.

Let $X$ be an open, bounded, connected subset of the Euclidean plane $\R^2$ with Lipschitz boundary $\partial X$, let $S \subset X$ be a finite set of $n$ points, and for each $s \in S$, let
\[
d_s\colon  \R^2 \to \R_+, \qquad \R_+ = \{ x \in \R : x \geq 0\}
\]
be a continuous ``distance'' function that assigns to a point $x \in \R^2$ a nonnegative value $d_s(x)$.

 For $s\neq t \in S$ and $\gamma \in \R$, we define the \textbf{bisector}
\begin{equation}\label{eq: bisector}
B_\gamma(s,t) = \{x \in \R^2\colon d_s(x)-d_t(x) = \gamma\},
\end{equation}
and the \textbf{dominance region}
\begin{equation}\label{eq: dom}
R_\gamma(s,t) = \{x \in \R^2\colon d_s(x) - d_t(x) < \gamma \}.
\end{equation}
Note that the sets $R_\gamma(s,t)$ are open and increase with $\gamma$. Assume now we are given for each city $s \in S$ a \textbf{weight} $\lambda_s \in \R$, and let
\[
\lambda = (\lambda_1,\dots,\lambda_n)
\]
be the vector of all weights, where we have chosen an ordering of the cities $s_1,\dots,s_n \in S$. Then we say that $B_{\lambda_i-\lambda_j}(s_i,s_j)$ is the \textit{additively weighted bisector} of $s_i$ and $s_j$. The regions $R_\gamma(s,t)$ are unbounded in general, and we want to use them to create a partition of $X$. We say that
\begin{equation}\label{eq: voronoi}
\Omega_i (\lambda) = X \cap \bigcap_{j\neq i}R_{\lambda_i-\lambda_j}(s_i,s_j)
\end{equation}
is the \textit{additively weighted Voronoi region} of $s_i$ (with respect to $S$) in $X$. The \textbf{additively weighted Voronoi diagram} of $S$ in $X$ is defined as
\[
V_S (\lambda) =  \bar X \setminus \bigcup_{i=1}^n \Omega_i(\lambda).
\]
Adding the same constant to all weights does not change $\Omega_i(\lambda)$ and $V_S(\lambda)$. If $\lambda_i = \lambda$ for all $s_i \in S$, then one obtains an unweighted or standard Voronoi tessellation, such that $d_i(x) = d_j(x)$ for all points $x$ along the bisector $B_{\lambda_i  -\lambda_j}(s_i,s_j)$.

Following \cite{geiss2013}, we impose additional conditions on the system of functions $\{d_s\}_{s\in S}$ to have well-behaved Voronoi diagrams and to avoid pathological situations.

\begin{definition}\label{def:distance}
A system of continuous distance functions $d_s(\cdot)$, for $s \in S$, is called \textit{admissible} if for all $s\neq t \in S$, and for each bounded open set $C \subset \R^2$, there are two constants $m_{st}$ and $M_{st}$, with $m_{st} < M_{st}$, such that $\gamma \mapsto |C \cap R_\gamma(s,t)|$ is continuously increasing from $0$ to $|C|$, the Lebesgue measure of $C$, as $\gamma$ grows from $m_{st}$ to $M_{st}$. Moreover, $C \cap R_\gamma(s,t) = \emptyset$ if $\gamma \leq m_{st}$, and $C \subset R_\gamma(s,t)$ if $\gamma \geq M_{st}$.
\end{definition}

The next result says essentially that under this assumption, the bisectors have measure $0$.

\begin{lemma}[\cite{geiss2013}] \label{lem: bisectors} 
For a system of admissible distance functions $d_s(\cdot)$, where $s \in S$, for any two points $s\neq t \in S$ and any $\gamma \in \R$, and for any bounded open set $C \subseteq \R^2$, we have $|C\cap B_\gamma (s,t)|=0$. 
\end{lemma}

Two final observations will be important for our analysis. First, Lemma \ref{lem: bisectors} implies that, given a system of admissible distance functions and a vector of weights, the associated weighted Voronoi diagram is uniquely defined (up to an additive constant in the weight vector). Second, for $\gamma < - d_t(s)$, we have that $R_\gamma(s,t) = \emptyset$. Therefore for some values of the weights vector, some Voronoi regions may be empty.

\section{The Economic Model}\label{sec: model}

\subsection{Setup}\label{sec:setup}

We now lay out the economic model. To make the exposition more concrete, we refer to $x \in X$ as rural locations, and to $s \in S$ as urban locations or cities. An admissible distance function $d_s(x)$ represents the distance between each rural location $x \in X$ to each city $s\in S$. 

Agents, either farmers or urban workers, consume two goods: an agricultural good that is produced in the countryside, and a manufacturing good that is produced in cities. In the rest of the paper, we index goods with lowercase letters ($a$ for agricultural output and $m$ for manufacturing output), and agents with uppercase letters ($A$ for farmers and $M$ for urban workers). Farmers ship their produce to an urban market, where consumption activities also take place.\footnote{With mobile labor, all our results hold up under the alternative assumption that consumption takes place at the production location. We discuss the robustness of our results for the case with immobile labor in Appendix \ref{app:homeconsumption}. Finally, because manufacturing goods are freely traded, where they are consumed is immaterial.}

Agricultural output per capita is given by a function $y^a\colon  X  \to \R_+$, assumed to be continuous and bounded from above and below: 
\begin{equation}\label{eq: Ybounds} \notag
0< y^a_{\min} \leq y^a(x) \leq y^a_{\max}, \quad \text{for } x \in  X
\end{equation}
Manufacturing output in city $s_i \in S$ is denoted by $y_i^m$, with $y_i^m > 0$ for $i = 1,...,n$.  

We carry out our analysis under two alternative hypotheses concerning the mobility of workers: first, we consider a case where the population distribution across locations and sectors is fixed; second, we consider a case where individuals are freely mobile, so that the population distribution is determined via welfare-equalization conditions. In both cases, the rural population will be a function $L^A\colon  X  \to \R_+$ that is assumed to be continuous, positive, and bounded from below by a positive constant, while the urban population in city $s_i \in S$ will be denoted by $L^M_i>0$, for $ i = 1,...,n $.  
 
Finally, farm goods are costly to transport. Each time a farmer ships his goods from a rural location $x \in X$ to a city $s \in S$, a share  $\Delta(x,s)$ melts in transit---that is, shipping costs take the iceberg form. We also assume that shipping costs are an exponential function of distance: 
\[
\Delta(x,s) = \exp(\delta d_s(x)), \quad \delta >0
\]
In contrast, manufacturing goods are traded between cities at no cost. 

\subsection{Consumption Problem}\label{sec:consumptionproblem}

All agents in the economy order consumption baskets according to a utility function: $u\colon  \R^2_{+} \to \R$, where we denoted $\R^2_{+} = \R_{+} \times \R_{+}$. We assume the utility function takes the CES form,
\begin{equation}\notag
u(c^m, c^a) = \left( (c^m)^{\alpha} + (c^a)^{\alpha} \right)^{\frac{1}{\alpha}}, \quad \alpha < 1, \alpha \neq 0,
\end{equation}
where $1/(1-\alpha)$ is the elasticity of substitution between manufacturing and agricultural goods. The same theoretical analysis can be carried out for $\alpha \to 0$ and a Cobb-Douglas utility function. 

Let $p$ and $q$ denote, respectively, the price of the agricultural good and the manufacturing good. An agent whose income is $\omega > 0$, and who faces consumption prices $q, p > 0$, solves the following constrained concave maximization problem:
\begin{align}\label{eq: consumptionproblem}
\max_{c^m, c^a}u(c^m, c^a) \quad \text{such that} \quad q c^m + p c^a \leq \omega, \quad c^m,c^a \geq 0, 
\end{align}
yielding the unique demand functions 
\begin{align}\notag
c^m(q,p,\omega) & = \frac{p^{\frac{\alpha}{1-\alpha}}}{p^{\frac{\alpha}{1-\alpha}}+q^{\frac{\alpha}{1-\alpha}}}\frac \omega q ,\\ \notag
c^a(q,p,\omega) & = \frac{q^{\frac{\alpha}{1-\alpha}}}{p^{\frac{\alpha}{1-\alpha}}+q^{\frac{\alpha}{1-\alpha}}} \frac \omega p.
\end{align}

The indirect utility function associated with problem \eqref{eq: consumptionproblem} is
\begin{eqnarray}\label{eq: indirectutility}
V( q,p,\omega ) = v(q,p)\omega, \quad\text{with}\quad v(q,p) = \left( \frac{1}{ q^{ \frac{\alpha}{1 - \alpha }} }+ \frac{ 1}{p^{ \frac{\alpha}{1 - \alpha }}} \right)^{\frac{ 1- \alpha }{ \alpha } }.
\end{eqnarray}
The next lemma presents some properties of the indirect utility function that will be useful for characterizing the equilibrium. Since these properties are either well-known or straightforward upon inspection of \eqref{eq: indirectutility}, we omit the lemma's proof.
\begin{lemma}\label{lemma: propertiesV} The indirect utility function $V$ has the following properties:
\begin{enumerate}[$i.$]
\item $\lim_{p_i \to +\infty} V\left(q, p_i, \omega(x, s_i) \right) \to  + \infty$
\item $\lim_{p_i \to 0} V\left(q,p_i, \omega_i \right) \to  + \infty$
\item \begin{equation}\notag
\frac{\partial V(q, p_i, \omega)}{\partial \omega} = v(q, p_i), \quad i =1 \dots n
\end{equation}
\item $v(q,p) >0$ for $q,p >0$
\end{enumerate}
\end{lemma}

\subsection{Farmer's Trading Problem}\label{sec: farmerchoice}
In general, incomes and prices will depend on location, and farmers will optimally choose a trading location to maximize their welfare given the form of the indirect utility in \eqref{eq: indirectutility}. A farmer located at $x \in X$ who sells his produce at city $s_i \in S$ receives an income equal to $\omega(x, s_i ) \coloneqq p_i y^a(x) / \Delta(x, s_i)$, where $p_i ( = p_{s_i})$ is the price of the agricultural good in city $s_i$, and $q$ is the price of the manufacturing good, equalized across urban markets by standard nonarbitrage arguments. Therefore, this farmer will solve
\begin{equation}\label{def:vx}
V(x,p) =  \max_{i=1,\dots,n} V\left(p_i, q, \omega(x, s_i)\right) =  \max_{i=1,\dots,n} \left( 1 + \left( \frac{p_i}{q} \right)^{\frac{\alpha}{1-\alpha}} \right)^{\frac{1-\alpha}{\alpha}} \frac{y^a(x)}{\Delta(x, s_i)}.
\end{equation}
Thus $V(x,p)$ denotes the indirect utility function of a farmer in $x \in X$, maximized over trading locations\footnote{This formulation implicitly assumes that farmers can only choose to serve one urban market. However, in equilibrium farmers may be indifferent between serving multiple cities only on a measure zero set.}. After taking logs in equation, the trading problem in \eqref{def:vx} can be reformulated as
\begin{align}\label{eq: farmerproblem}
\min_{i = 1,\dots,n}& d_i(x) - \frac{1}{\delta}\log \hat{v}(q, p_i), \; \\ \notag \text{with}  \; \hat{v}(q, p_i) \coloneqq & p_i v(q, p_i) =\left( 1 + \left( \frac{p_i}{q} \right)^{\frac{\alpha}{1-\alpha}} \right)^{\frac{1-\alpha}{\alpha}} ,
\end{align}
and where $d_i(\cdot)$ stands for $d_{s_i}(\cdot)$. Comparing this expression with equations \eqref{eq: dom} and \eqref{eq: voronoi} in Section \ref{sub:voronoi}, it is evident that the solution to the trading problem delivers an additively weighted Voronoi tessellation, where the additive weights $\lambda_i$, $i=1,\dots,n$ are given by 
\begin{align}\notag
\lambda_i &= \frac{1}{\delta}\log \hat{v}(q, p_i) \\ \label{eq: weight}
& = \frac{1}{\delta} \frac{1-\alpha}{\alpha} \log \left( 1 + \left( \frac{p_i}{q} \right)^{\frac{\alpha}{1-\alpha}} \right)
\end{align}
The Voronoi weight summarizes the attractiveness of an urban market. In this setting, where farmers are net sellers of agricultural goods, the weight is positively related to the price $p$, implying that markets with higher prices will attract farmers from longer distances. 

Let $\Omega_i(\lambda) \subset  X$ denote the set of farmers who decide to ship their goods to city $s_i$, i.e., the city $i$'s market area, where $\lambda =(\lambda_1,\dots,\lambda_n)$ denotes the full vector of weights. The total supply of farm goods to city $s_i$ is given by
\begin{eqnarray*}
\int_{\Omega_i(\lambda)}\frac{y^a(x)L^A(x)}{\Delta(x,s_i)} dx.
\end{eqnarray*}
\subsection{Excess Demands and Equilibrium}

We define the excess demand function for farm goods in city $s_i$ as follows: 
\begin{equation}\label{eq: edA}
Z_i(p) =c^a(q, p_i, \omega_i)L^M_i+ \int_{\Omega_i(\lambda(p))}\left[ c^a\left(q,p_i,\omega(x,s_i)\right) - \frac{y^a(x)}{\Delta(x,s_i)}\right] L^A(x)dx, \quad \text{for } i =1,\dots,n,
\end{equation}
where $p$ denotes the full vector of prices in the economy: $p = (p_1,...,p_n,q) \in \R_{++}^{n+1}$, with
\[
\R_{++} = \{x \in \R :x > 0\}.
\]
and $\lambda(p) = \{\lambda_1(q,p_1), \lambda_2(q,p_2),\dots,\lambda_n(q,p_n)\}$ is defined from Equation \eqref{eq: weight}.
 
Since the manufacturing good is freely traded, there is only one market for the manufacturing good. The excess demand for manufacturing goods is
\begin{equation}\label{eq: edM}
Z_{n+1}(p) = \sum_{i =1}^n \left[\left(c^m(q, p_i, \omega_i)  -  y_i^m\right) L^M_i+ \int_{\Omega_i(\lambda(p))} c^m\left(q,p_i,\omega(x,s_i)\right)L^A(x) dx \right]
\end{equation}
where $y_i^m= y^m_{s_i}$. We conclude the description of the model with a formal definition of an equilibrium.
\begin{definition}\label{def:eq}
An \textit{equilibrium with immobile labor} in this economy is a price vector $p^{*}$ such that $Z_i(p^{*}) = 0$ for $i = 1,\dots,n+1$.
\end{definition}
Clearly, if market clearing is satisfied for the first $n$ markets, the $n+1$th market will clear too by Walras' Law.  
\subsection{Discussion}

In Appendix \ref{sec:altfor}, we relax some of the assumptions and discuss some variations of the basic framework. In Appendix \ref{sec:manuftrade}, we allow for trade costs in the manufacturing sector in the model with immobile labor. The analysis combines our results with the universal gravity framework of \cite{allen2020}. In Appendix \ref{sec:spillover}, we allow for productivity spillovers in the urban sector. To characterize the equilibrium, we adapt some of the techniques in \cite{allen2022b} to incorporate endogenous market areas. These extensions show how our framework connects with recent advances in the analysis of quantitative spatial models. 
Then, in Appendix \ref{app:CESprod} we consider an alternative formulation where the CES structure is imposed on production rather than on preferences, as in \cite{nagy2020a}. In Appendix \ref{sec:business}, we provide a reinterpretation of the model that applies to a single-city setting, in the spirit of \cite{fujita1982}; here, $X$ is interpreted as metropolitan area and $S$ as a set of business districts. Finally, in Appendix \ref{app:homeconsumption}, we consider the assumption that farmers consume agricultural goods at the production location, rather than at the trading location.

\section{Equilibrium With Immobile Labor}\label{sec: results}

\subsection{Existence and Uniqueness}\label{sec:existence}

We now characterize the equilibrium properties of the model for a fixed population distribution. The main result of this section establishes that an equilibrium price vector in the sense of Definition \ref{def:eq} exists and is unique. This result holds for all geographies and independently of the shipping cost parameter $\delta$. Our line of proof proceeds in two steps: first, we show that a price vector $p$ is an equilibrium if and only if $p$ is a critical point of a certain cost function; then, we show that this cost function attains a unique maximum. 

The chosen cost function is $\F\colon \R_{++}^{n+1} \to \R$, defined as
\begin{equation}\label{eq: costfunction}
 \F(p) = - \sum_{i=1}^n V(p_i, q, \omega_i) L^M_i -\int_{X} V(x,p)L^A(x) dx,
\end{equation}
Proposition \ref{prop: equiv}, in Appendix \ref{app:intermediate}, shows that a price vector $p \in \R_{++}^{n+1}$ is an equilibrium if and only if $\nabla \F (p) = 0$. The proof applies Roy's identity after verifying that \eqref{eq: costfunction} is indeed differentiable. 

Given Proposition \ref{prop: equiv}, we next need to show that the cost function \eqref{eq: costfunction} attains a maximum. The next theorem provides this result, and also shows the maximum to be unique.

\begin{theorem}\label{theorem: existence} 
Consider the model presented in Sections \ref{sub:voronoi} and \ref{sec: model}.  Then there exists a unique (normalized) equilibrium price vector, i.e., a vector $p^* \in \R_{++}^{n+1}$ such that $Z(p^*)=0$. 
\end{theorem}
The proof appears in Appendix~\ref{app:proofs}. 

In terms of the spatial structure of the equilibrium, Theorem \ref{theorem: existence} tells us that, for any exogenous distribution of workers and productivities over space, a unique spatial tessellation exists that guarantees that supply equals demand on all markets. 

\subsection{Comparative Statics}\label{sec:compstat}

A key feature of our approach is that the spatial tessellation is an equilibrium outcome, and as such it depends on the parameters of the model. To illustrate this, we provide a sufficient condition for the size of a city's market area to grow in response to an increase in its urban population. 

Looking at equations \eqref{eq: farmerproblem} and \eqref{eq: weight}, we observe that the market area of a city $s_i$ will grow after an increase in $L^M_i$ if and only if its Voronoi weight increases \textit{relative} to the Voronoi weights of other cities. The main result of this section is the following theorem.
\begin{theorem}\label{thm:cs}
Consider the model presented in Sections \ref{sub:voronoi} and \ref{sec: model}, and let $0 < \alpha < 1$.  Then there exists $\delta_0 >0$ such that for $\delta \geq \delta_0$ the unique normalized equilibrium point $p = (p_1,\dots, p_n,1) \in \R^{n+1}_{++}$, the vector of Voronoi weights $\{ \lambda_1(p_1), \dots \lambda_n(p_n)\}$ satisfies
\begin{equation}\label{eq:csp}
\frac{\partial \lambda_i}{\partial L^M_i} \geq \frac{\partial \lambda_j}{\partial L^M_i}, \quad 1 \leq i,j \leq n
\end{equation}
\end{theorem}
This theorem implies that, if a city's urban population grows, then its market area will expand at the expense of the market areas of the adjoining cities. Intuitively, increased urban demand drives up the price of the agricultural goods, thus inducing marginal farmers located near the border to switch trading destination. 

While intuitive, the proof of Theorem \ref{thm:cs} is challenging because the standard tools developed in \cite{mascolell1995} and based on the gross substitution property do not apply in our setting, for two reasons. First, the excess demand system in \eqref{eq: edA} and \eqref{eq: edM} does not satisfy gross substitution, because the impact of a price change in city $s_i$ is nil in cities that do not share a border with $s_i$. Second, as mentioned above, the result requires us to characterize relative, rather than absolute, changes in the Voronoi weights. 

To overcome these challenges, Appendix \ref{app:intermediatecs} provides a series of intermediate results that are of independent interest. First, Proposition \ref{prop:gsz} shows that the excess demand system satisfies a property weaker than gross substitution, whereby the direct impact of a price change in city $s_i$ is positive only in cities that share a border with city $s_j$.  Second, Proposition \ref{prop:gs1} shows that this property is enough to determined the sign of the derivatives of the endogenous variables of the model with respect to the exogenous parameters. The proof relies on a key graph-theoretic interpretation of Proposition \ref{prop:gsz}: the matrix of price derivatives $[\partial Z_i / \partial p_j]_{i,j}$, with $i, j \leq n$, can be viewed as a directed graph where the Voronoi regions correspond to the vertices and two regions are connected by an edge if and only if they have a border in common.\footnote{This graph-theoretic interpretation is very close to the notion of \textit{connected strict substitution} introduced in \cite{berry2013}. In general,  though this property may be used to show the uniqueness of a Walrasian equilibrium, it is not enough to obtain comparative statics results, which are not present in their work.} This graph is strongly connected, i.e., there is a sequence of links connecting each pair of regions. Intuitively, this implies that a shock to one region will eventually be transmitted to all regions in the tessellation. Finally, Lemma \ref{lem:cs} uses the properties of M-matrices \cite{berman1994} to derive restrictions not only on the sign but also on the \textit{magnitude} of the partial derivatives, under the additional condition that $\delta$ is sufficiently large.

\subsection{A Gradient-Type Algorithm}\label{sec:algo}
Thanks to Theorem~\ref{theorem: existence} and Proposition~\ref{prop: equiv}, the unique normalized Walrasian equilibrium of our model is the unique global maximum of the cost function $\F$ defined in \eqref{eq: costfunction}, normalized with $q=1$. It is thus possible to find it numerically using a gradient-type algorithm, based on the following identity:
\[
\frac{\partial \F}{\partial p_i}(p)=Z_i(p_i)v(q,p_i), \quad p\in \R^{n+1}_{++},
\]
which is obtained at the end of the proof of Proposition~\ref{prop: equiv}.

In the algorithm below, $tol >0$ is the tolerance to control the size of the gradient as a stopping criterion, and $\tau_k>0$ is the step size. 
\begin{algorithm}[H]
\caption{Gradient ascent to find the maximum of $\F$}\label{algo}
Choose an initial guess $p^0 = (p^0_1,\dots,p^0_n) \in \R^n_{++}$; set $k = 0$ and iterate as follows:
\begin{algorithmic}[1]
\State Compute the weight vector $\lambda^k$ via formula \eqref{eq: weight} (with $q=1$);
\State Compute the additively weighted Voronoi regions $\{ \Omega_i(\lambda^k)\}_{i=1}^n$ via formulae \eqref{eq: dom} and \eqref{eq: voronoi};
\State Compute the vector of excess demand functions $Z(p^k) = (Z_1(p^k),\dots,Z_n(p^k))$ with formula \eqref{eq: edA};
\State Update the price vector $p_i^{k+1} = p_i^k +\tau_k Z_i(p^k)v(1,p_i^k)$, for $i=1,\dots,n$;
\State If $\max_{i = 1,\ldots,n} \| Z_i(p^k)v(1,p_i^k) \| > tol$, set $k = k+1$ and repeat.
\end{algorithmic}
\end{algorithm}

\section{Equilibrium With Mobile Labor}\label{sec:factor}

\subsection{Setup and Definition}

In this section, we allow for labor mobility across locations and sectors, so that the urban population vector $\{L^M_i\}_{i=1\dots n}$ and the (continuous) function $L^A$ for the rural population become endogenous objects to be determined jointly with the price vector. 

In contrast to the previous section, where the set of cities with $L^M_i > 0$ was (by construction) fixed, we now have to consider the possibility that urban population is zero in some urban locations. Therefore, the set $S$ should now be understood as a set of \textit{potential} urban sites, not all of which will necessarily host a city. Indeed, we will see that the set of inhabited urban sites is itself an equilibrium outcome depending on the parameters of the model. While the set of potential urban sites $S$ is a primitive of the model, no other restriction is placed on it except that it is finite. 

Assuming there are no costs to moving to another location or sector, all inhabited locations in the economy, either urban or rural, must yield the same level of welfare in equilibrium. Formally, we enrich the notion of equilibrium in Definition \ref{def:eq} with some supplemental conditions on the level of welfare. 
\begin{definition}\label{def:eqmob}
Let $\bar V>0$ be the common level of welfare in the economy and $V(x,p)$ be as defined in \eqref{def:vx}. We say that a price vector $p^* \in \R^{n+1}_{++}$, a vector $L^M = (L^M_1,\dots,L^M_n) \in \R^n_{++}$, and a function $L^A\colon  X \to \R_{++}$ are an \textit{equilibrium with mobile labor} if $p^*$ is an equilibrium price vector and the following conditions hold:
\begin{equation}\label{def:factor}
\begin{gathered}
V (p^*_i,q^*,q^*y_i^m) = \bar V \; \text{ for $i$ such that } L_i^M > 0,\\
V (p^*_i,q^*,q^*y_i^m) < \bar V \; \text{ for $i$ such that } L_i^M = 0,\\
V(x,p^*) = \bar V\; \text{ for all $x \in X$ such that } L^A(x) > 0,\\
V(x,p^*) < \bar V\; \text{ for all $x \in X$ such that } L^A(x) = 0.
\end{gathered}
\end{equation}
\end{definition}
The scalar $\bar{V}$ can be considered either exogenous or endogenous in the analysis. In the former case, one can imagine that individuals have the option to relocate to a large outer economy that offers a fixed level of utility $\bar{V}$. In the latter case, the aggregate population in the economy $\bar{L}$ is fixed and $\bar{V}$ is determined in equilibrium. 

Finally, to avoid degenerate solutions, we will introduce a congestion force into the model, in the form of decreasing marginal returns to agricultural production.\footnote{In the absence of a congestion force, farmers would cluster in the $(x,s_i) \in X \times S$  ``commuting'' pairs in $\text{argmax}_{x, s_i \in X \times S} V(p_i, \omega(x,s_i))$. Without agricultural productivity differences, they would never agree to incur a shipping cost and would choose $x, s_i \in S \times S$.  Specifically, they would cluster in the urban sites with the greatest productivity. } Specifically, we assume that output per capita in the agricultural sector is given by
\begin{equation}\label{eq:condy}
y^a(x) = \frac{a(x)}{(L^A(x))^{1-\beta}}, \quad x \in X,
\end{equation}
where $a\colon X \to \R_{++}$ is a continuous function with $0<a_{\min} \leq a(x) \leq a_{\max}$ for all $x \in X$, for some positive constants $a_{\min}, a_{\max}$, and $0<\beta<1$. As an example, this is the expression that we would get from a Cobb-Douglas production function that combines agricultural labor with a fixed input (i.e., land).

\subsection{The Voronoi Weights}

The analysis of the consumption problem in Section \ref{sec:consumptionproblem} still holds with labor mobility, and the expression for the indirect utility is still the one given in \eqref{eq: indirectutility}.The farmer's trading problem in Section \ref{sec: farmerchoice} is also unchanged, except that the welfare-equalization condition for urban workers in \eqref{def:factor} now yields an explicit expression for the Voronoi weights.  Throughout the rest of the section, we normalize $q=1$. 

For any $\bar{V}$ such that $\bar V^{\frac{\alpha}{1-\alpha}}-(y_i^m)^{\frac{\alpha}{1-\alpha}} > 0$, the unique agricultural price in any city $s_i \in S$ such that $V(1,p,y_i^m) = \bar{V}$ is
\begin{equation}\label{eq:pi_V}
p_i = \frac{y_i^m}{\left( \bar V^{\frac{\alpha}{1-\alpha}}-(y_i^m)^{\frac{\alpha}{1-\alpha}}\right)^{\frac{1-\alpha}{\alpha}}}.
\end{equation}
This is the price that any urban location in $S$ must offer to attract a positive number of workers. Of course, more productive cities will face higher consumption prices. When agricultural and manufacturing goods are imperfect substitutes, i.e., $0 < \alpha < 1$, this implies that more-productive cities devote a lower percentage of their income to agricultural goods than less-productive cities. 

Using \eqref{eq:pi_V} into the formula for the Voronoi weight in \eqref{eq: weight}, we can write the difference between the weights of any two cities $s_i, s_j \in S$ as
\begin{equation}\label{eq:diflambda}
\lambda_i- \lambda_j =\frac{1}{\delta} \frac{1-\alpha}{\alpha}\log\left(\frac{\bar V^{\frac{\alpha}{1-\alpha}}-(y_j^M)^{\frac{\alpha}{1-\alpha}}}{\bar V^{\frac{\alpha}{1-\alpha}}-(y_i^m)^{\frac{\alpha}{1-\alpha}}}  \right), \quad i,j=1,\dots,n,
\end{equation}

\subsection{Existence and Uniqueness}

Let us now define the total population, given by
\begin{equation}\label{eq:totpop}
\bar L =\sum_{i=1}^n \left(L^M_i + \int_{\Omega_i(\lambda)} L^A(x)dx \right). 
\end{equation}
This aggregate constraint governs the relationship between total population $\bar{L}$ and the welfare scalar $\bar{V}$, both in an ``open economy'' scenario with fixed $\bar{V}$, and in a ``closed economy'' scenario with fixed $\bar{L}$. 

We are able to show the existence of a unique equilibrium with labor mobility for an arbitrary population $\bar L >0$.
\begin{theorem}\label{theo:fact}
Let $0<\alpha<1$, $\delta >0$ and $\bar L >0$ the total population.

Let the assumptions of Theorem~\ref{theorem: existence} hold, with the additional conditions that each distance function $d_s:X \to \R_+$, for $s \in S$, is of class $C^2$ and $a\colon X \to \R_{++}$ is of class $C^1$. Then there exists a unique equilibrium with mobile labor.
\end{theorem}
The proof (presented in Appendix~\ref{app:proofs}) is based on two propositions. First, Proposition \ref{prop:fact} shows that, for a fixed $\bar V$, a unique equilibrium exists with labor mobility. Second, Proposition \ref{prop:fact2} provides a one-to-one correspondence between $\bar V$ and total population $\bar L$. 

The proof of the latter is technically difficult because, as Equation \eqref{eq:diflambda} makes clear, the welfare scalar enters the expression for the Voronoi weights. As a result, to study the sign of the derivative of Equation \eqref{eq:totpop} with respect to $\bar{V}$, we have to keep track of the impact of $\bar{V}$ on the shape of all Voronoi regions. Mathematically, we have to evaluate the derivative of functions of the form
\begin{equation}\notag
a \mapsto I(a) = \int_{\Omega_i(\lambda(a))}f(x,a)dx,
\end{equation} 
where $a$ is a generic parameter. Note that the parameter $a$ affects both the value of $f$ inside the domain and the boundaries of the domain of integration. On a one-dimensional geography, we can apply the Leibniz rule to obtain the derivative with respect to $a$. A two-dimensional heterogeneous geography, where the boundaries lack an analytic representation, requires a strict generalization of the Leibniz rule from the mathematical theory of shape optimization, which is introduced in Lemma \ref{lem:shapeder}. 

Besides this technical result, Proposition \ref{prop:fact2} follows from the following observation: when a parameter change simultaneously affects all Voronoi weights, each border segment between two adjoining market areas is subject to two opposing forces, one from each side of the border. By looking at each border segment in isolation and collecting the corresponding terms pairwise, we can unambiguously determine the sign of the overall effect.

\subsection{Comparative Statics}\label{sec:compstatmob}

We now turn to the comparative statics of the equilibrium with labor mobility. Since we look at small changes around the equilibrium point, this exercise takes the set of inhabited cities as given. In the next subsection, we also explore how the parameters of the model affect the set of inhabited cities. 

For ease of exposition, we restrict our attention to a closed-economy scenario where $\bar{V}$ is fixed and the size of the economy $\bar{L}$ can adjust in response to economic shocks. As Equation \eqref{eq:diflambda} makes clear, this approach allows us to abstract from any \textit{indirect} effects of economic shocks on the spatial tessellation due to changes in the value of $\bar{V}$. We are still able to characterize the comparative statics for a closed economy (see Appendix \ref{app:addrescsL}). 

Our first result characterizes the impact of the shipping-cost parameter $\delta$ on the structure of the spatial tessellation. 
\begin{lemma}[Effect of $\delta$]\label{lemma:csmob}
Let the assumptions of Theorem~\ref{theo:fact} hold.

\begin{enumerate}[$i.$]
\item Take two cities, $s_i, s_j \in S$, with $y_i^m > y^m_j$. Then,
\begin{equation}\notag
\frac{\partial (\lambda_i - \lambda_j)}{\partial \delta}  
< 0 .
\end{equation}
\item Take a city $s_i \in S$, and let $\mathcal{N}_i = \{s_j \in S, s_j \neq s_i: \partial \Omega_i \bigcap \partial \Omega_j \neq \emptyset\} $ denote the set of its neighboring cities. If $y_i^m > \max_{s_j \in \mathcal{N}_i} y^m_j$, then
\begin{equation}\notag
\frac{\partial \vert \Omega_i(\lambda)\vert}{\partial \delta } < 0.
\end{equation}
The opposite inequality holds if $y_i^m < \min_{s_j \in \mathcal{N}_i} y^m_j$.
\end{enumerate}
\end{lemma}
The first part of Lemma \ref{lemma:csmob} concerns the effect of $\delta$ on the Voronoi weights of two cities. It is saying that, as $\delta$ decreases, the more-productive city becomes relatively more attractive for farmers compared to the less-productive city. This result applies independently of whether or not the respective market areas share a border. If they do share a border, then the more-productive city expands its market area at the expense of the less-productive city. The reason is that as distance bites less on shipping costs, the price component of the farmers' trading decisions becomes more relevant. Since agricultural prices are higher in more-productive cities, these cities become relatively more attractive.  

The result also suggests that it cannot be determined in general whether the overall market area of a city will expand (or shrink), unless all its urban neighbors are less (or more) productive. This is the content of the second part of the lemma. An immediate corollary is that a reduction in size of the shipping costs will inflate the market area of the most-productive city in the economy and contract the market area of the least-productive city. 

The next lemma has a structure similar to that of Lemma \ref{lemma:csmob}, but it focuses on changes in the value of the welfare scalar $\bar{V}$. This is of interest because, thanks to Theorem \ref{theo:fact}, varying $\bar{V}$ is equivalent to varying $\bar{L}$ in the opposite direction in a closed-economy scenario. Differently than in the case of $\delta$, the effect of a larger population depends on the elasticity of substitution between urban and rural goods.  
\begin{lemma}[Effect of $\bar{V}$]\label{lemma:csmobV}
Let the assumptions of Theorem~\ref{theo:fact} hold.
\begin{enumerate}[$i.$]
\item Take two cities, $s_i, s_j \in S$, with $y_i^m > y^m_j$. If $0 < \alpha < 1$, then
\begin{equation}\notag
\frac{\partial (\lambda_i - \lambda_j)}{\partial \bar{V}}  
< 0 .
\end{equation}
If $\alpha < 0$, the opposite inequality holds. 
\item Take a city $s_i \in S$, and let $\mathcal{N}_i = \{s_j \in S, s_j \neq s_i: \partial \Omega_i \bigcap \partial \Omega_j \neq \emptyset\} $ denote the set of its neighboring cities. If $0 < \alpha < 1$, and $y_i^m > \max_{s_j \in \mathcal{N}_i} y^m_j$, then
\begin{equation}\notag
\frac{\partial \vert \Omega_i(\lambda)\vert}{\partial \bar{V} } < 0,
\end{equation}
whereas if $y_i^m < \min_{s_j \in \mathcal{N}_i} y^m_j$, then
\begin{equation}\notag
\frac{\partial \vert \Omega_i(\lambda)\vert}{\partial \bar{V} } > 0.
\end{equation}
If $\alpha < 0$, the opposite inequalities hold. 
\end{enumerate}
\end{lemma}
To fix ideas, consider the case $0 < \alpha <1$. Then the Lemma \ref{lemma:csmobV} implies that as population flows into the economy (either endogenously because $\bar{V}$ falls, or exogenously in a closed economy), more-productive cities will expand their market areas at the expense of less-productive adjoining cities. The mechanism is different than the one discussed with regard to the effect of $\delta$, because  $\bar{V}$ enters the formula for the equilibrium price \eqref{eq:pi_V} directly. Here, agricultural prices adjust to preserve welfare equalization when $\bar{V}$ changes. Under substitution, more-productive cities devote a lower share of their budget to agricultural goods, and therefore agricultural prices have to increase relatively more when $\bar{V}$ falls. In turn, those cities become relatively more attractive for farmers. 

Changes in the size of market areas due to $\delta$ or $\bar{V}$ are reflected, other things equal, in the size of the urban population. However, besides shifting the borders, a parameter change will also have additional effects coming from the interior of the Voronoi regions, which may go in the opposite direction. For this reason, the comparative statics of urban population are ambiguous in general.

\subsection{City Formation}

We now study how the set of inhabited cities depends on the parameters of the model. For simplicity, we maintain the (irrelevant) hypothesis that $d_i(s_i) = 0$. 

Let $T(\lambda) = \{s_i \in S: \Omega_i(\lambda) \neq \emptyset \}$ denote the set of inhabited cities. For all cities $s_i \in S$:
\begin{equation}\label{eq:cityinhab2}
s_i \in T(\lambda) \iff \lambda_i - \lambda_j >  - d_j(s_i) , \quad j = 1\dots n, \; j\neq i.
\end{equation} 
This characterization is the basis of our discussion. Intuitively, it is saying the following: if not even farmers located at $ x = s_i$ are willing to trade at $s_i$, because another urban location is more attractive, then no other farmer is.

The following lemma is useful for understanding how cities are activated in the model.  
\begin{lemma}\label{lemma:cityformy}
Let the weights be defined as in \eqref{eq:diflambda} and let $S^*$ be the set of cities with the highest level of urban productivity in the economy, i.e., $S^* = \{ s_i \in S: s_i \in \argmax_{s_j \in S} y^m_j\}$. 
\begin{enumerate}[$i.$]
\item Take two urban locations $s_i, s_j \in S$ with $y_i^m \geq y^m_j$; then,
\begin{equation}\notag
\lambda_i - \lambda_j > - d_j(s_i) 	
\end{equation}
is always satisfied. 
\item  $S^* \subseteq T$.
\item If $y_i^m = \bar{y}^m$ for $i = 1\dots n$, then $T = S$ and we obtain a standard Voronoi tessellation.
\end{enumerate}
\end{lemma}
Lemma \ref{lemma:cityformy}, part \textit{i.}, is saying that an urban location can never prevail over a more-productive one, because, intuitively, a farmer located at $x = s_i$ will never agree to incur a shipping cost to trade at an equally productive or less-productive location. An immediate consequence, stated in part \textit{ii.}, is that the most productive urban locations in the economy are always inhabited. Finally, part \textit{iii.} points to the fact that productivity differences across cities are a key ingredient for obtaining variations of the set $T$; we come back to this point below. 

According to Lemma \ref{lemma:cityformy}, it is still possible that farmers located at $x = s_i \in S$ decide to travel some distance to carry their goods to a \textit{more} productive urban location. The question is therefore whether the remaining locations in $S$, other than the most-productive ones, will be inhabited in equilibrium. Condition \eqref{eq:cityinhab2} tells us that this will depend on
\begin{enumerate}[$i.$]
\item the vector of urban productivities, $\{ y_i^m \}_{s_i \in S}$
\item the bilateral distances between urban locations in $S$, $\{d_i(s_j)\}_{s_j \in S}$
\item the parameters of the model $\alpha$, $\delta$ and $\bar{V}$
\end{enumerate}

Our next result focuses on the role of parameters $\delta$ for the case $0 < \alpha < 1$. 
\begin{lemma}\label{lemma:cityformdelta} For each city $s_i \in S \setminus S^*$, let $S^*_i = \{s_k \in S: y^m_k > y_i^m \}$ denote the nonempty set of cities with higher values of urban productivity than city $s_i$, $i=1,\dots,n$. Then, $\delta^*_i$ exists such that $s_i \in T$ if and only if $\delta > \delta^*_i$, with 
\begin{equation}\label{eq:cityformdelta}
\delta^*_i \coloneqq \left( \frac{1-\alpha}{\alpha} \right)\max_{s_j \in S^{*}_i} \frac{1}{d_j(s_i) }   \log\left(\frac{\bar V^{\frac{\alpha}{1-\alpha}}-(y_i^m)^{\frac{\alpha}{1-\alpha}}}{\bar V^{\frac{\alpha}{1-\alpha}}-(y_j^M)^{\frac{\alpha}{1-\alpha}}}  \right).
\end{equation}

\end{lemma}

Using \eqref{eq:cityformdelta}, the potential urban locations in $S$ can be ordered without loss of generality such that $\delta^*_i \geq \delta^*_{i+1} $ for all $i=1\dots n-1$. Therefore, starting from some $\delta > \delta^*_1 $ where all locations in $S$ are inhabited and $T = S$, a continuous reduction in the size of shipping cost will ``sequentially'' shut down cities with increasingly lower values of $\delta^*_i$, until, for $\delta < \delta^*_{n-1}$, the urban population will be concentrated in $S^*$, the set (possibly a singleton) of urban locations with the highest value of urban productivity.    

Space plays an explicit role in condition \eqref{eq:cityformdelta}. This is even more transparent if we work with the following sufficient condition for $s_i \in T$: 
\begin{equation}\notag
\delta > \left( \frac{1-\alpha}{\alpha} \right)  \frac{1}{\min_{s_j \in S^*_{i} } d_j(s_i) } \max_{s_j \in S^*_{i}} \log\left(\frac{\bar V^{\frac{\alpha}{1-\alpha}}-(y_i^m)^{\frac{\alpha}{1-\alpha}}}{\bar V^{\frac{\alpha}{1-\alpha}}-(y_j^M)^{\frac{\alpha}{1-\alpha}}}  \right).
\end{equation}
This condition depends explicitly on $\min_{s_j \in S^*} d_j(s_i)$, the distance from city $s_i$ to the closest more-productive urban location, and becomes tighter as this term becomes smaller. Urban locations situated near more-productive ones will be inhabited only for high values of the shipping-cost parameter $\delta$. 

Conditions similar to \eqref{eq:cityformdelta} can be derived with respect to other parameters too. For instance, we may ask what happens when total population $\bar{L}$ increases in the economy. Again, thanks to Theorem \ref{theo:fact}, we can equivalently answer this question by varying $\bar{V}$ in the opposite direction. 
\begin{lemma}\label{lemma:cityformV}
For each city $s_i \in S$, there exists  
\begin{equation}\label{eq:cityformV}
\bar{V}^*_i \coloneqq y_i^m \max_{j=1\dots n, j \neq i} \left[ \frac{(y^m_j / y_i^m)^{\frac{\alpha}{1-\alpha}} -\Delta(s_i,s_j)^{\frac{-\alpha}{1-\alpha}} }{1 - \Delta(s_i,s_j)^{\frac{-\alpha}{1-\alpha}} } \right]^{\frac{1-\alpha}{\alpha} } 
\end{equation}
such that, if $0 < \alpha < 1$, then $s_i \in T$ if and only if $\bar{V} > \bar{V}^*_i$; and if $\alpha < 0$, then $s_i \in T$ if and only if $\bar{V} < \bar{V}^*_i$.
\end{lemma}
The lemma implies that, when the two goods are demand substitutes ($0 < \alpha < 1$), an increase in total population will reduce the number of inhabited cities in the economy and cause the urban population to cluster in the most-productive ones. Conversely, when the goods are complements ($\alpha < 0$), an increase in total population will increase the number of inhabited cities.\footnote{For $\alpha \to 0$ (a Cobb-Douglas utility function), the Voronoi weights do not depend on $\bar{V}$; therefore, the set of inhabited cities does not depend on total population.}

\subsection{Discussion}

Major economic theories of city formation, such as the system-of-cities theories à la \cite{henderson1974} and the new economic geography framework \cite{fujita2002}, emphasize the role of economies of scale for the emergence of cities. By contrast, in our model, economic activity may cluster in a subset of the available locations in the absence of positive spillovers at the urban level. Two ingredients are crucial for obtaining this result: first, potential urban locations possess heterogeneous characteristics; second, urban locations are meaningfully ordered in space.

In the absence of heterogeneity, the welfare-equalization condition for urban workers requires that agricultural prices be equalized across cities. 
Under these circumstances, farmers always direct themselves to the nearest urban market, and all cities are able to attract farmers in equilibrium. This is the point made in Lemma \ref{lemma:cityformy}, part \textit{iii.}, above.

But heterogeneity alone is not enough. Farmers need an outside option that, keeping their rural location $x \in X$ fixed, allows them to abandon the urban market $s_i \in S$ they are currently serving. This comes from the second ingredient. When hinterlands are arranged in space, there is a lower bound to the equilibrium agricultural price that can be sustained in an urban location, because farmers have the option to trade with other cities. If the agricultural price required to equalize the utility of urban workers falls below this lower bound, the urban location fails to attract farmers and therefore remains empty. By contrast, when hinterlands are not arranged in space, this outside option is missing, because in this case the choices of $x$ and $s_i$ are bundled together. 

\begin{figure}[!htb]
\centering
\caption{Spatial Equilibrium With Mobile Labor and Two Urban Sites }
\includegraphics[width = 0.8\textwidth]{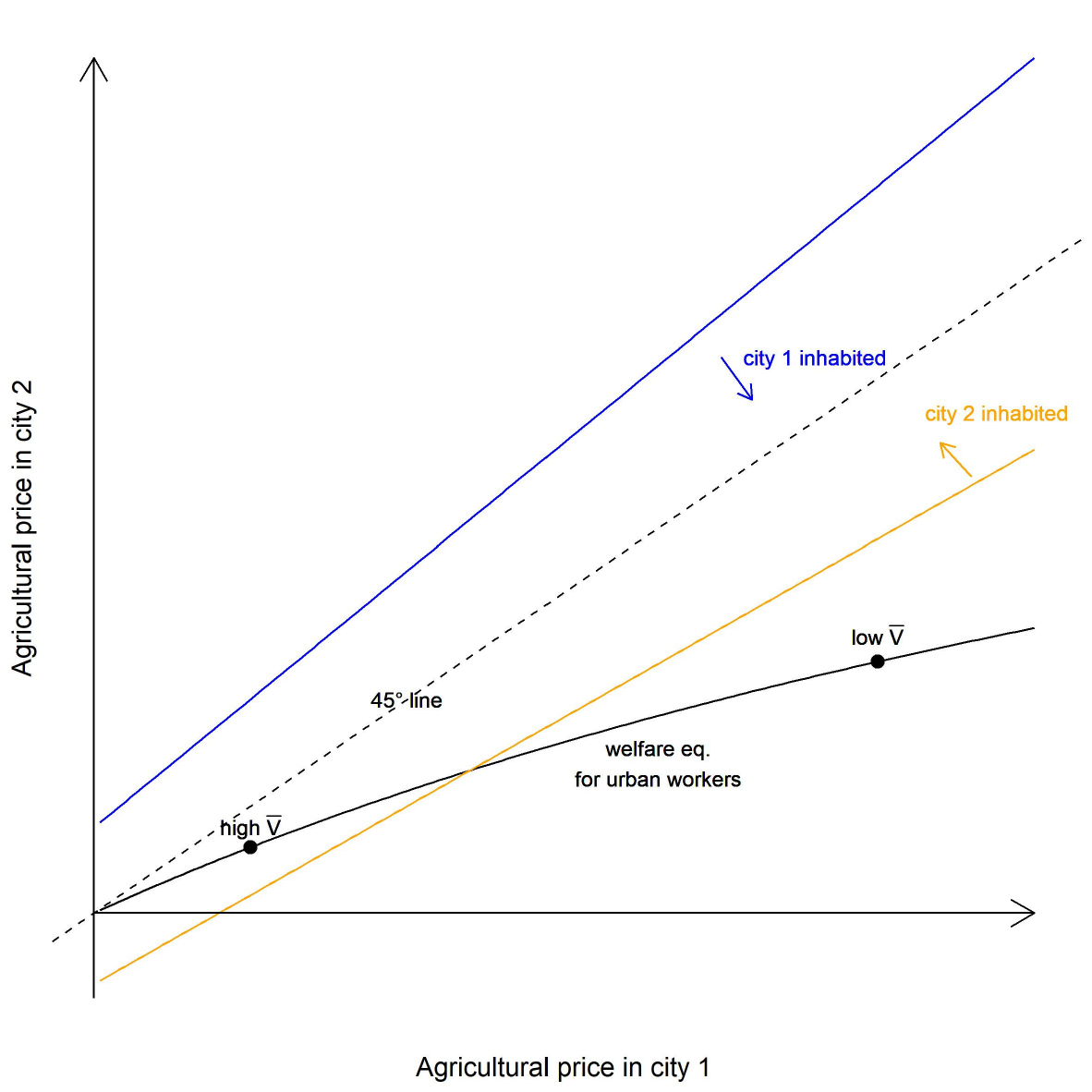}
\parbox{15cm}{\footnotesize \emph{Notes:} This figure represents the spatial equilibrium conditions of the model for the case with $0<\alpha<1$, and two urban sites, i.e., $S = \{s_1, s_2\}$, and $y^m_1 > y^m_2$. For high values of $\bar{V}$ (small population), both cities are inhabited, whereas at low values of $\bar{V}$ (large population), only city 1 is inhabited.}
\label{fig:cityformation}
\end{figure}

Figure \ref{fig:cityformation} depicts this discussion for a simple case with two urban sites, $S = \{s_1, s_2\}$, such that $\Delta(s_1, s_2) = \Delta(s_2, s_1) = \Delta > 1$ and $y^m_1 > y^m_2$. In this case, the spatial equilibrium of the model is described by three equations: 
\begin{align*}
p_2 &= \left\{ \left[1 + p_1^{\frac{\alpha - 1}{\alpha}}\right]\left(\frac{y^m_1}{y^m_2}\right)^{\frac{\alpha}{1-\alpha}} - 1 \right\}^{\frac{\alpha -1}{\alpha}}  & \text{welfare eq. for urban workers} \\
p_2 &< \left[\left(1 + p_1^{\frac{\alpha}{1-\alpha}} \right) \Delta^{\frac{\alpha}{1-\alpha}} - 1 \right]^{\frac{1-\alpha}{\alpha}} & \text{city 1 inhabited} \\
p_2 &> \left[\left(1 + p_1^{\frac{\alpha}{1-\alpha}} \right) \left(\frac{1}{\Delta}\right)^{\frac{\alpha}{1-\alpha}} - 1 \right]^{\frac{1-\alpha}{\alpha}} & \text{city 2 inhabited}
\end{align*}
The figure plots these conditions in the $(p_1, p_2)$ space, for the case $0 < \alpha < 1$. The solid black line represents the set of points where the utility of urban workers equalizes. Fixing $\bar{V}$ (or $\bar{L}$) is equivalent to selecting a point on this line, with lower values of $\bar{V}$ corresponding to points farther from the origin. The blue line traces the boundary below which city 1 is able to attract farmers. With $\delta > 0$, the slope of this line is strictly above one and its intercept strictly positive; therefore, city 1 will always be inhabited in this example. Similarly, the orange line represents the boundary above which city 2 is able to attract farmers. In the presence of trade costs, its slope is strictly below one and its intercept strictly negative; therefore, it starts below the welfare-equalization condition for urban workers, but it crosses it as $\bar{V}$ decreases and we move out from the origin. Beyond the crossing point, city 2 ceases to be viable. If $y^m_1 = y^m_2$, the welfare-equalization curve coincides with the $45$-degree line; therefore, both cities are always inhabited. When $\delta$ increases, the blue and orange lines shift and rotate outward in opposite directions, thus extending the region where both cities are viable. 

Finally, we note that the absence of economies of scale, as well as trade costs, in the urban sector accounts for the analytical ease of our approach as compared to spatial models with realistic geographies, or even simple geographies such as \cite{fujita2002} or \cite{henderson1974}. At the same time, the model retains a mechanism for the agglomeration of economic activities via the endogenous trading choices of farmers. While this formulation allows us to clarify the theoretical underpinnings of a spatial model with endogenous market areas, it also abstracts from two features that are important for theoretical and empirical work. In Sections \ref{sec:manuftrade} and Sections \ref{sec:spillover} of the Appendix, we extend to the model to incorporate these features. The model remains sufficiently tractable to characterize the uniqueness of the equilibrium, but some of the closed-form expressions are lost.

\section{Applying the Model to Swiss Cantons}\label{sec:exp}

As a final exercise, we return to the model with immobile labor (from Section \ref{sec: model}) and apply it empirically to the case of Switzerland. Our objective here is to illustrate the workings of the model and the usefulness of having a notion of borders that can be meaningfully brought to the data. To do so, we will use the model to compute alternative spatial tessellations centered on Swiss administrative capitals and compare them to Switzerland's administrative borders. 

Switzerland is a small country with a rich internal geography. Officially, it is a confederation of 26 cantons, each canton having its own capital city. Though most of the population lives in urban areas, the city-size distribution is relatively compressed, with many small and medium-sized cities.

In the rest of this section, we first discuss two metrics that can be used to compute the distance between two tessellations (and thus have a notion of goodness of fit); then we describe the construction of the spatial grid used to discretize the Swiss geography and the empirical counterparts of our theoretical variables; finally, we discuss our results. Appendix \ref{app:data} contains more details on the data. 

\subsection{Comparing Tessellations}\label{sub:comptes}

Let $\Omega = \{\Omega_k\}_{k=1}^n$ and $\Omega' = \{\Omega'_k\}_{k=1}^n$ be two different partitions (not necessarily Voronoi tessellations) of a set $X$. We will use two different metrics to compare these partitions.
\begin{itemize}
\item The Hausdorff distance. For this, we consider the \textit{skeletons} $\partial \Omega$, $\partial \Omega'$ of the partitions $\Omega$, $\Omega'$, defined as 
\[
\partial \Omega =\{ x \in X : x \in \partial \Omega_k, \text{ for some } k =1,\dots,n\}
\]
and similarly for $\partial \Omega'$. We then define the Hausdorff distance between $\Omega$ and $\Omega'$ as
\begin{equation}
d_H (\Omega,\Omega') = \max_{x \in \partial \Omega'} \min_{y \in \partial \Omega}\|x-y\|,
\end{equation}
where $\| \cdot \|$ is the Euclidean distance.
\item The Area distance. Since we consider only partitions with the same number of subdomains, we can define the following distance:
\begin{equation}
d_A (\Omega,\Omega') =\sum_{j=1}^n \min_{k=1,\dots,n} \left| \Omega_j \triangle \Omega_k' \right|,
\end{equation}
where $A \triangle B = (A \cup B) \setminus (A \cap B)$ is the symmetric difference of $A$ and $B$, and $|\cdot |$ is the Lebesgue measure.
\end{itemize}
\subsection{Data}

We approximate the set $X$ with a $200 \times 200$ rectangular grid that fully encloses Switzerland. Each grid cell's area is approximately $1.38 \times 1.38$ square kilometers (km$^2$). After dropping cells that lie entirely outside the Swiss borders, our dataset contains 22,350 cells. We take the set $S$ to coincide with the set of 26 cantonal capitals, except for three cases, as explained in Appendix \ref{app:data}. Therefore our final set $S$ includes $23$ cities. 

To solve the model, we also need data on agricultural output $\{y^a(x) L^A(x)\}_{x \in X}$, manufacturing output $\{y_i^m L^M_i \}_{s_i \in S}$, and travel distances from rural cells to cities $\{d(x,s_i)\}_{x \in X, s_i \in S}$. To proxy agricultural output at the grid-cell level, we use the caloric suitability index developed in \cite{galor2016} (taking the post-1500 average across all crops). Because data on production or wages are not available at the municipal level, we use manufacturing gross value added at the canton level as a measure of manufacturing output. Finally, we use the Fast Marching Method (see \cite{allen2014}) to compute bilateral distances. One advantage of this method is that it allows travel distances to depend on geographical characteristics; in particular, we set transit speed through each grid cell to be inversely proportional to the cell's altitude, and we assign maximum speed to cells intersected by a river or a lake. Figure \ref{fig:gridelev} shows the elevation data on our grid.

\begin{figure}
\caption{Elevation data on our grid}
\centering
\includegraphics[height= 0.3\textheight]{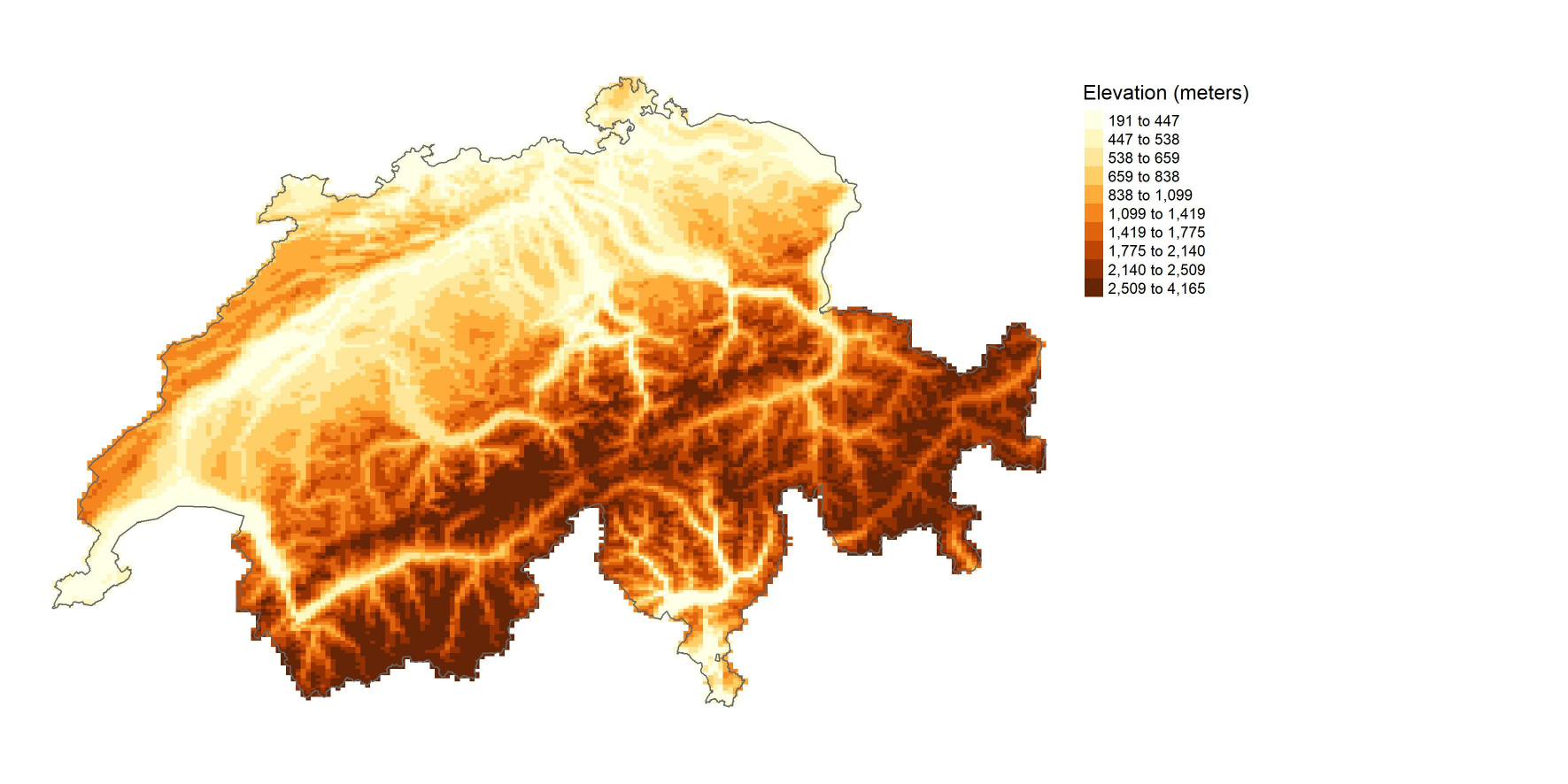}
\parbox{15cm}{\footnotesize \emph{Notes:} This figure reports the elevation data (in meters) at the grid cell level for Switzerland. Each grid cell's area is approximately $1.38 \times 1.38$ square kilometers. }
\label{fig:gridelev}
\end{figure}

\subsection{Simulated Tessellations}

We compute a spatial tessellation under three different scenarios, taking into account progressively more aspects of the model. First, we consider a standard Voronoi tessellation such that (1) all cities are assigned the same weight and (2) transit costs are set equal to one for all cells; this is a ``null'' tessellation that neglects both the role of geography and the role of market forces. Second, we construct a tessellation with equal weights, but this time using a distance matrix that accounts for the Swiss geography. Third, we factor in, on top of geography, the role of market forces by using Algorithm~\ref{algo} to compute the Voronoi tessellation consistent with the equilibrium conditions of the model. 

The model with immobile labor \ref{sec: model} features two parameters: the CES parameter $\alpha$, related to the elasticity of substitution between consumption goods, and the semielasticity of shipping costs with respect to distance travelled, $\delta$. We set $\alpha = 0.5$, consistent with the evidence that expenditure shares on agricultural goods are negatively related to income, and $\delta = 0.2$.\footnote{To choose this value, we calculated the equilibrium tessellation for 20 distinct values of $\delta$, ranging from $0.05$ to $1$ with $0.05$ intervals, and selected the value of $\delta$ that provides the most accurate fit. Although both metrics attain their minimum at $\delta = 0.2$, the Hausdorff distance appears constant within the interval of $[0.05,0.25]$, whereas the Area distance displays a smooth global minimum at $0.2$.
}

Figure \ref{fig:marketareas} presents  the results of our simulations. The black lines trace the borders of the theoretical tessellations, and the orange lines trace the administrative borders between Swiss cantons. In Figure \ref{fig:marketareasa}, the borders between market areas are straight lines (up to a discrete approximation), equidistant from the corresponding cities. In Figure \ref{fig:marketareasb}, borders tend to reflect the presence of mountains and rivers. This is most visible for the three southernmost market areas, which now appear to be entirely located below the Alpine watershed (for comparison, see the elevation map in Figure \ref{fig:gridelev}). Figure \ref{fig:marketareasc} shows the equilibrium tessellation, calculated with Algorithm \ref{algo} described in Section \ref{sec:algo}. 

The impact of market forces is most visible in the northeast, a relatively flat area where the largest cities are located. As a result, the balance between geography and market forces in this area is skewed toward the latter. By comparison, the three market areas in the south host smaller cities and are cut off by the Alps from the rest of Switzerland. As a result, their northern borders mainly reflect the underlying geography and are little affected by market forces.

\begin{figure}
\caption{Simulated market areas and administrative borders}
\begin{subfigure}{\textwidth}
\centering
\includegraphics[height= 0.25\textheight]{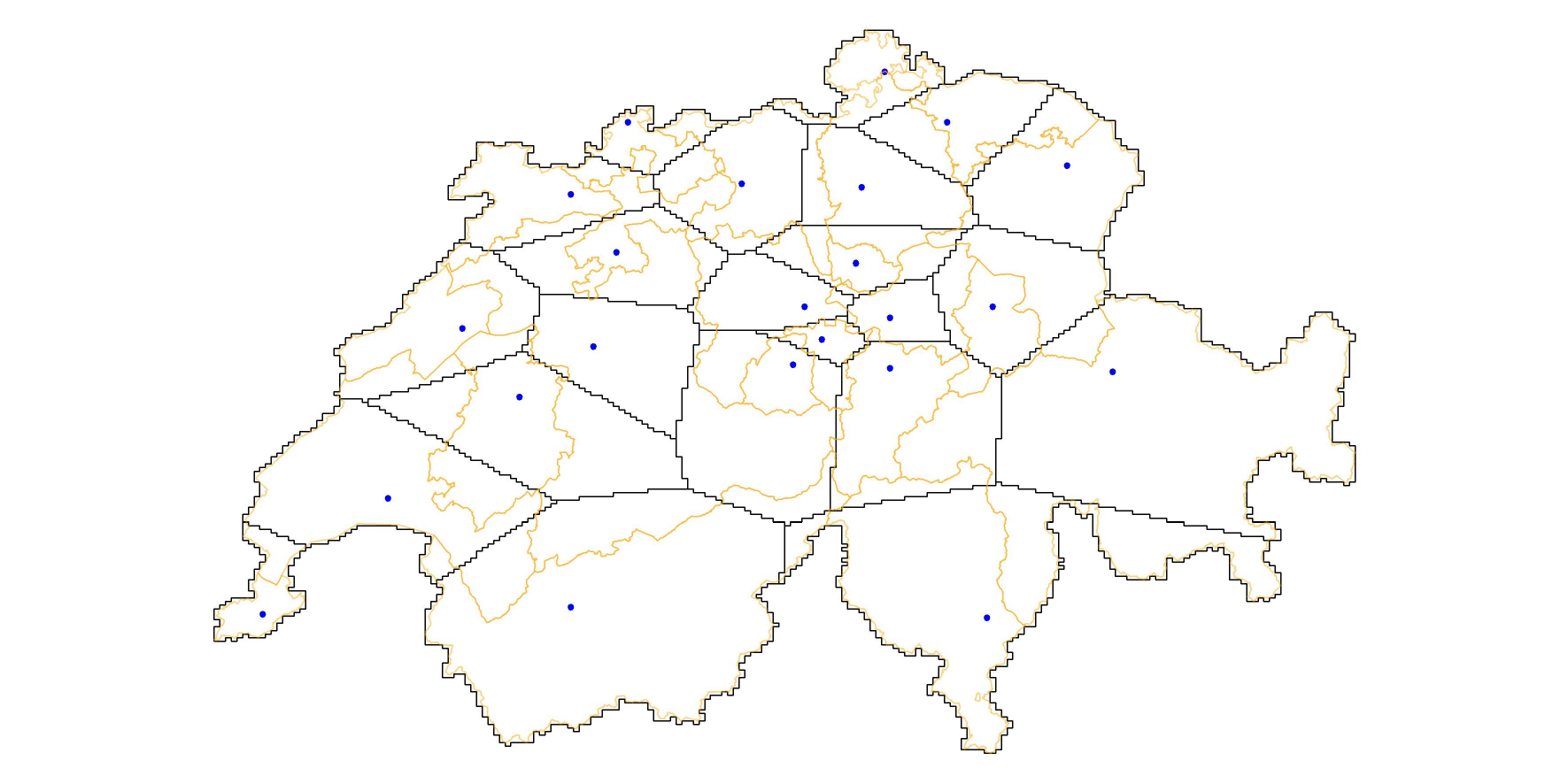}
\caption{Constant Weights, Uniform Geography}
\label{fig:marketareasa}
\end{subfigure}\\
\begin{subfigure}{\textwidth}
\centering
\includegraphics[height= 0.25\textheight]{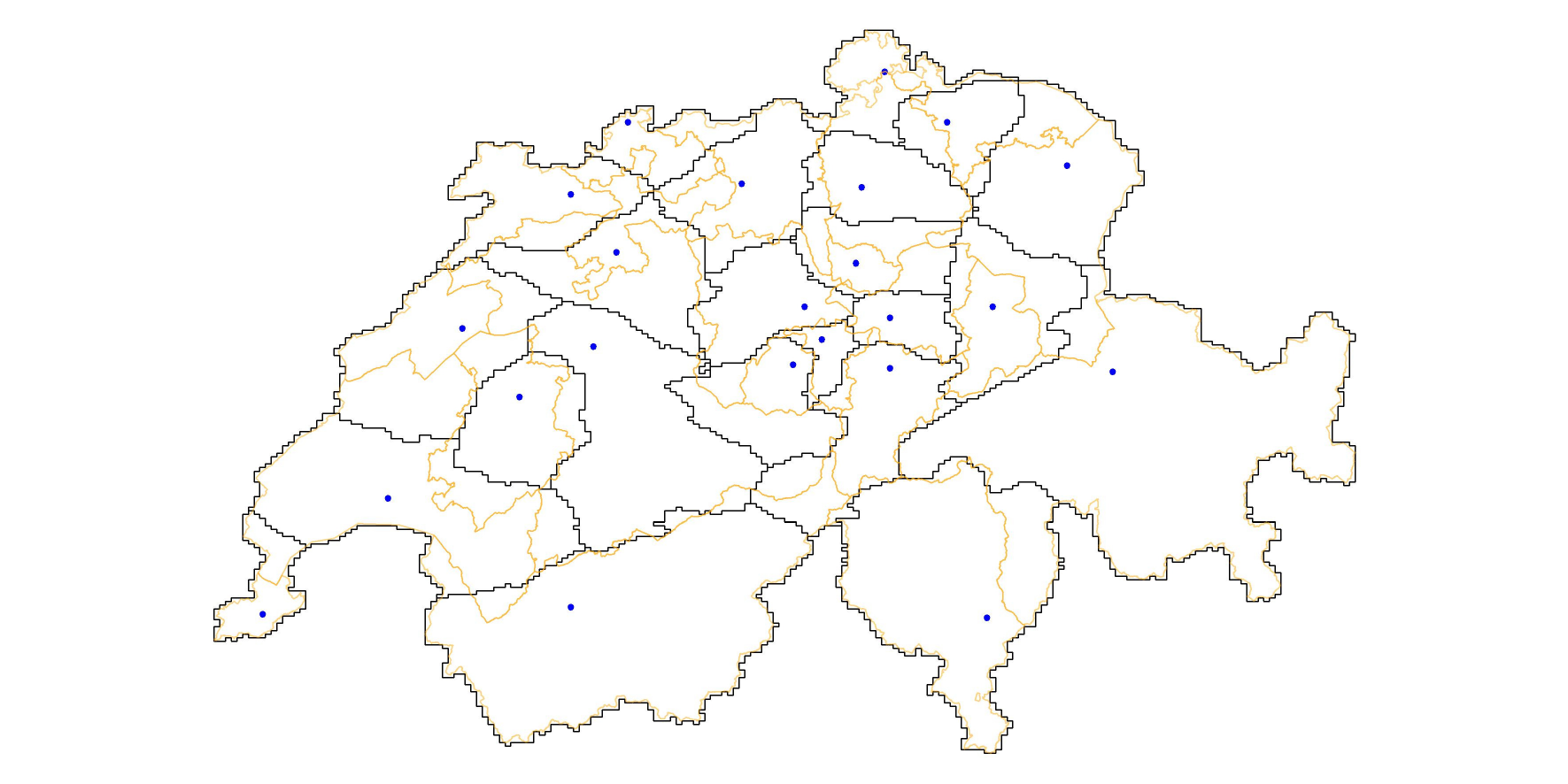}
\caption{Constant Weights, Heterogenous Geography}
\label{fig:marketareasb}
\end{subfigure}\\
\begin{subfigure}{\textwidth}
\centering
\includegraphics[height= 0.25\textheight]{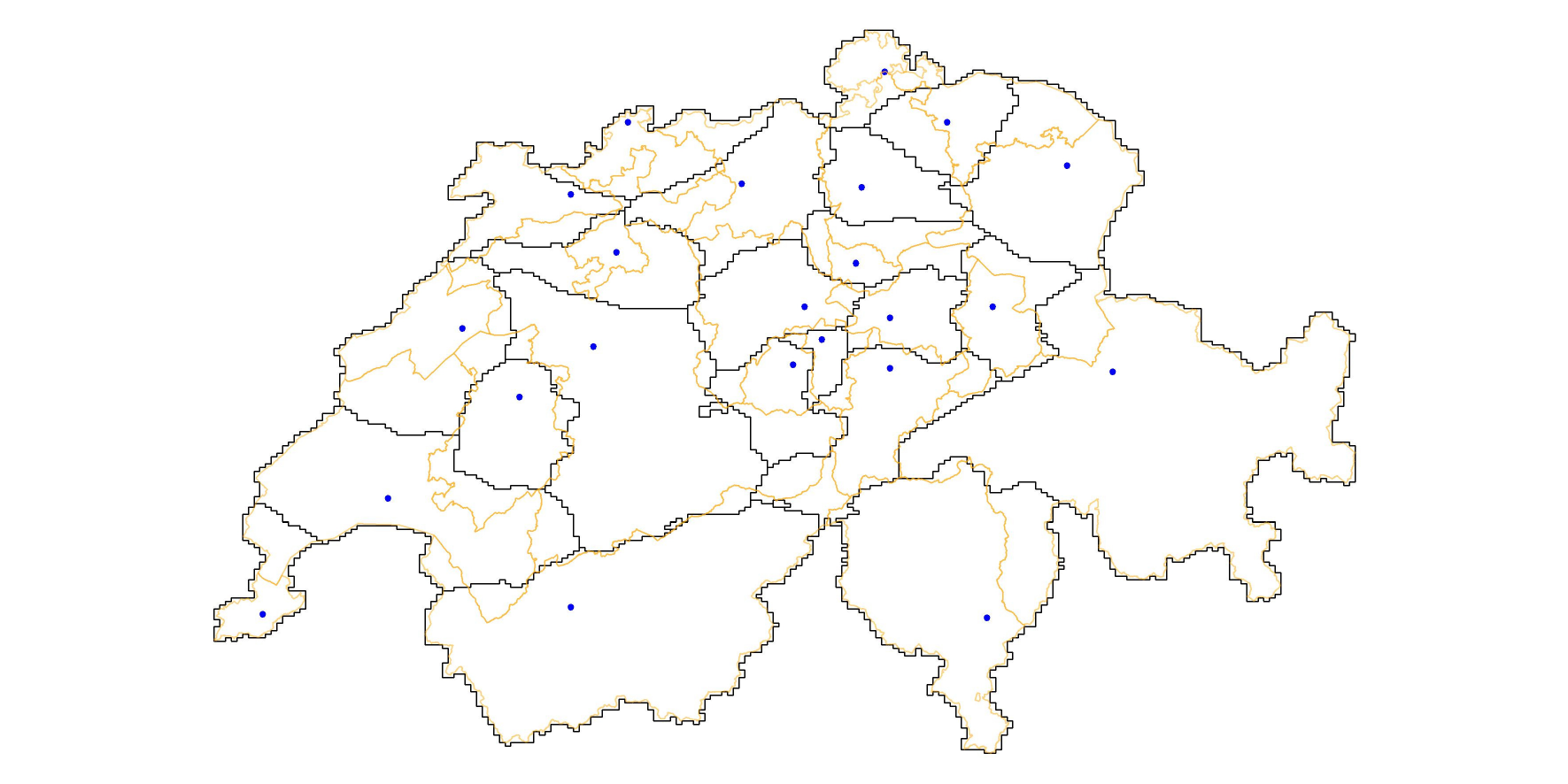}
\caption{Equilibrium Weights, Heterogenous Geography}
\label{fig:marketareasc}
\end{subfigure}
\label{fig:marketareas}
\parbox{15cm}{\footnotesize \emph{Notes:} The figure shows, in black, the simulated tessellations resulting from three different exercises. In the top panel, all cities are assigned the same weight and transit costs are set equal to one for all cells. In the middle panel, all cities are assigned the same weight, but the distance matrix accounts for the underlying Swiss geography. In the bottom panel, the tessellation is computed with Algorithm \ref{algo} according to the equilibrium of the model with immobile labor. For each panel, we superimpose in orange the tessellation of Swiss cantonal borders. }
\end{figure}

\subsection{Matching Administrative Borders}

It is interesting to investigate whether the borders that emerge from our simulations resemble the real-world borders, so in this exercise, we focus on administrative borders. Figure \ref{fig:marketareas} allows for a visual comparison between simulated tessellations (in black) and the observed cantonal borders (in orange). Table \ref{tab:borders} carries out this comparison more formally, reporting percentage changes in goodness of fit (for both metrics) relative to the Euclidean, unweighted baseline. If the model captures some of the forces that shape administrative borders, then we expect it do better than a completely atheoretical baseline. 

To measure the goodness of fit between the model-predicted borders and the data, we use the two metrics introduced in Section~\ref{sec:algo}. Let $\Omega^C$ be the tessellation of the Swiss cantons that we would like to approximately recover from our model. Let $\Omega^B$ be the baseline partition---Standard Voronoi with costant weights and uniform flat geography---shown in Figure~\ref{fig:marketareasa}, let $\Omega^G$ be the one with heterogeneous geography shown in Figure~\ref{fig:marketareasb}, and let  $\Omega^E$ be the equilibrium one shown in Figure~\ref{fig:marketareasc}. In columns 1, 2, and 3, respectively, we report the absolute distances: $d(\Omega^C,\Omega^B)$, $d(\Omega^C,\Omega^G)$, and $d(\Omega^C,\Omega^B)$.

In column 1, we report the relative improvements
\[
\frac{d(\Omega^C,\Omega^G)-d(\Omega^C,\Omega^B)}{d(\Omega^C,\Omega^B)}
\]
associated with using the partition $\Omega^G$ to approximate the cantonal borders with respect to the baseline, for $d = d_H, d_A$. We find that introducing geographical features in the distance computation (as in Figure \ref{fig:marketareasb}) improves the goodness of fit by 23.7\% relative to the baseline. In column 2, we then show the relative improvements 
\[
\frac{d(\Omega^C,\Omega^E)-d(\Omega^C,\Omega^B)}{d(\Omega^C,\Omega^B)}
\]
associated with using the equilibrium partition $\Omega^E$ with respect to the baseline, for $d = d_H, d_A$.  The numbers imply that the equilibrium tessellation is, according to both metrics, a better approximation of cantonal borders than a tessellation based only on geographic characteristics. The improvement is about $3.3$ percentage points in terms of Hausdorff distance and $8.0$ percentage points in terms of area discrepancy.   

All in all, these results suggest that the economic forces at work in our model may play a role in explaining the location of administrative borders. Obviously, administrative and political borders in general are determined by many other factors that are outside the scope of this study. 

\begin{table}[!htbp] \centering 
  \caption{Matching administrative borders} 
  \label{tab:borders} 
\resizebox{15cm}{!}{
\begin{tabular}{@{\extracolsep{5pt}} cccc|cc} 
\\[-1.8ex]\hline 
\hline \\[-3.8ex]
\\ & \multicolumn{3}{c}{Distance to tessellation of Swiss cantons}  & \multicolumn{2}{c}{\begin{tabular}{@{}c@{}} Relative distance  \\ with respect to Baseline\end{tabular}}  \\ \hline
Metric & Baseline & Geography &     \begin{tabular}{@{}c@{}} Geography \\ + \\ Equilibrium Weights\end{tabular}
 & Geography & \begin{tabular}{@{}c@{}} Geography \\ + \\ Equilibrium Weights\end{tabular} \\ 
\hline \\[-1.8ex] 
Hausdorff distance ($km$) & $23.195$ & $17.695$ & $16.919$ & $0.763$ & $0.729$ \\ 
Area discrepancy ($km^2$) & $29,456.470$ & $22,916.690$ & $20,563.460$ & $0.778$ & $0.698$ \\ 
\hline \\[-3.8ex] 
 \\ 
\end{tabular} }
\parbox[t]{15cm}{\footnotesize \textit{Notes:} This table reports the outcome of our empirical exercise. We compare the tessellations from alternative models to the tessellation of cantonal borders in Switzerland. For each model-tessellation, we compute a measure of distance to the data-tessellation. In the first row, we use the Hausdorff distance; in the second row, we use a measure of area discrepancy (see the main text). The reported figures represent the goodness-of-fit improvement, in percentage terms, relative to a model with no market forces and distances given by the Euclidean metric. For the model-tessellation in column 1, distances are computed taking into account elevation and waterways, but market forces are absent. The model-tessellation in column 2 is the equilibrium tessellation computed according to the model described in this paper.}
\end{table}

\section{Conclusions}\label{sec:conclusions}

We have presented a spatial model with the following key characteristics: a realistic geography, a spatially continuous distribution of agents interacting with one location from a finite set, a CES utility function, and iceberg distance costs. Together, these ingredients imply that the equilibrium outcome presents two novel properties: first, there is a well-defined notion of borders between market areas; second, when people are mobile, a number of locations may remain vacant. To derive these properties, we make use of a set of technical results from the mathematical theories of Voronoi diagrams and shape optimization. We hope that the equilibrium properties of the model and the technical tools developed in this paper will broaden the reach of applied work in spatial economics to new sets of questions. 

We are aware that the model may be generalized in a number of directions. Some features that we have abstracted from are, for instance, multiple sectors of production in urban locations and productivity growth. We plan to explore these extensions in future work.

\section*{Acknowledgments}

We would like to thank Jan Bakker, Vincenzo Denicolò, and Paolo Masella for useful comments and suggestions, and Simone Di Marino for pointing out a simpler approach to the proof of Theorem \ref{theorem: existence}. M.S. is member of the Gruppo Nazionale per l’Analisi Matematica, la Probabilit\`a e le loro Applicazioni (GNAMPA) of the Istituto Nazionale di Alta
Matematica (INdAM).
\newpage
\bibliographystyle{plain}

\clearpage

\appendix

\section{Intermediate Results and Proofs}\label{app:intermediate}

\subsection{Existence and Uniqueness of an Equilibrium Without Labor Mobility}\label{app:intermediateunique}
\begin{proposition}\label{prop: equiv}
Consider the model presented in Sections \ref{sub:voronoi} and \ref{sec: model}. Then $p \in \R_{++}^{n+1}$ is an equilibrium if and only if $\nabla \F (p) = 0$.
\end{proposition}
\begin{proof}
By the definition of equilibrium, we need to show that $\nabla \F (p) = 0$ is equivalent to $Z(p) = 0$. Since $\F$ is homogeneous of degree zero, we have that $p \cdot \nabla \F (p) = 0$ for every $p \in \R_{++}^{n+1}$, so in particular $\frac{\partial \F}{\partial q}$ is determined by $\frac{\partial \F}{\partial p_i}$, $i=1,\dots,n$ and we can deduce its formula afterwards. 

To compute $\nabla \F$, let
\begin{equation}\notag
\F_1(p) = -  \sum_{i=1}^n V(p_i, q, qy_i^m) L^M_i , \quad \F_2(p) =  -\int_{X} V(x, p)L^A(x) dx,
\end{equation}
so that $\nabla \F = \nabla \F_1 + \nabla \F_2$. Since $\F_1$ is differentiable in $\R_{++}^{n+1}$, we obtain
\begin{align}\notag
\frac{\partial \F_1}{\partial p_i}(p) &= - \frac{\partial V(p_i, q, \omega_i)}{\partial p_i}L^M_i \\
&= c^a(p_i, q, \omega_i)L^M_i\frac{\partial V(p_i, q, \omega_i)}{\partial \omega_i},\notag
\end{align}
where we have applied Roy's identity in the second step. For $\F_2$, the situation is more subtle. To differentiate under the integral sign, we need to verify some properties of the integrand. 

Consider the change of variable $p \mapsto \bar p$, where
$\bar p_i  = p_i \big/ \left( q^{\frac{\alpha}{\alpha - 1}} + p_i^{\frac{\alpha}{\alpha - 1}}  \right)^{\frac{\alpha - 1}{\alpha }}$
, $ \bar q = 1$, and define $\overline V(x,\bar p) = V(x,p)$, that is
\begin{equation}\label{def: barV}
\overline V(x,\bar p) = \max_{i =1,\ldots,n} \left\{ \frac{\bar p_i y^a(x)}{\Delta(x,s_i)}\right\}.
\end{equation}
First, for each $\bar p \in \R_{++}^{n+1}$, the map $x \mapsto \bar V(x,\bar p)L^A(x)$ is measurable on $X$, since it is continuous and bounded. Then we have to show that the map $\bar p \mapsto \bar V(x,\bar p) L^A(x)$ is differentiable for almost every $x \in X$. This is true since the bisectors have measure $0$, which are the only points where the function is not differentiable. So we find that
\begin{equation}\notag
\frac{\partial \bar V(x,\cdot)L^A(x)}{\partial \bar p_i}(\bar p) =  \chi_{\Omega_i}(x) \frac{y^a(x)}{\Delta(x,s_i)} L^A(x) ,
\end{equation}
almost everywhere in $x \in X$, where $\chi_{\Omega_i}(x)$ is the characteristic function of the Voronoi region $\Omega_i$. Finally we can bound this derivative by a measurable function independent of $\bar p$ as follows:
\begin{equation}\notag
\left| \chi_{\Omega_i}(x)\frac{y^a(x)L^A(x)}{\Delta(x,s_i)}\right| \leq  y^a(x)L^A(x),
\end{equation}
for all $p \in \R_{++}^{n+1}$.

Then, letting $\bar \F_2(\bar p) = \F_2(p)$, we can apply \cite[Theorem 6.28]{klenke2013}, which gives
\begin{equation}\notag
\frac{\partial \bar \F_2}{\partial \bar p_i}(\bar p) = -\int_{\Omega_i} \frac{\partial \overline V(x,\bar p)}{\partial \bar p_i}L^A(x)  dx =  -\int_{\Omega_i} \frac{ y^a(x) L^A(x)}{\Delta(x,s_i)}dx, \quad \bar p \in \R_{++}^{n+1}.
\end{equation}
By the chain-rule formula, we obtain
\begin{equation}\notag
\frac{\partial \F_2}{\partial p_i}( p) = -\int_{\Omega_i} \frac{d  V(p_i,q, \omega(x, s_i))}{d p_i} L^A(x) dx, \quad p \in \R_{++}^{n+1}.
\end{equation}
By Roy's identity,
\begin{equation}\notag
\frac{\partial  V(p_i, q,\omega)}{\partial p_i} = -c^a(p_i,q,\omega)\frac{\partial   V(p_i,q, \omega)}{\partial \omega},
\end{equation}
and, because $\omega(x,s_i) = p_i y^a(x)/\Delta(x,s_i)$, we finally obtain: 
\begin{equation}\notag
\frac{\partial \F_2}{\partial p_i}(p) = \int_{\Omega_i} 
\left[ c^a\left(p_i, q, \omega(x, s_i) \right) - \frac{y^a(x)}{\Delta(x, s_i)}  \right] L^A(x) \frac{\partial V\left(p_i, q, \omega(x, s_i)\right)}{\partial \omega(x,s_i)}  dx, \quad  p \in \R_{++}^{n+1}.
\end{equation}
Using Lemma \ref{lemma: propertiesV}, we can write, for $p \in \R_{++}^{n+1}$:
\begin{equation}\notag
\frac{\partial \F}{\partial p_i}(p) = \left(c^a(q,p_i, \omega_i)L^M_i + \int_{\Omega_i} 
c^a\left(p_i, q, \omega(x, s_i) \right)L^A(x) dx - \int_{\Omega_i} \frac{y^a(x) L^A(x)}{\Delta(x, s_i)} dx\right) v(q,p_i),
\end{equation}
Given that the the marginal utility of wealth $v(q,p_i)$ is strictly positive (see Lemma \ref{lemma: propertiesV}), it is easy to see that $\frac{\partial \F}{\partial p_i}(p) = 0$ is equivalent to $Z_i(p) = 0$. This proves the equivalence of the statement, since $p \cdot Z(p) =0$ gives $Z_{n+1}= 0$ and $p\cdot \nabla \F(p)=0$ gives $\frac{\partial \F}{\partial q}(p) = 0$. 
\end{proof}

\subsection{Comparative Statics of the Equilibrium Without Labor Mobility}\label{app:intermediatecs}
Let us focus on an arbitrary parameter (or vector of parameters) $a$.
Consider the normalized excess demand system 
$$\bar{Z}(p, a) = (Z_1(\bar p,a), Z_2(\bar p,a),...,Z_n(\bar p,a)),$$ 
where $a$ is a vector of exogenous parameters and $\bar p = \bar p(a) = (p_1, p_2, \dots, p_n)$ with $q$ normalized to one. Since $\bar{Z}(\bar p,a)$ is differentiable, if $p$ is a equilibrium price vector, i.e., $\bar Z( \bar p(a),a) =0$, we can use the implicit function theorem and write
\begin{equation}\label{eq:compstat}
D_{a}{\bar p}(a) = -(D_{\bar p} \bar{Z}(\bar p, a))^{-1}D_a \bar{Z}(\bar p, a).
\end{equation} 

\begin{proposition}\label{prop:gsz}
Let the assumptions of Proposition~\ref{prop: equiv} hold and let the excess demand function $Z$ be defined via \eqref{eq: edA}. Then $Z$ satisfies
\begin{enumerate}[(i)]
\item \label{GGS1}
$\begin{cases}
\frac{\partial Z_i}{\partial p_j} > 0 \text{  if  }  \partial \Omega_i \cap \partial \Omega_j \neq \emptyset,\\ 
\frac{\partial Z_i}{\partial p_j} \geq 0  \quad \text{otherwise},  
\end{cases} \quad \text{for } i,j=1,\dots,n, \text{ with } i \neq j;$
\item \label{GGS2} for  $0 < \alpha < 1, \;\;\frac{\partial Z_i}{\partial q} > 0, \quad \text{for }i =1\dots,n$; 
\end{enumerate}
\end{proposition}

\begin{proof}

\begin{enumerate}[Proof of (i)]
\item First, take two neighboring cities $s_i,s_j \in S$, i.e., $\partial \Omega_i(\lambda) \cap \partial \Omega_j(\lambda) \neq \emptyset$, and a price vector $p$. Denote by $p'$ a price vector such that $p'_j > p_j$ and $p'_i = p_i$ for $i \neq j$. We have that $\Omega_i(\lambda') \subset \Omega_i(\lambda)$, which implies
\begin{eqnarray*}
\int_{\Omega_i(\lambda')}\frac{y^a(x)L^A(x)}{\Delta(x,s_i)}dx < \int_{\Omega_i(\lambda)}\frac{y^a(x)L^A(x)}{\Delta(x,s_i)}dx.
\end{eqnarray*}
By Equation \eqref{eq: edA}, $Z^A_i(p') > Z^A_i(p)$.\footnote{To see this, rewrite \eqref{eq: edA} as \begin{equation*}
Z_i(p) = - \frac{ q^{\frac{\alpha}{\alpha-1}}}{ q^{\frac{\alpha}{\alpha-1}} + p_i^{\frac{\alpha}{\alpha-1}}} \int_{\Omega_i(\lambda)}\frac{y^a(x)L^A(x)}{\Delta(x,s_i)}dx + c^a(q, p_i, \omega_i)L^M_i.
\end{equation*}
}

Second, take two cities $s_i,s_j \in S$ that don't have a border in common. In this case, $\Omega_i(\lambda') = \Omega_i(\lambda)$, and thus $Z_i(p') = Z_i(p')$. 
\item By \eqref{eq: edA}, $\frac{\partial Z_i(p)}{\partial q} >0$, since both $c^a(q, p_i, \omega(x, s_i))$ and $c^a(q, p_i, \omega_i)$ are increasing in $q$. \qedhere
\end{enumerate}
\end{proof}

The property \eqref{GGS1} has a key graph theoretic interpretation. We recall the following definition:
\begin{definition}\label{def:graph}
The associated \textit{directed graph} $G(A)$ of a $n \times n$ matrix $A$ consists of $n$ vertices $P_1,\dots,P_n$ where an edge leads from $P_i$ to $P_j$ if and only if $a_{ij} \neq 0$. A directed graph $G$ is \textit{strongly connected} if for any ordered pair $(P_i,P_j)$ of vertices in $G$, there exists a sequence of edges (a \textit{path}) with leads from $P_i$ to $P_j$.
\end{definition}
Now property \eqref{GGS1} shows that the directed graph $G(D_{\bar p} \bar{Z}(\bar p, a))$ consists of vertices $P_1,\dots,P_n$ corresponding to each Voronoi region $\Omega_1,\dots,\Omega_n$ and $P_i$ and $P_j$ are connected by an edge if and only if $\Omega_i$ and $\Omega_j$ share a bisector in common.
These remarks allows to derive properties of the inverse matrix $(D_{\bar p} \bar{Z}(\bar p, a))^{-1}$.


\begin{proposition} \label{prop:gs1}
Let the assumptions of Proposition~\ref{prop: equiv} hold and suppose that $\bar{Z}(\bar p, a) = 0$. Then the matrix $(D_{\bar p} \bar{Z}(\bar p, a))^{-1} $ has all entries negative.
\end{proposition}
\begin{proof}
The proof goes along the same lines of the proof of \cite[Proposition 17.G.3]{mascolell1995} in combination with some more technical results on M-matrices from \cite{berman1994}. First, note that $D_p Z(p, a)p = 0$ because excess demands are homogenous of degree 0. Now from property \eqref{GGS2}, $\partial Z_i / \partial q > 0$, $i = 1,...n$, we have $D_{\bar p} \bar{Z}(\bar p, a) \bar p \ll 0$. Let $I$ denote the $n\times n$ identity matrix, and take $r$ large enough for the matrix $A = \frac{1}{r}D_{\bar p} \bar{Z}(\bar p, a) + I $ to have all its entries nonnegative, thanks to property \eqref{GGS1}. This means that $-D_{\bar p} \bar{Z}(\bar p, a)$ is a nonsingular M-matrix \cite{berman1994}, which satisfies $D_{\bar p} \bar{Z}(\bar p, a) \bar p \ll 0$. We already remarked that, thanks to Proposition~\ref{prop:gsz}, the directed graph $G(D_{\bar p} \bar{Z}(\bar p, a))$ has edges for each pairs of neighboring Voronoi regions.   Since the Voronoi regions are a (finite) partition of the connected set $X$, this graph is strongly connected as in Definition~\ref{def:graph}. Then, thanks to \cite[Theorem 2.2.7]{berman1994}, we have that $-D_{\bar p} \bar{Z}(\bar p, a)$ is irreducible. Now $-D_{\bar p} \bar{Z}(\bar p, a)$ is a irreducible nonsingular M-matrix such that $D_{\bar p} \bar{Z}(\bar p, a) \bar p \ll 0$, and so by \cite[Theorem 6.2.7]{berman1994} its inverse $-(D_{\bar p} \bar{Z}(\bar p, a))^{-1}$ has all positive entries.
\end{proof}

\begin{lemma}\label{lem:cs}
Let the assumptions of Proposition~\ref{prop: equiv} hold and suppose that $\bar{Z}(\bar p, a) = 0$. There exists $\delta_0 >0$ such that for $\delta \geq \delta_0$ we have
\begin{equation}\label{est:diag}
(D_{\bar p} \bar{Z}(\bar p, a))^{-1}_{ii} \leq (D_{\bar p} \bar{Z}(\bar p, a))^{-1}_{ki}, \quad 1\leq i,k\leq n.
\end{equation}
\end{lemma}

\begin{proof}
It is easy to show that
\[
\lim_{\delta \to +\infty} \frac{\partial Z_i}{\partial p_j} = 0 \quad \text{for } j \neq i,
\]
and
\[
\lim_{\delta \to +\infty} \frac{\partial Z_i}{\partial p_i} =-\frac{p_i^{\frac{2-\alpha}{\alpha-1}}}{q^{\frac{\alpha}{\alpha-1}}+p_i^{\frac{\alpha}{\alpha-1}}}\left(p_i^{\frac{\alpha}{\alpha-1}}+\frac{1}{1-\alpha}q^{\frac{\alpha}{\alpha-1}} \right)qy_i^m L_i^M <0,
\]
because all the terms containing the shipping-cost function $\Delta(x,s_i)$ disappear.

Then there exists $\delta_0>0$ such that for $\delta \geq \delta_0$, the row sums of $D_{\bar p} \bar{Z}(\bar p, a)$ are all nonpositive, i.e., $D_{\bar p} \bar{Z}(\bar p, a) e \leq 0$ for $e=(1,\dots,1)^t$. 
The estimate \eqref{est:diag} follows directly from \cite[Lemma 9.3.14]{berman1994}, and the fact that $-D_{\bar p} \bar{Z}(\bar p, a)$ is an irreducible nonsingular M-matrix, as shown in Proposition~\ref{prop:gs1}.
\end{proof}

\subsection{Comparative Statics of the Equilibrium With Labor Mobility}

\begin{proof}[Proof of Lemma \ref{lemma:csmob}]
To prove part \textit{i.}, take the derivative of \eqref{eq:diflambda} with respect to $\delta$:
\begin{equation}\notag
\frac{\partial (\lambda_i - \lambda_j)}{\partial \delta} = \frac{1}{\delta^2}\frac{1-\alpha}{\alpha} \log
\left(\frac{\bar V^{\frac{\alpha}{1-\alpha}}-(y_i^m)^{\frac{\alpha}{1-\alpha}}}{\bar V^{\frac{\alpha}{1-\alpha}}-(y_j^M)^{\frac{\alpha}{1-\alpha}}}  \right)
, \quad i = 1\dots n.
\end{equation} 
If $0<\alpha<1$, 
\begin{equation*}
\frac{\partial (\lambda_i - \lambda_j)}{\partial \delta} \leq 0 \iff 
(y_i^m)^{\frac{\alpha}{1-\alpha}} \geq (y^m_j)^{\frac{\alpha}{1-\alpha}},
\end{equation*}
and the inequality holds for $y_i^m \geq y^m_j$. If, instead, $\alpha<0$, 
\begin{equation*}
\frac{\partial (\lambda_i - \lambda_j)}{\partial \delta} \leq 0 \iff 
(y_i^m)^{\frac{\alpha}{1-\alpha}} \leq (y^m_j)^{\frac{\alpha}{1-\alpha}},
\end{equation*}
and the inequality holds again for $y_i^m \geq y^m_j$. 

Clearly, if $y_i^m > \max_{s_j \in \mathcal{N}_i} y^m_j$, then the inequality holds for all $s_j \in \mathcal{N}_i$. Part \textit{ii.} then follows from the the properties of admissible distance functions in Definition \ref{def:distance}. 
\end{proof}

\begin{proof}[Proof of Lemma \ref{lemma:csmobV}]
To prove part \textit{i.}, take the derivative of \eqref{eq:diflambda} with respect to $\bar{V}$:
\begin{equation}\notag
\frac{\partial (\lambda_i - \lambda_j)}{\partial \bar{V}} = \frac{1}{\delta} \bar{V}^{\frac{\alpha}{1-\alpha} -1} \frac{ (y_j^M)^{\frac{\alpha}{1-\alpha}} -(y_i^m)^{\frac{\alpha}{1-\alpha}}  }{\left(\bar V^{\frac{\alpha}{1-\alpha}}-(y_i^m)^{\frac{\alpha}{1-\alpha}} \right)\left(\bar V^{\frac{\alpha}{1-\alpha}}-(y_j^M)^{\frac{\alpha}{1-\alpha}} \right)}, \quad i = 1\dots n.
\end{equation} 
Since the denominator is always positive,  
\begin{equation*}
\frac{\partial (\lambda_i - \lambda_j)}{\partial \bar{V}} \leq 0 \iff 
(y_i^m)^{\frac{\alpha}{1-\alpha}} \geq (y^m_j)^{\frac{\alpha}{1-\alpha}}.
\end{equation*}
If $0 < \alpha < 1$, the inequality holds for $y_i^m \geq y^m_j$. If instead $\alpha < 0$, the inequality holds for $y_i^m \leq y^m_j$. 

To prove part \textit{ii.}, suppose first that $0<\alpha<1$. Clearly, if $y_i^m > \max_{s_j \in \mathcal{N}_i} y^m_j$, then 
\begin{equation*}
\frac{\partial (\lambda_i - \lambda_j)}{\partial \bar{V}} \leq 0 \text{ for all } s_j \in \mathcal{N}_i.
\end{equation*}
whereas the opposite is true if $y_i^m < \min_{s_j \in \mathcal{N}_i} y^m_j$. The result then follows from the the properties of admissible distance functions in Definition \ref{def:distance}. 
\end{proof}

\subsection{City Formation}
\begin{proof}[Proof of Lemma \ref{lemma:cityformy}]
From \eqref{eq:diflambda}, it is clear that $y_i^m > y^m_j$ implies $\lambda_i - \lambda_j > 0$. This proves part \textit{i.} since the right-hand side is strictly negative.  

Now, if $s_i \in S^*$, that is, if $y_i^m \geq y^m_j$ for all $s_j \in S, s_j \neq s_i$, then $\lambda_i -\lambda_j > 0$ for all $s_j \in S$. Therefore condition \eqref{eq:cityinhab2} is satisfied and $s_i \in T$. This proves part \textit{ii.}

Now suppose $y_i^m = \bar{y}^m$ for all $s_i \in S$. Then $\lambda_i - \lambda_j = 0$ and condition \eqref{eq:cityinhab2} for all $s_i,s_j \in S$. This proves the final part of the Lemma.
\end{proof}
\begin{proof}[Proof of Lemma \ref{lemma:cityformdelta}] Use \eqref{eq:diflambda} into condition \eqref{eq:cityinhab2}. Isolating $\delta$ delivers expression \eqref{eq:cityformdelta}. 
\end{proof}
\begin{proof}[Proof of Lemma \ref{lemma:cityformV}] Use \eqref{eq:diflambda} into condition \eqref{eq:cityinhab2}. Isolating $\bar{V}$ delivers expression \eqref{eq:cityformV}. 
\end{proof}
\section{Proofs of the Main Theorems}\label{app:proofs}

\subsection{Proof of Theorem \ref{theorem: existence}}
\begin{proof}[Proof of Theorem \ref{theorem: existence}]
By Proposition~\ref{prop: equiv}, we saw that equilibrium points are extrema of the cost function $\F$. It is then enough to show that $\F$ has a unique global maximum for $q=1$. First note that $\sup_{p \in \R_{++}^n} \F(p) < + \infty$ since $ \F( p) \leq 0$ for every $ p \in \R_{++}^n$, where now $p = (p_1,\dots,p_n)$ with an abuse of notation.

Consider the change of variable $p \mapsto \bar p$, where 
$$\bar p_i  = p_i \big/ \left( 1 + p_i^{\frac{\alpha}{\alpha - 1}}  \right)^{\frac{\alpha - 1}{\alpha }} = \left(1+p_i^{\frac{\alpha}{1-\alpha}}\right)^{\frac{1-\alpha}{\alpha}}, \qquad \bar q = 1$$
and define $\bar \F(\bar p) = \F(p)$, that is
\[
\bar \F(\bar p) = - \sum_{i=1}^n  \left(1-\bar p_i^{\frac{\alpha}{\alpha-1}}\right)^{\frac{\alpha-1}{\alpha}} y_i^m L^M_i
-\int_{X} L^A(x)\overline V(x, \bar p)\mathrm{d}x,
\]
where $\overline V(x,\bar p)$ was already defined in \eqref{def: barV}. If we show that $\bar \F$ is strictly concave in $D = \{\bar p \in \R^n: \bar p_i > 1, i=1,\dots,\}$ and it attains a global maximum, then the same will be true for $\F$ since the the map $p \mapsto \bar p$ is a smooth change of variable from $\R_{++}^n$ to $D$.

Let 
\[
\bar \F_1(\bar p) = \sum_{i=1}^n \left(1 - \bar p_i^{\frac{\alpha}{\alpha-1}}\right)^{\frac{\alpha-1}{\alpha}}y_i^m L^M_i, \quad \bar \F_2( \bar p) = \int_{X}L^A(x)\bar V(x,\bar p)dx,
\]
so that $\bar \F = - \bar \F_1 - \bar \F_2$. It is easy to check that $\bar \F_1(\bar p)$ is strictly convex in $D$ since its Hessian is just a diagonal matrix with strictly positive entries. The function $\bar V(x,p)$ is convex in $p$, since it is the $\max$ of linear functions. The convexity carries directly to $\bar \F_2(\bar p)$, which is just obtained by integration on another variable. Since the sum of a strictly convex function and a convex function is strictly convex, we have just showed that $\bar \F$ is strictly concave in $D$. 

Let $m = \bar \F(1,\ldots,1)$ and consider the superlevel set
\[
C_m = \{  \bar p \in D : \bar \F( \bar p) \geq m\}.
\]
By definition, $C_m \neq \emptyset$, and it is closed since $\bar \F$ is continuous. We claim that there is $\lambda > 1$ such that $C_m \subseteq \{ \bar p \in D :  \frac{1}{\lambda} \leq \bar p_i-1 \leq \lambda, i=1\dots,n\}$. Indeed, if we consider a sequence $\{ \bar p^k\}_{k \in \N}$ with $ \bar p^k \to  \bar p$ where $\bar p_i = 1$ for at least one $i \in \{ 1,\dots,n\}$, then it is easy to check that $\bar \F( \bar p^k) \to -\infty$. The same happens for a converging sequence $\{\bar  p^k\}_{k \in \N}$ with $\bar p_i^k \to +\infty$ for at least one $i \in \{1,\dots,n\}$.

$C_m$ is then a compact nonempty superlevel set of a continuous function. By Weierstrass theorem, $\F$ attains a maximum $\bar p^*$ in $C_m$, and by strict concavity it is the only global maximum of $\bar \F$, where in particular $\nabla \bar \F (\bar p^*) = 0$.

Finally, define $p^*$ with $p^*_i = ( (\bar p_i^*)^\frac{\alpha}{1-\alpha}-1)^{\frac{1-\alpha}{\alpha}}$. Then $p^*$ is the unique point in $\R^n_{++}$ such that $\nabla \F(p^*) =0$, by the chain rule, and thus the unique equilibrium price vector.
\end{proof}

\subsection{Proof of Theorem~\ref{thm:cs}}

\begin{proof}[Proof of Theorem~\ref{thm:cs}]
Let $L_i$ be the population of the city $s_i \in S$, for a fixed $i =1,\dots,n$. We have that 
\[
\frac{\partial Z_i}{\partial L_i}=c^a(q,p_i,\omega_i), \quad \frac{\partial Z_j}{\partial L_i}=0 \; \text{for } j \neq i. 
\]
Inserting this in formula \eqref{eq:compstat}, instead of $D_a \bar Z$, gives
\[
\frac{\partial p_j}{\partial L_i} = - c^a(q,p_i,\omega_i)(D_p\bar Z(p,L_i))^{-1}_{ji} >0,\quad 1\leq i,j,\leq n,
\]
thanks also to Propositions~\ref{prop:gsz} and \ref{prop:gs1}. Then, by estimate \eqref{est:diag} of Lemma~\ref{lem:cs}, we obtain
\[
\frac{\partial p_i}{\partial L_i} \geq \frac{\partial p_j}{\partial L_i}, \quad j \neq i,
\]
which means that the largest price increase takes place in the city $s_i$, where population grew.
\end{proof}

\subsection{Proof of Theorem~\ref{theo:fact}}
The proof of Theorem~\ref{theo:fact} is a direct consequence of the following two propositions.

\begin{proposition}\label{prop:fact}
Let $0<\alpha<1$, $\delta >0$ and $\bar V >0$ be such that $\bar V > \max_{i=1,\dots,n} y_i^m$.
Moreover, assume that there exists a continuous function $a\colon X \to \R_{++}$ with $0<a_{\min} \leq a(x) \leq a_{\max}$ for all $x \in X$, for some positive constants $a_{\min}, a_{\max}$, and $0<\beta<1$ such that
\begin{equation}
y^a(x) = \frac{a(x)}{(L^A(x))^{1-\beta}}, \quad x \in X.
\end{equation} 
Then, under the assumptions of Theorem~\ref{theorem: existence}, there exists a unique factor-mobility equilibrium.
\end{proposition}

\begin{proposition}\label{prop:fact2}
Let the assumptions of Proposition~\ref{prop:fact} hold, and assume also that the function $a:X \to \R_{++}$ and each distance function $d_s:X \to \R_+$, for $s \in S$, is of class $C^2$. Then the map $\bar V \mapsto \bar L (\bar V)$ is strictly decreasing and thus one-to-one between the sets $(\max_i y_i^m, +\infty)$ and $(0,+\infty)$.

In particular, the following limits hold:
\begin{equation} \label{limitsbarL}
\lim_{\bar V \to +\infty} \bar L = 0, \qquad \lim_{\bar V \to \max_i y_i^m} \bar L = +\infty.
\end{equation}
\end{proposition}

\begin{proof}[Proof of Proposition~\ref{prop:fact}]
Thanks to the homogeneity, we restrict ourselves to the case $q=1$.

The function $p_i \mapsto V(1,p_i,y_i^m) = \left(1+\frac{1}{p_i^{\frac{\alpha}{1-\alpha }}}\right)^{\frac{1-\alpha}{\alpha}}y_i^m$ is strictly decreasing for $p_i >0$, and maps $(0,+\infty)$ to $(y_i^m,+\infty)$ when $0 < \alpha < 1$.
Then there exists a unique $p^*_i$ such that $V(1,p^*_i,y_i^m)= \bar V$, given explicitly by
\begin{equation}\label{eq:pbarV}
p_i = \left( \left( \frac{\bar V}{y_i^m}\right)^{\frac{\alpha}{1-\alpha}}-1\right)^{\frac{\alpha-1}{\alpha}}= \frac{y_i^m}{\left( \bar V^{\frac{\alpha}{1-\alpha}}-(y_i^m)^{\frac{\alpha}{1-\alpha}}\right)^{\frac{1-\alpha}{\alpha}}}.
\end{equation}
The corresponding Voronoi weights satisfy, after some straightforward computations,
\begin{equation}\label{eq:weightmob}
\lambda_i = \frac{1}{\delta}\log(\hat v(1,p_i)) = \frac{1}{\delta}\left(\log(\bar V) -\frac{1-\alpha}{\alpha}\log \left( \bar V^{\frac{\alpha}{1-\alpha}}-(y_i^m)^{\frac{\alpha}{1-\alpha}}\right)\right).
\end{equation}
Note that it it might be that some Voronoi regions are empty. To have only nondegenerate Voronoi regions $\{\Omega_i(\lambda^*)\}_{i=1}^n$, where $\lambda^*$ is the vector of Voronoi weights, the weights would need to satisfy
\[
\lambda_j - \lambda_i < d_j(s_i), \qquad \text{for all } i \neq j.
\]
Since the difference of two weights can be written as
\begin{equation}\notag
\lambda_i- \lambda_j =\frac{1}{\delta} \frac{1-\alpha}{\alpha}\log\left(\frac{\bar V^{\frac{\alpha}{1-\alpha}}-(y_j^M)^{\frac{\alpha}{1-\alpha}}}{\bar V^{\frac{\alpha}{1-\alpha}}-(y_i^m)^{\frac{\alpha}{1-\alpha}}}  \right), \quad i,j=1,\dots,n,
\end{equation}
an additional condition on $V$ and $\delta$ would guarantee that every Voronoi region is nondegenerate.

Once the Voronoi diagram is set, we can determine the rural and the urban population. Thanks to the assumptions in \eqref{eq:condy}, Equation $V(x,p^*) = \bar V$ can be written as
\[
v(1,p^*_i)p^*_i\frac{a(x)}{\Delta(x,s_i)(L^A(x))^{1-\beta}}= \bar V, \quad x \in \Omega_i(\lambda^*)
\]
for $i =1,\dots,n$, such that $\Omega_i(\lambda^*) \neq \emptyset$. This allows us to find the continuous population $L^A$ as 
\begin{align}\label{eq:LA}
L^A(x) &= \left(v(1,p^*_i)p^*_i\frac{a(x)}{\Delta(x,s_i)\bar V}\right)^\frac{1}{1-\beta}\\ \label{eq:LA2}
&= \left(p^*_i\frac{a(x)}{\Delta(x,s_i)y_i^m}\right)^\frac{1}{1-\beta}, \quad \text{for } x \in \Omega_i(\lambda^*),
\end{align}
for $\Omega_i(\lambda^*)\neq \emptyset$ and where in the second equality we used the first identity in \eqref{def:factor}.
Finally, the equilibrium condition $Z_i(p^*) =0$, allows us to obtain the value of $L^M_i$, as follows:
\begin{equation}\notag
L^M_i = \frac{1}{ c^a(1, p^*_i, \omega_i)}\left[\int_{\Omega_i(\lambda^*)}\left(\frac{y^a(x)}{\Delta(x,s_i)}- c^a\left(1,p^*_i,\omega(x,s_i)\right)\right)L^A(x) dx  \right],
\end{equation}
for $i=1,\dots,n$ such that $\Omega_i(\lambda^*)\neq \emptyset$. Using the definitions of $c^a$, $\omega_i$, $\omega(x,s_i)$, this expression can be further simplified to
\begin{align}\notag
L^M_i &= \frac{(p^*_i)^\frac{1}{1-\alpha}}{y_i^m }\int_{\Omega_i(\lambda^*)}\frac{y^a(x)L^A(x)}{\Delta(x,s_i)} dx\\ \label{eq:pop}
 &=  \frac{(p^*_i)^\frac{1}{1-\alpha}}{y_i^m } \int_{\Omega_i(\lambda^*)}\frac{a(x)}{\Delta(x,s_i)}(L^A(x))^\beta dx\\ \label{eq:pop2}
 &= \frac{(p^*_i)^\frac{1-\alpha\beta}{(1-\alpha)(1-\beta)}}{(y_i^m)^{\frac{1}{1-\beta}} }\int_{\Omega_i(\lambda^*)}\left(\frac{a(x)}{\Delta(x,s_i)}\right)^{\frac{1}{1-\beta}} dx, 
\end{align}
where in the last equality, we used identity~\eqref{eq:LA2}. For $i$ such that $\Omega_i(\lambda^*)= \emptyset$, we clearly have $L^M_i = 0$.
\end{proof}

\begin{proof}[Proof of Proposition~\ref{prop:fact2}]
\textbf{i) The limits in \eqref{limitsbarL} hold.}

We begin with the simple case where $y_i^m = \bar y^m >0$ for all $i=1,\dots,n$. We have that all the $p_i$'s are equal, and so the Voronoi tessellation is the standard (unweighted) one for all values of $\bar V$. In that case we have that:
\begin{itemize}
\item if $\bar V \to +\infty$ then $p_i \to 0$. From \eqref{eq:LA2} we find that $L^A(x) \to 0$ for all $x \in X$. From
 \eqref{eq:pop2} we obtain that $L_i^M \to 0$ for all $i=1,\dots,n$, so the population of each city decreases to zero;
\item if $\bar V \to \bar y^m$ then $p_i \to +\infty$ and so the rural population $L^A(x) \to +\infty$ for all $x \in X$. As before, from
 \eqref{eq:pop2} we obtain that $L_i^M \to +\infty$ for all $i=1,\dots,n$.\end{itemize}
In particular, the total population $\bar L$ varies continuously from $0$ to $+\infty$, depending on the values of $\bar V$.

Now let us consider the case where $y_i^m$ do not coincide. First assume that there exists a city $s_{\bar k}$, with $1\leq \bar k\leq n$, with the highest manufacturing output, i.e., $\max_i y_i^m = y_{\bar k}^M$  and $y_i^m < y_{\bar k}^M$ for all $i \neq k$.
In this case, the Voronoi tessellation depends on $\bar V$. 

When $\bar V \to +\infty$, thanks to \eqref{eq:pi_V}, we have that $p_i \to 0$, for all $i=1,\dots,n$ and, from \eqref{eq:LA2} we find that $L^A(x) \to 0$ for all $x \in X$.

Then, as in the case where all $y_i^m$ coincide, we find that, for all $i=1,\dots,n$, $L_i^M \to 0$, for $i$ such that $\Omega_i(\lambda^*) \neq \emptyset$, and so $\bar L \to 0$. 

If we let $\bar V \to y_{\bar k}^M$, then something different happens. In this case, only $p_{\bar k} \to +\infty$ while $\lim_{\bar V \to y_{\bar k}^M} p_i < +\infty$ for $i \neq \bar k$. Moreover $\lim_{\bar V \to y_{\bar k}^M} \lambda_{\bar k} = +\infty$ while $\lim_{\bar V \to y_{\bar k}^M} |\lambda_i| < +\infty$ for $i \neq \bar k$. This means that, as expected, the region $\Omega_{\bar k}$ grows and eventually becomes the whole domain $X$. Then, for $|\bar V - y_{\bar k}^M|$ sufficiently small, $\Omega_{\bar k} = X$, or equivalently $\Omega_i = \emptyset$ for $i \neq \bar k$. Using \eqref{eq:LA2}, we find that the rural population grows only in the $\bar k$ region, which eventually becomes the whole domain, i.e., $L^A(x) \to +\infty$ for $x \in \Omega_{\bar k}$. In this regime, the total city population is concentrated only in the city $s_{\bar k}$. Eventually, as $\bar V \to y_{\bar k}^M$ we find, similarly to the previous case,
that $L_{\bar k}^M \to +\infty$ and so $\bar L \to +\infty$. 

The argument extends analogously in the case where two or more cities have the same maximum $y_i^m$.

In all cases, we found that the total population $\bar L$ varies continuously from $0$ to $+\infty$, depending on the values of $\bar V$.

\textbf{ii) The function $\bar V \mapsto \bar L (\bar V)$ is strictly decreasing.}

The total population $\bar L$ can be written explicitly as
\begin{equation}\label{eq:LV}
\bar L(\bar V) = \sum_{i=1}^n\frac{1}{\bar V^{\frac{1}{1-\beta}}}\left( \frac{\bar V^{\frac{\alpha}{1-\alpha}}}{\bar V^{\frac{\alpha}{1-\alpha}}-(y_i^m)^{\frac{\alpha}{1-\alpha}}}\right)^{\frac{1-\alpha\beta}{(1-\beta)\alpha}}\int_{\Omega_i(\lambda)}\left(\frac{a(x)}{\Delta(x,s_i)} \right)^{\frac{1}{1-\beta}}dx.
\end{equation}
The above formula is obtained from the definition \eqref{eq:totpop} of $\bar L$, together with \eqref{eq:LA2}, \eqref{eq:pop2}, and the formula for the welfare-equalizing price \eqref{eq:pbarV}. We want to show that $\bar L '(\bar V) = \frac{d \bar L}{d \bar V} (\bar V) < 0$ for every $\bar V > \max_i y_i^m$. 

The proof proceeds in three main steps. First, we rewrite \eqref{eq:LV} to ensure that the function inside the integral is continuous across the borders. Intuitively, this is possible because the farmers' optimal trading choices smoothen the indirect utility function across the borders (whereas the gradient of $a(x) / \Delta(x, s_i)$ obviously changes). Second, we introduce a lemma that allows us to work out an expression for the derivative of the integral term with respect to a parameter that may affect both the integrand and the domain of integration. Third, we use this result to show that the sign of the derivative of $\bar{L}$ with respect to $\bar{V}$ can be unambiguously determined.

\textit{Step 1.} Formula~\eqref{eq:LV} although explicit in $\bar V$, is not in a convenient form to be differentiated, since the function under the integral sign is not continuous across $\partial \Omega_i$. To make it continuous, we make a multiple of $V(x,p)$ appear, as follows.

Let $U(x,\bar V) = V(x,p(\bar V))$. From the definition of $V(x,p)$ and identity \eqref{eq:pbarV}, we obtain
\begin{equation}
U(x,\bar V) = \max_i \frac{\bar V y^a(x)}{\Delta(x,s_i)\left(\bar V^{\frac{\alpha}{1-\alpha}}-(y_i^m)^{\frac{\alpha}{1-\alpha}}\right)^{\frac{1-\alpha}{\alpha}}}.
\end{equation}

With this in mind, we want to make the quantity $\left(  \Delta(x,s_i)\left(\bar V^{\frac{\alpha}{1-\alpha}}-(y_i^m)^{\frac{\alpha}{1-\alpha}}\right)^{\frac{1-\alpha}{\alpha}} \right)^{-1}$ appear in the integral in \eqref{eq:LV}, because of its continuity properties across the bisectors.

We have
\begin{align*}
&\int_{\Omega_i(\lambda)}\left(\frac{a(x)}{\Delta(x,s_i)} \right)^{\frac{1}{1-\beta}}dx\\
&\quad = \left(\bar V^{\frac{\alpha}{1-\alpha}}-(y_i^m)^{\frac{\alpha}{1-\alpha}}\right)^{\frac{1-\alpha}{\alpha(1-\beta)}} \int_{\Omega_i(\lambda)}\left(\frac{a(x)}{\Delta(x,s_i)\left(\bar V^{\frac{\alpha}{1-\alpha}}-(y_i^m)^{\frac{\alpha}{1-\alpha}}\right)^{\frac{1-\alpha}{\alpha}}} \right)^{\frac{1}{1-\beta}} \! \! \! \! \! \! \! \! dx.
\end{align*}

Replacing this into \eqref{eq:LV}, we obtain, after several simplifications
\begin{equation}
\bar L (\bar V) = \sum_{i =1}^n f_i(\bar V) \int_{\Omega_i(\lambda)} g_i(x,\bar V) dx,
\end{equation}
where
\begin{equation}\label{def:fg}
f_i(\bar V) = \frac{\bar V^{\frac{\alpha}{1-\alpha}}}{\bar V^{\frac{\alpha}{1-\alpha}}-(y_i^m)^{\frac{\alpha}{1-\alpha}}}, \quad g_i(x,\bar V) = \left(\frac{a(x)}{\Delta(x,s_i)\left(\bar V^{\frac{\alpha}{1-\alpha}}-(y_i^m)^{\frac{\alpha}{1-\alpha}}\right)^{\frac{1-\alpha}{\alpha}}} \right)^{\frac{1}{1-\beta}}.
\end{equation}

\textit{Step 2.} We now compute the derivative. We have
\begin{equation}\label{eq:derL}
\bar L'(\bar V) = \sum_{i=1}^n \left[f_i'(\bar V) \int_{\Omega_i(\lambda)} g_i(x,\bar V) dx + f_i(\bar V) \frac{\partial}{\partial \bar V}\left( \int_{\Omega_i(\lambda)} g_i(x,\bar V) dx \right)  \right].
\end{equation}
An easy computation shows that $f_i'(\bar V) < 0$ for all $\bar V > \max_i y_i^m$. Since at least one $\Omega_i$ is nonempty, and $g_i (x,\bar V) > 0$ for all $x \in X$ and $\bar V > \max_i y_i^m$, we find that 
\begin{equation}\notag
\sum_{i=1}^n f_i'(\bar V) \int_{\Omega_i(\lambda)} g_i(x,\bar V) dx < 0,
\end{equation}
for $\bar V > \max_i y_i^m$. We now want to show that the second term in \eqref{eq:derL} is not positive. To do that, we use the following lemma, giving us a formula for the derivative of an integral on a variable domain.

\begin{lemma}\label{lem:shapeder} 
Let $a \mapsto f(\cdot, a) \in L^1(\R^2)$ be differentiable for every $a > 0$ and $f(\cdot, a) \in W^{1,1}(\R^2)$. Let $X \subset \R^2$ be a bounded Lipschitz domain partitioned in Voronoi regions $\{\Omega_i\}_{i=1}^n$ as in Section~\ref{sub:voronoi}, depending on a weight vector $\lambda = (\lambda_1,\dots,\lambda_n)$ and with a set of admissible distance function $\{d_s(\cdot)\}_{s \in S}$ of class $C^2$. Let $i = 1,\dots,n$ be such that $|\Omega_i (\lambda)| > 0$. Then,  the map $a \mapsto I(a) = \int_{\Omega_i(\lambda)}f(x,a)dx$ is differentiable for every $a > 0$ and we have 
\begin{equation}
I'(a) = \int_{\Omega_i(\lambda)} \frac{\partial f}{\partial a}(x,a) dx + \sum_{k \in \NN_i} \frac{\partial (\lambda_i-\lambda_k)}{\partial a}  \int_{\Gamma_{ik}} f(x,a)\nu_i(x)\cdot \frac{\partial \omega_{ki}}{\partial \lambda_i}(x, \lambda)d \sigma(x),
\end{equation}
where
\begin{itemize}
\item  $\NN_i = \{k = 1,\dots,n, k \neq i : \partial \Omega_k \cap \partial \Omega_i \neq \emptyset\}$,
\item $\nu_i(x)$ is the unit normal vector pointing outside of $\Omega_i$ at $x \in \partial \Omega_i$,
\item $\Gamma_{ik} \subset \partial \Omega_i$ with $\cup_{k \in \NN_i}\Gamma_{ik} = \partial \Omega_i$ and $\partial \Omega_k \cap \partial \Omega_i = \Gamma_{ik}$,
\item the map $(x,\lambda) \mapsto \omega_{ki}(x,\lambda)$ is a local parametrization of the curve $\Gamma_{ik}$ at a given set of weights $\lambda$,
\item $d\sigma(x)$ is the surface measure on $\partial \Omega_i$.
\end{itemize}
\end{lemma}

The proof of the lemma appears below.

\textit{Step 3.} We apply Lemma~\ref{lem:shapeder} with $g_i$ instead of $f$, $\bar{V}$ instead of $a$, and $\bar V > \max_i y_i^m$. We find
\begin{align*}
& \sum_{i=1}^n f_i(\bar V) \frac{\partial}{\partial \bar V}\left( \int_{\Omega_i(\lambda)} g_i(x,\bar V) dx \right)\\
 &\qquad =\sum_{i=1}^n f_i(\bar V) \left( \int_{\Omega_i(\lambda)} \frac{\partial g_i}{\partial \bar V}(x,\bar V) dx + \sum_{k \in \NN_i} \frac{\partial (\lambda_i-\lambda_k)}{\partial \bar V}  \int_{\Gamma_{ik}} g_i(x,\bar V)\nu_i(x)\cdot \frac{\partial \omega_{ki}}{\partial \lambda_i}(x, \lambda)d \sigma(x) \right).
\end{align*}
It is easy to show that $\frac{\partial g_i}{\partial \bar V}(x,\bar V) < 0$ for all $\bar V > \max_i y_i^m$ and $x \in \Omega_i(\lambda)$. Then the first term of the above expression is strictly negative, that is
\[
\sum_{i=1}^n f_i(\bar V)  \int_{\Omega_i(\lambda)} \frac{\partial g_i}{\partial \bar V}(x,\bar V) dx < 0, \quad \text{for } \bar V > \max_i y_i^m,
\]
and we focus on the second term. By construction, for each $i,k = 1,\dots,n$ such that $\partial \Omega_i \cap \partial \Omega_k = \Gamma_{ik} \neq \emptyset$, we have that $g_i(x,\bar V) = g_k(x,\bar V)$, for $x \in \Gamma_{ik}$. Moreover, we have that $\nu_i(x) = - \nu_k(x)$ and $\frac{\partial \omega_{ki}}{\partial \lambda_i}(x, \lambda) = - \frac{\partial \omega_{ki}}{\partial \lambda_k}(x, \lambda)$ for $x \in \Gamma_{ik}$. This gives the identity
\[
\int_{\Gamma_{ik}} g_i(x,\bar V)\nu_i(x)\cdot \frac{\partial \omega_{ki}}{\partial \lambda_i}(x, \lambda)d \sigma(x) = \int_{\Gamma_{ik}} g_k(x,\bar V)\nu_k(x)\cdot \frac{\partial \omega_{ki}}{\partial \lambda_k}(x, \lambda)d \sigma(x),
\]
which we use to rewrite the last term in the derivative of $\bar L$ as follows:
\begin{align*}
&\sum_{i=1}^n f_i(\bar V)\left[\sum_{k \in \NN_i} \frac{\partial (\lambda_i-\lambda_k)}{\partial \bar V}  \int_{\Gamma_{ik}} g_i(x,\bar V)\nu_i(x)\cdot \frac{\partial \omega_{ki}}{\partial \lambda_i}(x, \lambda)d \sigma(x)\right] \\
&\qquad =\sum_{\mathclap{\substack{i,k=1 \\ \partial \Omega_i \cap \partial \Omega_k \neq \emptyset}}}^n (f_i(\bar V) - f_k(\bar V))\frac{\partial (\lambda_i-\lambda_k)}{\partial \bar V}  \int_{\Gamma_{ik}} g_i(x,\bar V)\nu_i(x)\cdot \frac{\partial \omega_{ki}}{\partial \lambda_i}(x, \lambda)d \sigma(x).
\end{align*}
We now want to show that also this last quantity is not positive. We first show that $(f_i(\bar V) - f_k(\bar V))\frac{\partial (\lambda_i-\lambda_k)}{\partial \bar V} \leq 0$ for all $i,k=1,\dots,n$, $i \neq k$. Indeed, from \eqref{def:fg} and \eqref{eq:diflambda}, after a series of computations, we obtain that
\[
(f_i(\bar V) - f_k(\bar V))\frac{\partial (\lambda_i-\lambda_k)}{\partial \bar V}  = - \frac{1}{\delta}\frac{\bar V^{\frac{3\alpha-1}{1-\alpha}} ((y_k^M)^{\frac{\alpha}{1-\alpha}}- (y_i^m)^{\frac{\alpha}{1-\alpha}})^2}{(\bar V^{\frac{\alpha}{1-\alpha}}-(y_k^M)^{\frac{\alpha}{1-\alpha}})^2 (\bar V^{\frac{\alpha}{1-\alpha}}-(y_i^m)^{\frac{\alpha}{1-\alpha}})^2} \leq 0.
\]
Finally $\nu_i(x)\cdot \frac{\partial \omega_{ki}}{\partial \lambda_i}(x, \lambda) \geq 0$ for all $x \in \Gamma_{ik}$. Indeed, by definition of the Voronoi diagram, the region $\Omega_i(\lambda)$ grows uniformly as $\lambda_i$ grows, which means that the vector field $\frac{\partial \omega_{ki}}{\partial \lambda_i}(x, \lambda)$ on $\Gamma_{ik}$ points outside of $\Omega_i$ at each $x \in \Gamma_{ik}$, which is equivalent to say that its scalar product with the outward unit normal vector $\nu_i(x)$ is not negative. This ends the proof of the strict monotonicity of $\bar V \mapsto \bar L(\bar V)$.

\textbf{iii) End of the proof.}

Since the function $\bar L$ is strictly decreasing and continuous, and since the two limits \eqref{limitsbarL} hold, we have that $\bar L$ defines a one-to-one correspondence between $(\max_i y_i^m, +\infty)$ and $(0,+\infty)$.
\end{proof}

\begin{proof}[Proof of Lemma~\ref{lem:shapeder}]
This is essentially a direct application of a classical result in shape optimization \cite[Theorem 5.2.2]{henrot2006}. Let $I$ be the function defined as $a \mapsto I(a) = \int_{\Omega_i(\lambda)}f(x,a)dx$ that we want to differentiate. Since the movements of the domain $\Omega_i$ are parametrized by the Voronoi weights $\lambda$, we first use the chain rule to get
\[
I'(a) = \sum_{k=1}^n \frac{\partial \lambda_k}{\partial a}\frac{\partial}{\partial \lambda_k}\left(\int_{\Omega_i(\lambda)}f(x,a)dx \right)=\sum_{\mathclap{\substack{k=1\\ \partial \Omega_k \cap \partial \Omega_i \neq \emptyset}}}^n \frac{\partial \lambda_k}{\partial a}\frac{\partial}{\partial \lambda_k}\left(\int_{\Omega_i(\lambda)}f(x,a)dx \right),
\]
where in the second equality, we discarded all indices $k$ such that $\Omega_k$ does not share the border with $\Omega_i$, and so $\Omega_i$ is insensitive to the variations of these $\lambda_k$.

Now, let $i = 1,\dots,n$ be fixed. For for $k=1,\dots,n$, $k \neq i$ (the case $k=i$ will be considered later), and a fixed weight vector $\lambda = (\lambda_1,\dots, \lambda_n)$, we consider $B = B_{\lambda_i-\lambda_k}(s_i,s_k)$, the bisector containing $\Gamma_{ik}$. For $t \geq 0$, we let $\lambda_t = (\lambda_1,\dots,\lambda_k+t,\dots,\lambda_n)$ and $B(t)$ the bisector between $s_i$ and $s_k$ obtained with $\lambda_t$ instead of $\lambda$. Let $\omega_{ik}: \R \times [0,T) \to \R^2$ be a parametrization of $B(t)$.

Now, thanks to the regularity of the distance functions, the curve $B$ is of class $C^2$, and thus there exists a tubular neighborhood $U \supset B$ where, for every $y \in U \cap X$ there is a unique $x \in B$, $r \in \R$ such that $y = x + r \nu_i(x)$. This defines a projection $P_k: U \to B$ as $P_k(y)= x$. Here $\nu_i(x)$ is the unit normal vector field on $B$ that coincides with the outward unit normal to $\partial \Omega_i$ on $\Gamma_{ik}$. Then there exists $T>0$ such that $B(t) \cap X \subset U$ for all $t \in [0,T)$.

We now construct a diffeomorphism $\Phi_k : [0,T) \times \R^2 \to \R^2$ with the property that $\Phi_k(t,\Omega_i(\lambda)) = \Omega_i(\lambda_t)$. In this way, taking the derivative with respect to $\lambda_k$ in the last integral term is equivalent to taking the derivative with respect to $t$ at $t=0$. Let $\Phi_k(t,x) = x +r_k(x,t) \nu_i(x)$, where, for $x \in B$, we define $r_k(x,t) = (\omega_{ik}(s,t)-x)\cdot \nu_i(x)$, where $s = s(x) \in \R$ is the unique value such that $P_k(\omega_{ik}(s,t)) = x$. Note that for $t = 0$, $\omega_{ik}(s,0) = x$, so $\Phi_k(x,0) = x$, for all $x \in B$. For $x \in \R^2 \setminus U$ we let $r_k(x,t) = 0$, for all $t \in [0,T)$. Finally, for $x \in U \setminus B$, we extend $r_k$ to be $C^2$ and with $r_k(x,0)= 0$. This is possible because $\omega_{ik} \in C^2(\R \times [0,T))$. By construction, we have that $\Phi_k(\Omega_i(\lambda)) = \Omega_i(\lambda_t)$ and $\Phi(t,\cdot)$ belongs to the class $C^1([0,T); C^1(\R^2,\R^2))$.

We then apply \cite[Theorem 5.2.2]{henrot2006} (see also \cite[Corollary 5.2.8 and Remark 5.2.9]{henrot2006})  by using $\Phi_k$ to model the movement of $\Omega_i(\lambda)$ with respect to $\lambda_k$. Using the fact that 
\[\frac{\partial \Phi_k}{\partial t}(t,x)|_{t=0}= \left(\frac{\partial \omega_{ki}}{\partial t}(s,0)\cdot \nu_i(x)\right) \nu_i(x),
\] 
we find, by rewriting $\frac{\partial \omega_{ik}}{\partial t}(s,0)$ as $\frac{\partial \omega_{ik}}{\partial \lambda_k}(x,\lambda)$ with a slight abuse of notation:
\[
\frac{\partial}{\partial \lambda_k}\left(\int_{\Omega_i(\lambda)}f(x,a)dx \right) = \int_{\Omega_i(\lambda)}\frac{\partial f}{\partial \lambda_k}(x,a)dx + \int_{\Gamma_{ik}}f(x,a) \nu_i(x)\cdot \frac{\partial \omega_{ik}}{\partial \lambda_k}(x,\lambda)d\sigma(x), \quad \text{for } k\neq i.
\]
For $k = i$, we define $\Phi_i(t,x) = x+\sum_{k \in \NN_i} r_k(x,t) \nu_i(x)$. Then we find
\[
\frac{\partial}{\partial \lambda_i}\left(\int_{\Omega_i(\lambda)}f(x,a)dx \right) = \int_{\Omega_i(\lambda)}\frac{\partial f}{\partial \lambda_i}(x,a)dx + \sum_{k \in \NN_i}\int_{\Gamma_{ik}}f(x,a) \nu_i(x)\cdot \frac{\partial \omega_{ik}}{\partial \lambda_i}(x,\lambda)d\sigma(x),
\]
where we used the fact that $\cup_{k \in \NN_i}\Gamma_{ik} = \partial \Omega_i$.

Now, since $\frac{\partial \omega_{ik}}{\partial \lambda_k}(x,\lambda) = -\frac{\partial \omega_{ik}}{\partial \lambda_i}(x,\lambda)$ for $x \in \Gamma_{ik}$, we obtain
\begin{align*}
I'(a) &= \sum_{\mathclap{\substack{k=1\\ \partial \Omega_k \cap \partial \Omega_i \neq \emptyset}}}^n \frac{\partial \lambda_k}{\partial a} \int_{\Omega_i(\lambda)}\frac{\partial f}{\partial \lambda_k}(x,a)dx + \sum_{k \in \NN_i} \frac{\partial (\lambda_i-\lambda_k)}{\partial a}  \int_{\Gamma_{ik}} f(x,a)\nu_i(x)\cdot \frac{\partial \omega_{ik}}{\partial \lambda_i}(x, \lambda)d \sigma(x)\\
& = \int_{\Omega_i(\lambda)} \frac{\partial f}{\partial a}(x,a) dx + \sum_{k \in \NN_i} \frac{\partial (\lambda_i-\lambda_k)}{\partial a}  \int_{\Gamma_{ik}} f(x,a)\nu_i(x)\cdot \frac{\partial \omega_{ik}}{\partial \lambda_i}(x, \lambda)d \sigma(x),
\end{align*}
where we again used the chain rule to come back to the derivative with respect to $a$ in the first term. The proof follows.
\end{proof}

\section{Additional Results}

\subsection{Comparative Statics With Labor Mobility in an Open Economy} \label{app:addrescs}

This section continues the comparative statics analysis of Section \ref{sec:compstatmob}. Again we restrict our attention to an open economy scenario where the welfare scalar is an exogenous parameter and total population is free to vary. Lemma \ref{lemma:csmobprod} studies the effect of a shock to urban productivity $y_i^m$, $i=1\dots n$, in one city.
\begin{lemma}\label{lemma:csmobprod} Let the Assumptions of Theorem \ref{theo:fact} hold. Moreover, suppose that $0 < \alpha < 1$.  Then for all cities $s_i \in S$, $i = 1 \dots n$,
\begin{enumerate}[$i.$]
\item
\begin{equation*}
\frac{\partial \lambda_i}{\partial y_i^m} > 0, \quad \text{ and } \quad \frac{\partial \lambda_j}{\partial y_i^m} = 0 \quad \text{ for } s_j \neq s_i 
\end{equation*}
\item
\begin{equation*}
\frac{\partial L^M_i}{\partial y_i^m} > 0, \text{ and } \quad \frac{\partial L^M_j}{\partial y_i^m} \leq 0 \quad \text{ for }   s_j \neq s_i ,
\end{equation*}
with strict inequality if $\partial \Omega_i \bigcap \partial \Omega_j \neq \emptyset$.
\end{enumerate}
\end{lemma}
\begin{proof}
With $\bar{V}$ fixed, part \textit{i} is evident upon inspection of \eqref{eq:weightmob}. Concerning the second part of the Lemma, rewrite \eqref{eq:pop2} as: 
\begin{equation}\notag
L^M_i = \frac{(y_i^m)^\frac{\alpha}{1-\alpha}}{ \left[ \bar{V}^{\frac{\alpha}{1-\alpha}} -  (y_i^m)^{\frac{\alpha}{1-\alpha}} \right]^{\frac{1-\alpha \beta}{\alpha(1-\beta)}} } \int_{\Omega_i(\lambda^*)}\left(\frac{a(x)}{\Delta(x,s_i)}\right)^{\frac{1}{1-\beta}} dx, \quad i = 1\dots n,
\end{equation}
where we have used equation $\eqref{eq:pi_V}$. For city $s_i \in S$, the first term of the product clearly increases with $y_i^m$, because $1 - \alpha \beta > 0$. The second term is increasing by the result in part \textit{i}. For a city $s_j \in S$, $j \neq i$, the first term of the product does not vary with $y_i^m$, whereas the second term is weakly decreasing with $y_i^m$, with strict inequality if $\partial \Omega_i \bigcap \partial \Omega_j \neq \emptyset$, i.e. if $s_i$ and $s_j$ share a border. 
\end{proof} 
The effect of urban productivity $y^m_i$ on urban population in $s_i$ runs through three different channels. First, there is an increase in nominal income in the city. Second, for $0 < \alpha < 1$, the labor force shifts from the rural to the urban sector. Third, and finally, the city's market area expands, which helps to mitigate the pressure of increased agricultural prices.

\subsection{Comparative Statics With Labor Mobility in a Closed Economy} \label{app:addrescsL}

In a closed economy, the change in the value of a parameter will induce an endogenous response in the level of welfare $\bar{V}$. Differently than total population $\bar{L}$, however, the welfare scalar enters directly into the expression for the Voronoi weights (see Equation $\eqref{eq:diflambda}$). 

Consider the total effect of a change in $\delta$ on the difference between the Voronoi weights of two cities $s_i, s_j \in S$. Using the chain rule, we find: 
\begin{equation}\notag
\frac{\mathrm{d} (\lambda_i - \lambda_j)}{\mathrm{d} \delta}
= \frac{\partial (\lambda_i  - \lambda_j)}{\partial \delta} + \frac{\partial (\lambda_i-\lambda_j)}{\partial \bar{V}} \frac{\partial \bar{V}}{\partial \delta}
\end{equation}
The partial derivatives with respect to $\delta$ and $\bar{V}$ are characterized in Lemmas \ref{lemma:csmob} and \ref{lemma:csmobV}, respectively. The only missing piece is thus the term $\partial \bar{V} / \partial \delta$. The proof of the next lemma leverages again the tools developed in Lemma \ref{lem:shapeder}. 
\begin{lemma}\label{lemma:Vdelta}
Under the assumptions of Theorem \ref{theo:fact}, we have that $\partial \bar{V} / \partial \delta < 0$.
\end{lemma}
\begin{proof}
Apply the implicit function theorem to the expression for $\bar{L}$ given in \eqref{eq:LV}, now viewed as a function of $\delta$ and $\bar{V}$ jointly. Since $\bar L$ is constant, we get
\begin{equation}\notag
\frac{\partial \bar{V}}{\partial \delta} = - \frac{\partial \bar{L} / \partial \delta }{\partial \bar{L} / \partial \bar{V}}
\end{equation}
That $\partial \bar{L} / \partial \bar{V}$ is strictly negative has been proved in Theorem \ref{theo:fact}, part \textit{ii}. The proof that $\partial \bar{L} / \partial \delta < 0$ is only sketched here, as it repeats the same steps. The conclusion then follows. Write
\begin{equation}\notag
\bar{L}(\delta) = \sum_{i = 1}^n f_i \int_{\Omega_i} g_i(x, \delta) \mathrm{d} x 
\end{equation}
where $f_i$ and $g_i$ are defined as in \eqref{def:fg}. Applying Lemma \ref{lem:shapeder}, we obtain: 
\begin{equation}\notag
\frac{\partial \bar{L}}{\partial \delta} = \sum_{i=1}^n f_i \int_{\Omega_i(\lambda)} \frac{\partial g_i}{\partial \delta}(x,\delta) dx + \sum_{\mathclap{\substack{i,k=1 \\ \partial \Omega_i \cap \partial \Omega_k \neq \emptyset}}}^n (f_i - f_k)\frac{\partial (\lambda_i-\lambda_k)}{\partial \delta}  \int_{\Gamma_{ik}} g_i(x,\delta)\nu_i(x)\cdot \frac{\partial \omega_{ki}}{\partial \lambda_i}(x, \lambda)d \sigma(x).
\end{equation}
It is easy to show that $\frac{\partial g_i}{\partial \delta}(x,\delta) < 0$ for all $\bar V > \max_i y_i^m$ and $x \in \Omega_i(\lambda)$. Then the first term of the above expression is strictly negative, and we focus on the second term. Although an analytic expression for the product $(f_i - f_k) \times  \partial(\lambda_i - \lambda_k)/\partial \delta $ cannot be derived, it is easy to show that whenever $y_i^m > y^m_k$, then $f_i -  f_k > 0$ and $ \partial(\lambda_i - \lambda_k)/\partial \delta <0$ (the latter being the conclusion of Lemma \ref{lemma:csmob}), and vice versa. Therefore, all the addends in the second term are negative and $\partial \bar{L} / \partial \delta < 0$.
\end{proof}
The next lemma characterizes the comparative statics of $\delta$ in a closed-economy scenario with a fixed population. 
\begin{lemma}[Effect of $\delta$ with $\bar{L}$ fixed.]\label{lemma:csmobL}
Let the assumptions of Theorem~\ref{theo:fact} hold, and suppose that $\bar{L}$ is fixed; moreover, suppose that $\alpha < 0$. 

\begin{enumerate}[$i.$]
\item Take two cities, $s_i, s_j \in S$, with $y_i^m > y^m_j$. Then, if $\alpha < 0$,
\begin{equation}\notag
\frac{\mathrm{d} (\lambda_i - \lambda_j)}{\mathrm{d} \delta}  
< 0 .
\end{equation}
\item Take a city $s_i \in S$, and let $\mathcal{N}_i = \{s_j \in S, s_j \neq s_i: \partial \Omega_i \bigcap \partial \Omega_j \neq \emptyset\} $ denote the set of its neighboring cities. If $y_i^m > \max_{s_j \in \mathcal{N}_i} y^m_j$, then: 
\begin{equation}\notag
\frac{\mathrm{d} \vert \Omega_i(\lambda)\vert}{\mathrm{d} \delta } < 0.
\end{equation}
The opposite inequality holds if $y_i^m < \min_{s_j \in \mathcal{N}_i} y^m_j$.
\end{enumerate}
By contrast, when $0 < \alpha < 1$, the effect of $\delta$ is ambiguous.
\end{lemma}
These results follow from combining Lemmas \ref{lemma:csmob}, \ref{lemma:csmobV}, and \ref{lemma:Vdelta}. First, Lemma \ref{lemma:csmob} says that the direct effect of an increase in $\delta$ favors less productive cities. Second, Lemma \ref{lemma:Vdelta} says that an increase in $\delta$ also induces a drop in the level of welfare. Finally, Lemma \ref{lemma:csmobV} says that a drop in the level of welfare goes in favor of less productive cities if and only if $\alpha < 0$. By contrast, when $0 <\alpha < 1$, a drop in the level of welfare favors more productive cities and thus the sign of the overall effect becomes ambiguous. 

\section{Other extensions and formulations}\label{sec:altfor}

\subsection{Manufacturing trade costs}\label{sec:manuftrade}

In this section, we extend the analysis presented in Section \ref{sec: results} by including trade frictions for the manufacturing good. In particular, suppose that the manufacturing good is subject to iceberg trade costs: for all $s_i, s_j \in S$, $T_{ij} > 1$ units of the good have to be produced in $s_i$ in order to deliver one unit in $s_j$. The characterization of the equilibrium combines our Proposition \ref{prop: equiv} with the universal gravity approach of \cite{allen2020}.

We focus on a simpler case with a Cobb-Douglas utility function such that $0 < \alpha < 1$ is the budget share of the manufacturing good. To introduce trade incentives in this sector, consider the Armington formulation: each city in $S$ produces a distinct variety of the manufacturing good, and consumers combine these varieties according to a CES aggregate nested into their utility function. 
That is, 
\begin{align*}\notag
u(C^m, c^a) = (C^m)^{\alpha}(c^a)^{1-\alpha}, \\
\text{ with } C^m = \left( \sum_{j = 1}^n ( c^m_j ) ^{\frac{\sigma -1}{\sigma}} \right)^{\frac{\sigma}{\sigma - 1}},
\end{align*}
and where $\sigma > 1$ is the elasticity of substitution between urban varieties.

The expression for the Voronoi weights is
\begin{equation}\notag
\lambda_i = \frac{\alpha}{\delta} \log \left( \frac{p_i}{Q_i} \right), \quad i = 1,\dots, n,
\end{equation}
where $Q_i$ is the CES price index for the composite urban good in city $s_i$: 
\begin{equation}\notag
Q_i^{1-\sigma} = \sum_{j= 1}^n T_{ji}^{1-\sigma} q_j^{1-\sigma}, \quad i = 1,\dots, n. 
\end{equation}

The excess demand for agricultural goods is: 
\begin{equation}\notag
Z^a_i(p,q) = (1-\alpha)\frac{q_i y^m_i L^M_i}{p_i} - \alpha \int_{\Omega_i(\lambda(p,q))} \frac{y^a(x)L^A(x)}{\Delta(x, s_i)} \mathrm{d} x, \quad i = 1,\dots, n.
\end{equation}
where $p$ and $q$ denote, respectively, the price vectors $(p_1, \dots, p_n)$ and $(q_1, \dots, q_n)$. The condition is similar to Equation \eqref{eq: edA} in the paper except that the manufacturing market is segmented at the urban level, and there is a distinct manufacturing price in each city. 

The excess demand for manufacturing goods in city $s_i$ is: 
\begin{equation}\notag
Z^m_i(p,q) = \alpha \sum_{j = 1}^n T^{1-\sigma}_{ij} \left[ \frac{Q_j^{\sigma - 1} }{q_i^\sigma} \left( q_j y^m_j L^M_j +  p_i \int_{\Omega_i(\lambda(p,q))} \frac{y^a(x)L^A(x)}{\Delta(x, s_i)} \mathrm{d} x \right)  \right ] - y^m_i L^M_i,
\end{equation}
for $i = 1,\dots, n$, where we have directly plugged in the well-known formula for a CES demand function. 

An equilibrium is a pair of price vectors $p$ and $q$ such that $Z^a_i(p,q)= 0$  and $Z^m_i(p,q)= 0$ for all $s_i \in S$. 
\begin{lemma} There is a unique equilibrium pair of price vectors, i.e. a pair of vectors $p^* \in \R^n_{++}$ and $q^*\in \R^n_{++}$ such that $Z^a(p,q) = 0$ and $Z^m(p, q) = 0$.  
\end{lemma}
\begin{proof}
With Cobb-Douglas preferences,
\begin{equation*} Z^a_i(p,q)= 0 \implies q_j y^m_j L^M_j  = \frac{\alpha}{1-\alpha} p_i \int_{\Omega_i(\lambda(p,q))} \frac{y^a(x)L^A(x)}{\Delta(x, s_i)} \mathrm{d} x,
\end{equation*} so that the market clearing conditions for the manufacturing sector can be rewritten as: 
\begin{equation}\notag
q_i^{\sigma} y^m_i L^M_i = \sum_{j = 1}^n T_{ij}^{1-\sigma} Q_j^{\sigma -1} q_j  y^m_j L^M_j, \quad i = 1,\dots, n.
\end{equation} 
This equation and the expression for the CES price index define a system of $2 \times n$  equations in terms of  $2 \times n$  unknowns, the vectors $q$ and $Q$. This system is an instance of the universal gravity framework of \cite{allen2020}. By their Theorem 1, it has a unique solution. Once $q$ and $Q$ are solved for, we can keep them fixed and use Proposition A.1 and Theorem 1 in our paper to find the vector $p$ that solves $Z^a(p,q) = 0$. 
\end{proof}

\subsection{Urban spillovers}\label{sec:spillover}

In this subsection, we extend the analysis presented in Section \ref{sec:factor} by including positive urban spillovers in the model with immobile labor. The characterization of the equilibrium combines our Lemma \ref{lem:shapeder} on the shape derivative with the techniques developed in \cite{allen2022b}.

We focus on a simpler case with Cobb-Douglas preferences such that $0 < \alpha < 1$ is the budget share of urban goods. In this case the welfare-equalizing price is
\begin{equation}\notag
p^{*}_i = \left( \frac{y^m_i}{\bar{V}}\right)^{\frac{1}{1-\alpha}}, \quad i = 1,\dots,n,
\end{equation}
and the Voronoi weights can be written as: 
\begin{equation}\notag
\lambda_i = \frac{1}{\delta} \frac{\alpha}{1-\alpha} \log y^m_i, \quad i = 1,\dots,n.
\end{equation}
Let $h: \R_{+} \to \R_{++}$ be a positive, increasing, and continuously differentiable  function and suppose that urban productivity in city $s_i \in S$ depends on $L^M_i$, the number of workers living in that city, as follows:
\begin{equation}\notag
y^m_i = m_i h(L^M_i), \text{ with } m_i > 0 , \quad i = 1, \dots, n,
\end{equation}
where $m_i$ is the exogenous component of city $s_i$'s productivity with bounds $0 < \ubar{m} < m_i < \bar{m} < \infty$ for all $s_i \in S$, and the function $h$ governs the strength of urban spillovers. Clearly, the Voronoi weights are now a function of  local urban population. Let us denote with 
\begin{equation}\notag 
\lambda_i(L^M_i) = \frac{1}{\delta} \frac{\alpha}{1-\alpha} \log m_i + \frac{1}{\delta} \frac{\alpha}{1-\alpha} \log h(L^M_i), \quad i = 1, \dots, n,
\end{equation}
the Voronoi weight in city $s_i$ when the local urban population is $L^M_i$.

Given the Voronoi weights, we can proceed as in Section \ref{sec:factor} to solve for the urban population in $s_i \in S$ from the market clearing condition for agricultural goods, $Z(p^*) = 0$. With Cobb-Douglas preferences, we obtain: 
\begin{align*}
L^M_i &= \frac{\alpha}{1-\alpha} \left(\frac{p^{*}_i}{y^m_i} \right)^{\frac{1}{1-\beta}} \int_{\Omega_i(\lambda(L^M))} \left( \frac{a(x)}{\Delta(x, s_i)}\right)^{\frac{1}{1-\beta}}  \\
& = \frac{\alpha}{1-\alpha} \frac{1}{\bar{V}^{\frac{1}{(1-\alpha)(1-\beta)}}}   (m_i h(L^m_i))^{\frac{\alpha}{(1-\alpha)(1-\beta)}} \int_{\Omega_i(\lambda(L^M))} \left( \frac{a(x)}{\Delta(x, s_i)}\right)^{\frac{1}{1-\beta}} \mathrm{d} x
\end{align*}
for $i = 1,\dots, n$. The key difference with respect to equation \eqref{eq:pop2}, apart from the different functional forms, is that $L^M$ also appears on the right-hand side of the equation, both directly via urban productivity, and indirectly via the Voronoi weights. Let $\Psi_i: \R^n_{+} \to \R_{+}$ denote the right-hand side of the equation, and let the continuous vector-valued function $\Psi: \R^n_{+} \to \R^n_{+}$ be defined from $(\Psi_1, \dots, \Psi_n)$. To find the equilibrium distribution of urban population on $S$, we have to solve the fixed point problem: 
\begin{equation}\notag
L^M = \Psi(L^M). 
\end{equation}
We approach this problem with some of the techniques developed in Allen, Arkolakis, and Lee (2022), in combination with our Lemma B.1 on the shape derivative. Because the domain and range of $\Psi$ are weakly positive (as some cities may remain inhabited), we will seek bounds on its derivatives, as opposed to its elasticities (see Remark 3 in their paper). The following lemma provides a sufficient condition for the existence of a unique equilibrium under some restrictions on the spillover function $h$. It focuses on the closed economy scenario where $\bar{L}$ is fixed and the welfare scalar is a free variable (see Equation \eqref{eq:LV}), but a similar condition applies for the case of $\bar{V}$ fixed. 
\begin{lemma} 
Let the assumptions of Theorem \ref{theorem: existence} hold, except that $y^m_i = m_i h(L^M_i)$ for $i=1,\dots,n$, and with the additional condition that each distance function $d_s:X \to \R_+$, for $s \in S$, is of class $C^2$ and $a\colon X \to \R_{++}$ is of class $C^1$. Moreover, suppose that the function $h$ is bounded above and below by $0 < \ubar{h} < \bar{h} < \infty$, and there exists $\gamma > 0 $ such that $\frac{h'(L)}{h(L)} < \gamma$ for all $L \in \R_{+}$. Then, if 
\begin{equation}\notag
\frac{\alpha^2 \gamma}{1-\alpha} \left(\frac{1}{1-\beta} +  \frac{1}{\delta} K \right) \left(\frac{\bar{a}}{\ubar{a}}\right)^{\frac{1}{1-\beta}}\left(\frac{ \bar{m}  \bar{h} }{\ubar{m} \ubar{h}}  \right) ^{\frac{\alpha}{(1-\alpha(1-\beta)}}   e^{\frac{\delta}{1-\beta}d_{max}} \frac{n}{\vert X \vert}  \bar{L} < 1, 
\end{equation}
where 
\begin{itemize}
\item $d_{max} = \max_{x \in X, s_i \in S} d_i(x)$
\item $K \geq \max_{s_i \in S}\int_{\partial \Omega_i}  \nu_i(x)\cdot \frac{\partial \omega_{i}}{\partial \lambda_i}(x, \lambda)d \sigma(x) $
\end{itemize}
there exists a unique equilibrium. 
\end{lemma}
Before presenting the proof, it is useful to discuss some of the properties of this condition. First, it is always satisfied for $\gamma \to 0$. Given the a priori bounds on $h$, $\gamma \to 0$ implies $h'(L) \to 0$ for all values of $L$, and the model reduces to the one without urban spillovers in Section \ref{sec:factor}. Consistent with our Theorem \ref{theo:fact}, the equilibrium is always unique irrespective of the geography and of the value of the other parameters.  Second, the term $1/(1-\beta)$ parametrizes the (inverse) strength of the congestion force in the economy, and it decreases as agricultural production becomes more land intensive. This effect arises from the interior of the Voronoi regions and typically appears in this class of models. By contrast, the term $K/\delta$ represents the size of the border readjustments due to changes in the Voronoi weights, and is specific to our setting. Intuitively, this term becomes smaller as the size of transport costs $\delta$ increases, because trading choices become less responsive to variation in the endogenous variables. Third, the remaining terms capture the characteristics of the underlying geography. In particular, the sufficient condition becomes less stringent as the average density of potential urban sites  $n / \vert X \vert$ decreases. This conveys the intuitive idea that a unique equilibrium is more difficult to achieve when urban sites are located in close proximity to one another. If two urban sites are very close to each other, relocating a small number of urban workers to either city may be sufficient for the other to be abandoned. As $n / \vert X \vert \to \infty$, other things equal, the sufficient condition is never satisfied. Finally, a strictly positive lower bound to $h$ is essential for the sufficient condition to have content. As a result, the standard formulation in the literature, $h(L) = L^{\gamma_1}$ for some $\gamma_1 > 0$, does not work in this setting. Alternatively, one could consider other formulations that allow urban sites to retain some productivity even when no one resides there, e.g., $h(L) = 1 + L^{\gamma_1}$.

\begin{proof}
As a first step, we make sure that the the integrand is continuous across the borders, by making the term $( p^*_i )^{\alpha} / \Delta(x, s_i)$ appear inside the integral. This gives: 
\begin{align*}
\Psi_i(L^M)  &= f \int_{\Omega_i(\lambda(L^M))} g_i(x, L^M_i) \mathrm{d} x, \\
\text{ with } f &= \frac{\alpha}{1-\alpha} \frac{1}{\bar{V}^{\frac{1}{(1-\alpha)(1-\beta)}}}, \quad \text{ and }  \\ 
g_i(x, L^M_i) &= \left( \frac{a(x)}{\Delta(x, s_i)}\right)^{\frac{1}{1-\beta}} m_i^{\frac{\alpha}{(1-\alpha)(1-\beta)}} ( h(L^m_i))^{\frac{\alpha}{(1-\alpha)(1-\beta)}}.
\end{align*}
The next steps follow the logic of the proof of Theorem 1 in \cite{allen2022b}. Given any $L^M$ and $L^{M'}$, according to the mean value theorem, and for each $i=1, \dots, n$, there exists $\hat{L}^M = t_i L^M + (1-t_i) L^{M'}$, with $0 < t_i < 1$, such that: 
\begin{equation}\notag
\Psi_i(L^{M'}) - \Psi_i(L^M) = \sum_{j = 1}^n \frac{\partial \Psi_i(\hat{L}^M) }{\partial L^M_j } (L^M_j - L^{M'}_j)
\end{equation}
From this expression, we compute the following bound:
\begin{align*}
\bigg\vert \Psi_i(L^{M}) - \Psi_i(L^{M'}) \bigg\vert &= \bigg\vert \sum_{j = 1}^n \frac{\partial \Psi_i(\hat{L}^M) }{\partial L^M_j } (L_j^M - L_j^{M'}) \bigg\vert \\
& \leq  \sum_{j = 1}^n \bigg\vert \frac{\partial \Psi_i(\hat{L}^M) }{\partial L^M_j } (L_j^M - L_j^{M'}) \bigg\vert \\
& \leq  \sum_{j = 1}^n \bigg\vert \frac{\partial \Psi_i(\hat{L}^M) }{\partial L^M_j }\bigg\vert \bigg\vert (L_j^M - L_j^{M'}) \bigg\vert \\ 
& \leq  \left( \sum_{j=1}^n \max \bigg\vert \frac{\partial \Psi_i(\hat{L}^M) }{\partial L^M_j }\bigg\vert \right) \max \bigg\vert (L_j^M - L_j^{M'}) \bigg\vert.
\end{align*}
Now consider the sup norm $\left\lVert\cdot\right\lVert$  on $\R^n$. Since the previous inequality holds for all $i=1,\dots, n$, we have: 
\begin{align*}
\left\lVert \Psi(L^{M}) - \Psi(L^{M'})\right\lVert & \leq A\left\lVert L^M - L^{M'} \right\lVert, \\
\text{with } \quad A & = \sum_{j=1}^n \max \bigg\vert \frac{\partial \Psi_i(\hat{L}^M) }{\partial L^M_j }\bigg\vert
\end{align*}
Next, we can use the shape derivative to compute $\partial \Psi_j / \partial L^M_i$ and estimate an upper bound on $A$. Applying Lemma B.1, we have
\begin{align*}
\bigg\vert \frac{\partial \Psi_i(L^M)}{\partial L^M_i} \bigg\vert & =  f \left\{ \int_{\Omega_i} \frac{\partial g_i(x, L^M_i)}{\partial L^M_i} \mathrm{d} x + \frac{\partial \lambda_i}{\partial L^M_i} \sum_{j \in \mathcal{N}_i} \int_{\Gamma_{ij}} g_i(x,L^M_i) \nu_i(x)\cdot \frac{\partial \omega_{ij}}{\partial \lambda_i}(x, \lambda)d \sigma(x) \right\} \\
& = f \bigg\{ \frac{\alpha}{(1-\alpha)(1-\beta)}\frac{h'(L^M_i)}{h(L^M_i)} \int_{\Omega_i} g_i(x, L^M_i) \mathrm{d} x \\ & + \frac{1}{\delta} \frac{\alpha}{1-\alpha} \frac{h'(L^M_i)}{h(L^M_i)}  \sum_{j \in \mathcal{N}_i} \int_{\Gamma_{ij}} g_i(x,L^M_i) \nu_i(x)\cdot \frac{\partial \omega_{ij}}{\partial \lambda_i}(x, \lambda) \mathrm{d} \sigma(x) \bigg\} \\ 
& \leq  f \frac{\alpha\gamma}{(1-\alpha)} \left\{\frac{1}{1-\beta} +  \frac{1}{\delta}\int_{\Omega_i} \nu_i(x)\cdot \frac{\partial \omega_{i}}{\partial \lambda_i}(x, \lambda) \mathrm{d} \sigma(x)\right\} \bar{a}^{\frac{1}{1-\beta}}\left( \bar{m}  \bar{h} \right) ^{\frac{\alpha}{(1-\alpha(1-\beta)}}  ,
\end{align*}
where we have used the a priori bounds on $a$, $m$, and $h$, and the assumption that $h'(L^M_i)/h(L^M_i) < \gamma$. Let $K$ denote the uniform bound on $\vert \int_{\Omega_i} \nu_i(x)\cdot \frac{\partial \omega_{i}}{\partial \lambda_i}(x, \lambda) \mathrm{d} \sigma(x)\vert $, which exists because: 1) $\max_{s_i \in S} \vert\nu_i(x)\cdot \frac{\partial \omega_{i}}{\partial \lambda_i}(x, \lambda)\vert$ is bounded, since the distance functions are of class $\mathcal{C}^2$; 2) $\max_{s_j \in S} \vert \partial \Omega_j (\lambda) \vert $ is bounded, since the function mapping $\lambda$ to $\vert \partial \Omega_j (\lambda) \vert$ is continuous, and since for $\lambda_j$ sufficiently large $\Omega_j = X$. Thus we obtain
\begin{equation*}
\bigg\vert \frac{\partial \Psi_i(L^M)}{\partial L^M_i} \bigg\vert \leq  f \frac{\alpha\gamma}{(1-\alpha)} \left\{\frac{1}{1-\beta} +  \frac{K}{\delta}\right\} \bar{a}^{\frac{1}{1-\beta}}\left( \bar{m}  \bar{h} \right) ^{\frac{\alpha}{(1-\alpha(1-\beta)}}  ,
\end{equation*}
For $s_i, s_j \in S$, with $s_i \neq s_j$, we have
\begin{align*}
\bigg\vert \frac{\partial \Psi_i(L^M)}{\partial L^M_j} \bigg\vert & =  f \frac{\partial \lambda_j}{\partial L^M_j} \int_{\Gamma_{ik}} g_i(x,L^M_i) \nu_i(x)\cdot \frac{\partial \omega_{ij}}{\partial \lambda_i}(x, \lambda)d \sigma(x) 
\end{align*}
for $\partial \Omega_i \cap \partial \Omega_j \neq \emptyset$, and zero otherwise, so that the previous bound also applies. 

Since we are keeping $\bar{L}$ fixed and the welfare scalar appears into $f$, we also need to provide an upper bound for this term. With Cobb-Douglas preferences, welfare equalization between urban and rural workers together with market clearing on agricultural markets imply $L^M_i = \alpha (L^M_i + \int_{\Omega_i} L^A(x) \mathrm{d} x )$ for all $i=1,\dots,n$. Therefore
\begin{equation}\notag
\sum_{i = 1}^n L^M_i = \alpha \bar{L} \iff f = \frac{ \alpha \bar{L} }{\sum_{i=1}^n \int_{\Omega_i(\lambda(L^M))} g_i(x, L^M_i) \mathrm{d} x }. 
\end{equation}
Another application of Lemma B.1 gives: 
\begin{equation*}
f \leq \frac{\alpha \bar{L}}{ \ubar{a}^{\frac{1}{1-\beta}}\left( \ubar{m}  \ubar{h} \right) ^{\frac{\alpha}{(1-\alpha(1-\beta)}} e^{-\frac{\delta}{1-\beta}d_{max}} \vert X \vert},
\end{equation*}
with $d_{max} = \max_{x \in X, s_i \in S} d_i(x)$. Combining the previous results, we obtain: 
\begin{align*}
A & = \sum_{j=1}^n \max \bigg\vert \frac{\partial \Psi_i(\hat{L}^M) }{\partial L^M_j }\bigg\vert \\
&\leq  \frac{\alpha^2 \gamma}{1-\alpha} \left\{\frac{1}{1-\beta} +  \frac{1}{\delta} K \right\} \left(\frac{\bar{a}}{\ubar{a}}\right)^{\frac{1}{1-\beta}}\left(\frac{ \bar{m}  \bar{h} }{\ubar{m} \ubar{h}}  \right) ^{\frac{\alpha}{(1-\alpha(1-\beta)}}   e^{\frac{\delta}{1-\beta}d_{max}}  \frac{n}{\vert X \vert}  \bar{L} 
\end{align*}
Since $\Psi$ is a mapping from a complete metric space to itself, with $A < 1$ the result obtains from the application of the Contraction Mapping Theorem. 
\end{proof}

\subsection{CES Production}\label{app:CESprod}
Our analysis holds if the CES structure is imposed on the production side, rather than on preferences. Specifically, suppose that there is only one consumption good, produced in cities with a CES technology that combines urban labor and intermediate agricultural goods: 
\begin{equation}\notag
Y_i^m = y_i^m \left((L^M_i)^{\alpha} + (K^a_i)^{\alpha}\right)^{\frac{1}{\alpha}}, \quad i=i,\dots,n,
\end{equation}
where $K^a_i$ denotes the demand for agricultural inputs in city $s_i \in S$. From the first-order condition of the profit-maximization problem, this can be expressed as 
\begin{equation}\label{eq:cesproddemand}
K^a_i = \frac{L^M_i}{ \left[ \left(\frac{p_i}{y_i^m}\right)^{\frac{\alpha}{1-\alpha}} -1 \right]^{\frac{1}{\alpha}}},
\end{equation}
where we have used the price-equals-unit-cost condition to write the urban wage as
\begin{equation}\label{eq:cesprodwage}
w_i = \left[ (y_i^m)^{\frac{\alpha}{\alpha-1}} -p_i^{\frac{\alpha}{\alpha-1}} \right]^{\frac{\alpha - 1}{\alpha}}, \quad i = 1\dots n,  
\end{equation}
The indirect utility of a farmer in $x \in \X$ who trades with city $s_i \in S$ is
\begin{equation}\notag
V(x, s_i) = \frac{p_i y^a(x)}{\Delta(x, s_i)},
\end{equation}
after normalizing the price of the consumption good to one. Therefore, the Voronoi weights take the following form: 
\begin{equation}\notag
\lambda_i = \frac{1}{\delta} \log p_i, \quad i=1,\dots, n.
\end{equation}

\paragraph{Immobile labor.} Without factor mobility, the vector of Voronoi weights determines the supply of agricultural goods to each urban market. At the same time, the demand for agricultural goods was expressed as a function of prices in \eqref{eq:cesproddemand}.
Therefore, analogously to Definition \ref{def:eq}, we can close the model by defining an equilibrium price vector as the vector of prices that solves
\begin{equation}\notag
K^a_i = \int_{\Omega_i(\lambda)} \frac{y^a(x) L^A(x)}{\Delta(x, s_i)} \mathrm{d} x,
\end{equation}
i.e., that clears the market for agricultural goods in all cities. Although we skip the computations in the interest of space, it is easy to check that Proposition \ref{prop: equiv} holds in this setting, after defining the functional $\F$ exactly as in \eqref{eq: costfunction}. It can also be verified that the function $\F$ is concave, so that Theorem \ref{theorem: existence} applies. 

\paragraph{Mobile labor.} With labor mobility, the solution of the model follows the same steps as in Section \ref{sec:factor}. The spatial equilibrium condition for urban workers takes the simple form $w_i = \bar{V}, \quad i=1\dots m$, which allows us to solve from the welfare-equalization price from \eqref{eq:cesprodwage} as 
\begin{equation}\notag
p^*_i = \left( (y_i^m)^{\frac{\alpha}{\alpha-1}} -\bar{V}^{\frac{\alpha}{\alpha-1}} \right)^{\frac{\alpha - 1}{\alpha}}, \quad i = 1\dots n, 
\end{equation}
Substituting this price into the indirect utility of farmers, we obtain the following expression for the Voronoi weights: 
\begin{equation}\notag
\lambda_i = -\frac{1}{\delta} \frac{1-\alpha}{\alpha} \log\left( (y_i^m)^{\frac{\alpha}{\alpha-1}} -\bar{V}^{\frac{\alpha}{\alpha-1}} \right), \quad i = 1\dots n
\end{equation}
The weights can be used to compute the Voronoi tessellation and the supply of agricultural inputs to each urban market $s_i \in S$. Finally, the market equilibrium condition for agricultural inputs pins down the equilibrium urban population in all cities as follows: 
\begin{equation}\notag
L^M_i = \left( \frac{p_i}{\bar{V}} \right)^{\frac{1}{1-\alpha} + \frac{\beta}{1-\beta}} \int_{\Omega_i(\lambda)} \left( \frac{a(x)}{\Delta(x, s_i)} \right)^{\frac{1}{1-\beta}} \mathrm{d} x, \quad i = 1 \dots n
\end{equation}

\subsection{Business Districts and Urban Structure}\label{sec:business} Suppose that the set $X$ represents a metropolitan area and the set $S$ the location of its business districts. Firms locate in the business districts and produce a freely traded consumption good using only labor at constant returns to scale. Residents live in the metropolitan area and pay a commuting cost to travel to the business district of their choice, where they supply inelastically one unit of labor. Business locations have different levels of labor productivity $z_i > 0$, $i=1, \dots, n$, while residential locations differ in terms of an unpriced local amenity $A: X \to \R_+$ and in terms of their inelastic housing supply $H: X \to \R_+$. The housing stock belongs to absentee landlords who live outside the city. All markets are competitive. 

Residents order consumption baskets according to a Cobb-Douglas utility function defined over the consumption of housing and the traded good. Under these assumptions, the indirect utility of a resident living in $x \in X$ and employed at business district $s_i \in S$ is 
\begin{equation}\notag
V(x, s_i ) = \frac{A(x)z_i}{\Delta(x, s_i) r(x)^{1-\theta}}, \quad 0 < \theta < 1
\end{equation}
where $\theta$ is the housing budget share, $r: X \to \R_+$ is the price of one unit of housing, and the price of the traded good has been normalized to one.  For any distribution of residents $L: X \to \R_{+}$, the Voronoi weight driving their commuting decisions is
\begin{equation}\notag
\lambda_i = \frac{1}{\delta} \log z_i.
\end{equation}
Finally, the market-clearing condition for the land market 
\begin{equation}\notag
r(x) H(x) = (1 - \theta) z_i L(x),
\end{equation}
pins down the equilibrium rent function. 

\paragraph{Immobile labor.} In this simple case where wages are exogenous, the equilibrium with immobile labor is trivial, as $\lambda$ immediately determines the allocation of workers to workplaces, firms make zero profits, and rents are also a function of exogenous variables only. 
\paragraph{Mobile labor.} Here, it is clear that if two business districts $s_i, s_j \in S$ have contiguous commuting areas, then lower commuting costs will shift the border in favor of the more-productive district. The most productive district in the metropolitan area will always gain from lower commuting costs, whereas the least productive district will always lose. For a given set $S$, a vector of labor productivities $\{z_i\}_{s_i \in S}$, and a value of the commuting cost parameter $\delta$, some potential business districts may not be able to attract commuters and will therefore remain empty.

\subsection{Home Consumption for Farmers}\label{app:homeconsumption}

In our main analysis, we assumed that farmers consume agricultural goods at the trading location, rather than at the production location. In the case with immobile labor, this assumption allows us to use our approach based on Proposition \ref{prop: equiv}, because it implies that the marginal value of wealth is constant inside a Voronoi region. For this reason, Proposition \ref{prop: equiv} also holds without this assumption in a number of relevant settings, for instance (1) when distance costs are paid in terms of utility rather than in terms of farm goods; and (2) when the CES structure is placed on production, and farm goods are inputs in production of urban goods (see Appendix \ref{sec:altfor}). 

For the case analyzed in the main text, we now develop a scenario where consumption takes place at the production location (i.e., in $x \in X$ for farmers). We can prove that (1) the model delivers a Voronoi tesellation, (2) an equilibrium always exists, and (3) a sufficient condition for the equilibrium to be unique is $0 < \alpha < 1$.

Suppose, then, farmers consume agricultural goods at their location in $X$ and carry only the \textit{surplus} to an urban market for sale. As a consequence, the iceberg shipping cost is incurred only on the latter. Therefore, the budget constraint for a generic farmer producing output $y^a$ and incurring a trade cost $\Delta$ is written as
\begin{equation*}\notag
qc^m \leq  \frac{p(y^a - c^a)}{\Delta}  \iff 
qc^m + \frac{p}{\Delta}c^a \leq \frac{p}{\Delta}y^a
\end{equation*}
The indirect utility of a farmer in $x \in X$ trading with city $s_i \in S$ is
\begin{align*}
V(p_i, \omega(x, s_i)) &= \left( q^{\frac{\alpha}{\alpha-1}} + \left(\frac{p_i}{\Delta(x, s_i)}\right)^{\frac{\alpha}{\alpha-1}}  \right)^{\frac{1-\alpha}{\alpha}} \frac{p_i y^a(x)}{\Delta(x,s_i)} \\
&= \left(1 + \left(\frac{p_i}{q\Delta(x, s_i)}\right)^{\frac{\alpha}{1-\alpha}}  \right)^{\frac{1-\alpha}{\alpha}} y^a(x).
\end{align*}
For all values of $\alpha$, it is easy to check that 
\begin{align*}
V(p_i, \omega(x, s_i)) \geq V(p_i, \omega(x, s_i)) &\iff \frac{p_i}{\Delta(x, s_i)} \geq \frac{p_j}{\Delta(x, s_j)} \\
& \iff d(x, s_i) - \frac{1}{\delta}\log p_i \leq d(x, s_j) - \frac{1}{\delta}\log p_j.
\end{align*}
Therefore the farmers' trading problem delivers an additively weighted Voronoi tessellation, where the Voronoi weights are defined from
\begin{equation}\label{eq:lambdahc}
\lambda_i = \frac{1}{\delta} \log p_i, \quad i = 1\dots n.
\end{equation}
Finally, let $p_i(x) \coloneqq p_i / \Delta(x, s_i)$, for $i =1\dots n$; then the excess demand system can be expressed as
\begin{equation}\label{eq:edAhc}
Z_i(p) =c^a(q, p_i, \omega_i)L^M_i+ \int_{\Omega_i(\lambda(p))}\left[ \frac{c^a\left(q,p_i(x),\omega(x,s_i)\right)}{\Delta(x,s_i)} - \frac{y^a(x)}{\Delta(x,s_i)}\right] L^A(x)dx, 
\end{equation} 
for $i =1,\dots,n$ and
\begin{equation}\label{eq:edMhc}
Z_{n+1}(p) = \sum_{i =1}^n \left[\left(c^m(q, p_i, \omega_i)  -  y_i^m\right) L^M_i+ \int_{\Omega_i(\lambda(p))} c^m\left(q,p_i(x),\omega(x,s_i)\right)L^A(x) dx \right].
\end{equation}

\paragraph{Immobile labor.} We show that an equilibrium exists and that it is unique for $0 < \alpha < 1$. We state these results in two separate lemmas, then we provide the proofs. The next lemma deals with equilibrium existence, and its proof follows the standard arguments in \cite[Proposition 17.C.1]{mascolell1995}.
\begin{lemma}\label{lemma:existence}
Let the excess demand system be defined from \eqref{eq:edAhc} and \eqref{eq:edMhc}. Then there exists a price vector $p^*$ such that $Z(p^*)=0$. 
\end{lemma}
To prove uniqueness, we use the connected substitutes condition of \cite{berry2013}. This condition allows for some cross-price derivatives to be zero and is therefore weaker than the gross substitute condition of \cite{mascolell1995}, which fails in our setting. More precisely, let $Z: \mathcal{P} \subseteq  \R^{n} \to \R^{n}$ be a generic demand system, and define 
\begin{equation}\notag
Z_0(p) = 1 - \sum_{i = 1}^n Z_i(p)
\end{equation}
for a fictional city $s_0$. Theorem 1 in \cite{berry2013} shows that $Z$ is invertible if the following assumptions hold: (1) $\mathcal{P}$ is a Cartesian product; (2) $Z_i(p)$ is weakly decreasing in $p_k$ for $i = 0,\dots, n$ and all $k = 1,\dots ,n$, $k \neq i$; and (3) given any subset $S'$ of $S$, there exists an $s_j$ in $S'$ and an $s_i$ in $(S \setminus S')\cup \{s_o\} $ such that $Z_i$ is strictly decreasing in $p_j$.

We apply their Theorem 1 to the (normalized) excess demand system $$\bar{Z}(\bar{p}) = \{Z_1(\bar{p}), \dots Z_n(\bar{p})\},$$ where $\bar{p} = (p_1, \dots, p_n)$ with $q$ normalized to one. 
\begin{lemma}\label{lemma:uniqueness}
Let $\bar{Z}(\bar{p})$ be a normalized excess demand system such that $Z_i(\bar{p})$ is defined via \eqref{eq:edAhc} for $i=1,\dots,n$, and suppose that $\bar{p}^{*}, \bar{p}^{**}$ are two price vectors such that $\bar{Z}(\bar{p}^*)=\bar{Z}(\bar{p}^{**}) = 0$. If $0 < \alpha < 1$, then $\bar{p}^{*} = \bar{p}^{**}.$
\end{lemma}

\begin{proof}[Proof of Lemma \ref{lemma:existence}]
Let $Z(p) = \{Z_1(p),\dots,Z_n(p), Z_{n+1}(p)\}$ be the vector of excess demands, such that $Z_i$ is given by (\ref{eq:edAhc}) for $i=1\dots n$ and $Z_{n+1}$ is given by (\ref{eq:edMhc}). We show that $Z(p)$ satisfies the conditions of \cite[Proposition 17.2.B]{mascolell1995}, namely 
\begin{enumerate}[$i.$]
\item $Z(p)$ is continuous. 
\item $Z(p)$ is homogenous of degree zero. 
\item $p'Z(p) = 0$ (Walras' law). 
\item there is a $L>0$ such that $\text{min}\{Z(p)\} > -L$ for all price vectors $p$. 
\item if $p^k \rightarrow p$, where $p \neq 0$ and some element of $p$ is zero, then $$\text{max}\{Z_1(p^k),...,Z_{n}(p^k),Z_{n+1}(p^k)\}\rightarrow \infty.$$
\end{enumerate}
If these conditions are satisfied, then $Z(p^*) = 0$ has a solution by \cite[Proposition 17.C.1]{mascolell1995}. 

Property $\textit{i.}$ is satisfied because each $Z_i$ is a composition of continuous functions. 

To prove property \textit{ii.}, multiply all prices by a constant $c$. The Voronoi weights become
\begin{equation}\notag
 \frac{1}{\delta} \log cp_i = \lambda_i + \frac{1}{\delta} \log c, \quad i = 1\dots n.
\end{equation}
Since the tessellation does not change when the same constant is added to all weights (see Section \ref{sub:voronoi}), we can still express the Voronoi regions as functions of $\lambda$.  Consumption demands $c^a$ and $c^m$ are also homogeneous of degree zero. Therefore $Z_i(cp) = Z_i(p)$ for all $i =1 \dots n$, and  $Z_{n+1}(cp) = Z_{n+1}(p)$. 

To prove property \textit{iii.}, use \eqref{eq: edA} and \eqref{eq: edM} to write
\begin{align*}
&\underset{i =1\dots n}{\sum} p_i Z_i(p) + qZ_{n+1}(p)  \\
&\quad= \sum_{i=1\dots n} p_i c^a(q,p_i,\omega_i)L^M_i + \sum_{i=1\dots n} \int_{\Omega_i(\lambda)} p_i \frac{ c^a(q,p_i,\omega(x,s_i))}{\Delta(x,s_i)}L^A(x) \mathrm{d} x \\
&\qquad-  \sum_{i=1\dots n} \int_{\Omega_i(\lambda)} \frac{p_i y^a(x)}{\Delta(x,s_i)} L^A(x) \mathrm{d} x +
\sum_{i=1\dots n} q c^m(q,p_i,\omega_i)L^M_i & \\
&\qquad - \sum_{i=1\dots n} q y_i^m L^M_i + \sum_{i=1\dots n} \int_{\Omega_i(\lambda)} q c^m(q,p_i(x),\omega(x,s_i))L^A(x) \mathrm{d} x \\
&\quad = \sum_{i=1 \dots n} \bigg( p_i c^a(q,p_i,\omega_i) +
 q c^m(q,p_i, \omega_i) - q y_i^m \bigg) L^M_i  + \\
  &\qquad \sum_{i=1\dots n} \int_{\Omega_i(\lambda)} \bigg(p_i(x) c^a(q,p_i(x),\omega(x,s_i)) +  q c^m(q,p_i(x),\omega(x,s_i)) -  \frac{p_i y^a(x)}{\Delta(x,s_i)}\bigg) L^A(x) \mathrm{d} x.
 \end{align*}
In the last line, all the terms in parentheses are zero because the budget constraint in \eqref{eq: consumptionproblem} binds at the optimum. 

To prove property \textit{iv.}, note that \eqref{eq: edA} implies 
\begin{equation}\notag
Z_i(p)  \geq - \int_{\Omega_i(\lambda)}\frac{y^a(x)L^a(x)}{\Delta(x, s_i) } \mathrm{d}x \geq - \int_{\Omega_i(\lambda)} y^a(x)L^a(x) \mathrm{d}x  \geq - \int_{X} y^a(x)L^a(x) \mathrm{d}x ,
\end{equation}
for all $i=1\dots n$, whereas \eqref{eq: edM} implies 
\begin{equation}\notag
Z_{n+1}(p)  \geq - \sum_i y_i^m L^M_i. 
\end{equation}
Therefore the property is satisfied for any $L \geq \max(\sum_i y_i^m L^M_i, \int_{X} y^a(x)L^a(x) \mathrm{d}x) $.

Finally, to prove property \textit{v.}, substitute the expressions for the consumption demands into \eqref{eq: edA} and \eqref{eq: edM}, and rewrite them, after some manipulations, as follows:
\begin{align*}
Z_i(p) &= \frac{1}{ 1+ \left(\frac{q}{p_i}\right)^{\frac{\alpha}{\alpha-1}} } \frac{q}{p_i} y_i^m L^M_i -  \int_{\Omega_i(\lambda)} \frac{1}{1 + \left(\frac{p_i(x)}{q}\right)^{\frac{\alpha}{\alpha-1}}}  \frac{y^a(x)L^A(x)}{\Delta(x,s_i)} \mathrm{d} x, \quad i = 1 \dots n,\\
Z_{n+1}(p) &= \sum_{i = 1\dots n}  \int_{\Omega_i(\lambda)}  \frac{1}{1 + \left(\frac{p_i(x)}{q} \right)^{\frac{\alpha}{\alpha-1}} } \frac{p_i(x)}{q} \frac{y^a(x)L^A(x)}{\Delta(x, s_i)} \mathrm{d} x - \sum_{i = 1\dots n}  \frac{1}{ 1+ \left(\frac{q}{p_i}\right)^{\frac{\alpha}{\alpha-1}} }  y_i^m L^M_i. 
\end{align*}
Therefore we obtain the following bounds:
\begin{align*}
 Z_i(p)  &\geq \frac{1}{ 1+ \left(\frac{q}{p_i}\right)^{\frac{\alpha}{\alpha-1}} } \frac{q}{p_i} y_i^m L^M_i  - \int_{X} y^a(x)L^A(x) \mathrm{d} x, \\
 Z_{n+1}(p) &\geq \sum_{i = 1\dots n} \int_{\Omega_i(\lambda)}  
 \frac{1}{1 + \left(\frac{p_i(x)}{q} \right)^{\frac{\alpha}{\alpha-1}} } \frac{p_i(x)}{q} \frac{y^a(x)L^A(x)}{\Delta(x, s_i)} \mathrm{d} x - \sum_{i = 1\dots n}   y_i^m L^M_i, 
\end{align*}
for $i=1 \dots n$. First, suppose that $q^k$ tends to a positive number, whereas $p^k_i$ tends to zero for some $i =1\dots n$. Then $Z_i(p^k) \to \infty$ for all $i$ such that $p^k_i \to 0$, because $\alpha / (\alpha - 1) < 1$ for $ \alpha < 1$ . Second, suppose that $q^k$ tends to zero and at least one $p^k_i$ tends to a positive number. Note that at least one Voronoi region is nonempty, and for this region the agricultural price must be positive, so that $Z_{n+1}(p^k) \rightarrow \infty$. In either case, $\text{max}\{Z_1(p^k),...,Z_{n}(p^k),Z_{n+1}(p^k)\}\rightarrow \infty$.
\end{proof}

The proof of Lemma \ref{lemma:existence} shows that the excess demand system satisfies Walras' law and homogeneity of degree zero. Therefore, we can restrict our attention to the normalized excess demand system $\bar{Z}(\bar{p}) = (Z_1(\bar p), Z_2(\bar p),...,Z_n(\bar p)),$
where $\bar p  = (p_1, p_2, \dots, p_n)$ with $q$ normalized to one.

\begin{proof}[Proof of Lemma \ref{lemma:uniqueness}]
We verify that the normalized excess demand system $\bar{Z}(\bar{p})$ satisfies the three assumptions of \cite[Theorem 1]{berry2013}. Define
\begin{equation}\notag
Z_0(p) = 1 - \sum_{i = 1}^n Z_i(p).
\end{equation}
First, the domain of $\bar{Z}$ is $\R_{++}^n$, and therefore it is a Cartesian product. Second, we need to prove that $Z_i(\bar{p})$ is weakly increasing in $p_k$ for all $i = 0,\dots,n$ $i\neq k$.  For $i =1\dots n$, this was shown in Proposition \ref{prop:gsz}. For $i = 0$, we obtain
\begin{align*}
\frac{\partial Z_0(p) }{\partial p_k } = - \frac{\partial c^a(p_k, \omega_k)}{\partial p_k}L^M_k - \int_{\Omega_k(\lambda)} \frac{\partial c^a(p_k, \omega(x, s_k))}{\partial p_k}L^A(x) \mathrm{d} x.
\end{align*}
Note that the impact of $p_k$ on $Z_k$ via changes in the Voronoi tessellation does not appear in the above expression, because the effect on $\Omega_k$ is offset by opposite changes in the Voronoi regions of $s_k$'s neighbors (see also the proof of Theorem \ref{theo:fact}). While $\partial c^a(p_k, \omega_k) /  \partial p_k$ is always negative, $\partial c^a(p_k, \omega(x, s_k))/\partial p_k$ is negative for $0 < \alpha < 1$ (because agricultural workers are also subject to an income effect). Therefore, $0 < \alpha < 1$ is a sufficient condition for  $\frac{\partial Z_0(p) }{\partial p_k }$ to be strictly positive. Finally, this also proves that the third condition is satisfied, because it can be checked on $Z_0$ for every $p_i$, $i=1, \dots, n$.
\end{proof}

\paragraph{Mobile labor.} With mobile labor, the existence of a unique equilibrium can be shown by applying the same steps as in Section \ref{sec:factor}.

\section{Data}\label{app:data}

\paragraph{Cities and Cantons.}

We obtained the shapefile for the 26 Swiss cantons from the Swiss Federal Office of Topograhy (\url{swisstopo.admin.ch}). The original shapefile contains 51 polygons, since some of the cantons are disconnected sets, with small enclaves inside the bordering cantons. We attribute these enclaves to the territory of the surrounding canton. Furthermore, we make two simplifications to the canton borders: first, we merge \textit{Basel-Stadt} (Basel-City) and \textit{Basel-Landschaft} (Basel-Countryside) into a single ``Basel'' canton; second, we merge \textit{Appenzell Innerrhoden} and \textit{Appenzell Ausserrhoden} with \textit{St. Gallen}. We manually georeferenced the cantonal administrative capitals. Our final shapefile contains 23 polygons, as depicted in Figure \ref{fig:cantons}.

\begin{figure}[!htb]
\includegraphics[width = \textwidth]{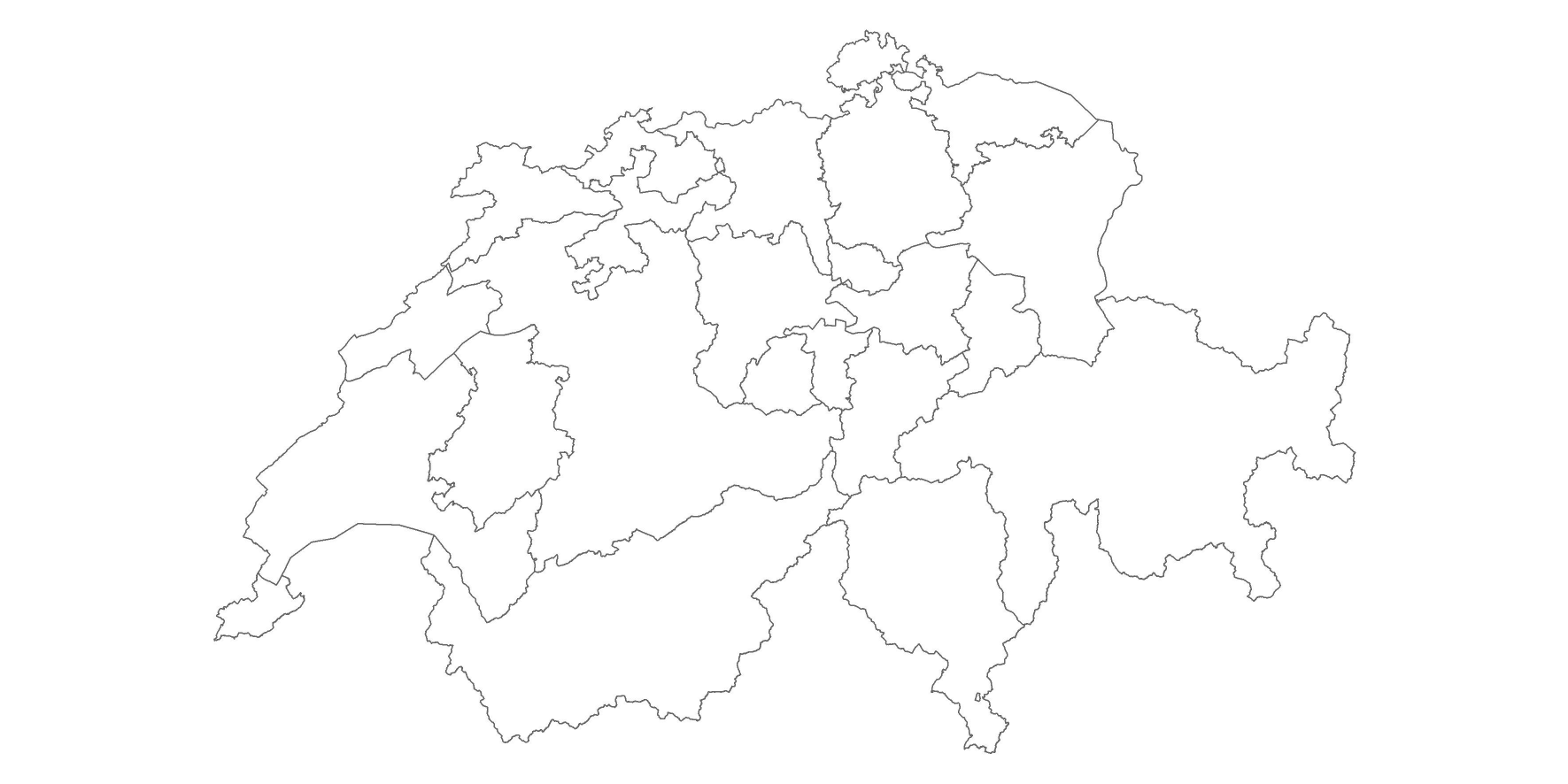}
\caption{Swiss Cantons After Cleaning}
\label{fig:cantons}
\end{figure}

\paragraph{Other data.} Gross value added by industry and canton comes from the Federal Statistical Office (\url{https://www.bfs.admin.ch/}). We use 2011 data. The Caloric Suitability Index is available at \url{ https://ozak.github.io/Caloric-Suitability-Index/}. We obtained the elevation data from Copernicus, the European Union's Earth observation program, at \url{https://land.copernicus.eu/imagery-in-situ/eu-dem/eu-dem-v1.1}. Geographic data on the hydrographic network is maintained by the Swiss Federal Office for the Environment and can be downloaded from \url{https://data.geo.admin.ch}.


\begin{thebibliography}{10}

\bibitem{alesina1997}
Alberto Alesina and Enrico Spolaore.
\newblock On the number and size of nations.
\newblock {\em The Quarterly Journal of Economics}, 112(4):1027--1056, 1997.

\bibitem{alesina2000}
Alberto Alesina, Enrico Spolaore, and Romain Wacziarg.
\newblock Economic integration and political disintegration.
\newblock {\em American Economic Review}, 90(5):1276--1296, December 2000.

\bibitem{allen2022}
Treb Allen.
\newblock The topography of nations.
\newblock September 2022.

\bibitem{allen2014}
Treb Allen and Costas Arkolakis.
\newblock Trade and the topography of the spatial economy.
\newblock {\em The Quarterly Journal of Economics}, 129(3):1085--1140, 2014.

\bibitem{allen2022b}
Treb Allen, Costas Arkolakis, and Xiangliang Li.
\newblock On the equilibrium properties of network models with heterogenous
  agents.
\newblock June 2022.

\bibitem{allen2020}
Treb Allen, Costas Arkolakis, and Yuta Takahashi.
\newblock Universal gravity.
\newblock {\em Journal of Political Economy}, 128(2):393--433, 2020.

\bibitem{berman1994}
Abraham Berman and Robert~J. Plemmons.
\newblock {\em Nonnegative {{Matrices}} in the {{Mathematical Sciences}}}.
\newblock Classics in {{Applied Mathematics}}. {Society for Industrial and
  Applied Mathematics}, January 1994.

\bibitem{berry2013}
Steven Berry, Amit Gandhi, and Philip Haile.
\newblock Connected {{Substitutes}} and {{Invertibility}} of {{Demand}}.
\newblock {\em Econometrica}, 81(5):2087--2111, 2013.

\bibitem{boots1980}
B.~N. Boots.
\newblock Weighting {{Thiessen Polygons}}.
\newblock {\em Economic Geography}, 56(3):248--259, July 1980.

\bibitem{deberg2008}
Mark De~Berg, Otfried Cheong, arc Kreveld, and Mark Overmars.
\newblock {\em Computational Geometry: Algorithms and Applications}.
\newblock {Springer}, second edition, 2008.

\bibitem{fetter1924}
Frank~A. Fetter.
\newblock The {{Economic Law}} of {{Market Areas}}.
\newblock {\em The Quarterly Journal of Economics}, 38(3):520--529, May 1924.

\bibitem{fujita1995}
Masahisa Fujita and Paul Krugman.
\newblock When is the economy monocentric?: Von {{Th\"unen}} and {{Chamberlin}}
  unified.
\newblock {\em Regional Science and Urban Economics}, 25(4):505--528, August
  1995.

\bibitem{fujita2002}
Masahisa Fujita, Paul Krugman, and Anthony~J. Venables.
\newblock {\em The {{Spatial Economy}}: {{Cities}}, {{Regions}}, and
  {{International Trade}}}.
\newblock {MIT Press}, 2002.

\bibitem{fujita1982}
Masahisa Fujita and Hideaki Ogawa.
\newblock Multiple equilibria and structural transition of non-monocentric
  urban configurations.
\newblock {\em Regional Science and Urban Economics}, 12(2):161--196, 1982.

\bibitem{galor2016}
Oded Galor and Ömer Özak.
\newblock The agricultural origins of time preference.
\newblock {\em American Economic Review}, 106(10):3064--3103, October 2016.

\bibitem{geiss2013}
Darius Gei{\ss}, Rolf Klein, Rainer Penninger, and G{\"u}nter Rote.
\newblock Optimally solving a transportation problem using {{Voronoi}}
  diagrams.
\newblock {\em Computational Geometry}, 46(8):1009--1016, October 2013.

\bibitem{hanjoul1989}
Pierre Hanjoul, Hubert Beguin, and Jean-Claude Thill.
\newblock Advances in the theory of market areas.
\newblock {\em Geographical analysis}, 21(3):185--196, 1989.

\bibitem{hebert1972}
Robert~F. Hebert.
\newblock A {{Note}} on the {{Historical Development}} of the {{Economic Law}}
  of {{Market Areas}}.
\newblock {\em The Quarterly Journal of Economics}, 86(4):563--571, November
  1972.

\bibitem{henderson1974}
J.~V. Henderson.
\newblock The sizes and types of cities.
\newblock {\em The American Economic Review}, 64(4):640--656, 1974.

\bibitem{henderson2005}
J.~Vernon Henderson and Hyoung~Gun Wang.
\newblock {Aspects of the rural-urban transformation of countries}.
\newblock {\em Journal of Economic Geography}, 5(1):23--42, 01 2005.

\bibitem{henrot2006}
Antoine Henrot and Michel Pierre.
\newblock {\em Variation et Optimisation de Formes: Une Analyse
  G\'eom\'etrique}, volume~48.
\newblock {Springer Science \& Business Media}, 2006.

\bibitem{hyson1950}
C.~D. Hyson and W.~P. Hyson.
\newblock The {{Economic Law}} of {{Market Areas}}.
\newblock {\em The Quarterly Journal of Economics}, 64(2):319--327, May 1950.

\bibitem{klenke2013}
Achim Klenke.
\newblock {\em Probability Theory: A Comprehensive Course}.
\newblock {Springer Science \& Business Media}, 2013.

\bibitem{krugman1993}
Paul Krugman.
\newblock First {{Nature}}, {{Second Nature}}, and {{Metropolitan Location}}.
\newblock {\em Journal of Regional Science}, 33(2):129--144, 1993.

\bibitem{lambert2019}
Nicolas~S. Lambert.
\newblock Elicitation and evaluation of statistical forecasts.
\newblock June 2022.

\bibitem{mascolell1995}
Andreu {Mas-Colell}, Michael~D. Whinston, and Jerry~R. Green.
\newblock Microeconomic {{Theory}}.
\newblock 1995.

\bibitem{merlo2016}
Antonio Merlo and Áureo~de Paula.
\newblock {Identification and Estimation of Preference Distributions When
  Voters Are Ideological}.
\newblock {\em The Review of Economic Studies}, 84(3):1238--1263, 09 2016.

\bibitem{nagy2020a}
D{\`a}vid~Kristzi{\`a}n Nagy.
\newblock Hinterlands, city formation and growth: {{Evidence}} from the
  {{U}}.{{S}}: Westward expansion.
\newblock {\em Review of Economic Studies}, forthcoming.

\bibitem{nagy2020}
Dávid~Krisztián Nagy.
\newblock Trade and urbanization: Evidence from hungary.
\newblock {\em American Economic Journal: Microeconomics}, 14(3):733--90,
  August 2022.

\bibitem{okabe2000}
Atsuyuki Okabe, Barry Boots, Kokichi Sugihara, and Sung~Nok Chiu.
\newblock {\em Spatial {{Tessellations}}: {{Concepts}} and {{Applications}} of
  {{Voronoi Diagrams}}}.
\newblock {John Wiley \& Sons, Ltd}, second edition, 2000.

\bibitem{parr1995}
John~B. Parr.
\newblock The {{Economic Law}} of {{Market Areas}}: {{A Further Discussion}}*.
\newblock {\em Journal of Regional Science}, 35(4):599--615, 1995.

\bibitem{redding2016}
Stephen~J. Redding.
\newblock Goods trade, factor mobility and welfare.
\newblock {\em Journal of International Economics}, 101:148--167, 2016.

\bibitem{redding2017}
Steven Redding and Esteban {Rossi-Hansberg}.
\newblock Quantitative {{Spatial Economics}}.
\newblock {\em Annual Review of Economics}, 9:21--58, 2017.

\bibitem{rossi-hansberg2020}
Esteban {Rossi-Hansberg}, Nicholas Trachter, Ezra Oberfeld, and Pierre-Daniel
  Sarte.
\newblock Plants in {{Space}}.
\newblock {\em Working paper}, page~66, 2020.

\bibitem{shieh1985}
Yeung-nan Shieh.
\newblock K. {{H}}. {{Rau}} and the {{Economic Law}} of {{Market Areas}}.
\newblock {\em Journal of Regional Science}, 25(2):191--199, 1985.

\end{thebibliography}
\end{document}